\documentclass[12pt,draftcls,onecolumn]{IEEEtran}
\ifCLASSINFOpdf
\usepackage[pdftex]{graphicx}
\graphicspath{{../pdf/}{../jpeg/}}
\DeclareGraphicsExtensions{.pdf,.jpeg,.png}
\else
\usepackage[dvips]{graphicx}
\graphicspath{{../eps/}}
\DeclareGraphicsExtensions{.eps}
\fi\usepackage{graphicx}

\usepackage{graphics}
\usepackage{epsfig}
\usepackage{epstopdf}
\usepackage{stfloats}
\usepackage[cmex10]{amsmath}
\usepackage{algorithmic}
\usepackage{array}
\usepackage{mdwmath}
\usepackage{mdwtab}
\usepackage{graphicx}
\usepackage{subfigure}
\usepackage{color}
\usepackage{amsfonts,amssymb}
\usepackage{hyperref} 
\usepackage{multirow}
\usepackage{diagbox}
\usepackage{caption}
\usepackage{ulem} 

\usepackage{amsfonts,amsthm,array} 
\usepackage{bbm}

\newtheorem{theorem}{Theorem}

\newtheorem{corollary}{Corollary}

\newtheorem{remark}{Remark}

\begin{document}
	
\title{On Secure NOMA-Aided Semi-Grant-Free Systems\thanks{Manuscript received.}}
	
\author{Hongjiang~Lei, 
		Fangtao~Yang,
		Hongwu~Liu, 
        Imran~Shafique~Ansari,\\
		Kyeong~Jin~Kim, 
        and~Theodoros~A.~Tsiftsis
}
\maketitle
\begin{abstract}
Semi-grant-free (SGF) transmission scheme enables grant-free (GF) users to utilize resource blocks allocated for grant-based (GB) users while maintaining the quality of service of GB users. This work investigates the secrecy performance of non-orthogonal multiple access (NOMA)-aided SGF systems. First, analytical expressions for the exact and asymptotic secrecy outage probability (SOP) of NOMA-aided SGF systems with a single GF user are derived. Then, the SGF systems with multiple GF users and the best-user scheduling scheme is considered. By utilizing order statistics theory, analytical expressions for the exact and asymptotic SOP are derived. {Monte Carlo simulation results are provided and compared with two benchmark schemes. The effects of system parameters on the SOP of the considered system are demonstrated and the accuracy of the developed analytical results is verified.} The results indicate that both the outage target rate for GB and the secure target rate for GF are the main factors of the secrecy performance of SGF systems.
\end{abstract}
	
\begin{IEEEkeywords}
		Non-orthogonal multiple access (NOMA),
		semi-grant-free (SGF) transmission scheme,
		grant-free (GF) user,
		grant-based (GB) user,
		secrecy outage probability.
\end{IEEEkeywords}
	
\section{Introduction}
\label{sec:introduction}
\subsection{Background and Related Work}
	
Ultra-reliable low latency communications (URLLC) and massive machine-type communications (mMTC) are the two most important scenarios for the next internet of things (IoT). URLLC focuses on mission-critical applications wherein unprecedented levels of reliability and latency are of the utmost importance in the fifth generation and it's beyond \cite{JiangW2021OJCS}.
In contrast, mMTC aspires to connect a vast number of intelligent devices to the Internet.
{
The user initiates the traditional grant-based (GB) access scheme with an access request to the base station (BS) in long term evolution.
The BS responds by allocating an access grant through a four-step handshake procedure strategy.
Once the BS grants the access request, data packets can be successfully transmitted without collision under ideal channel conditions.
However, GB scheme does not suit these scenarios due to high latency and heavy signaling overhead \cite{Shirvanimoghaddam2019Mag, Tanab2021NET}.
Moreover, the initial request transmission is still subject to collision and could require multiple transmissions depending on traffic load and the available resources at the BS.}

{To tackle these issues, grant-free (GF) transmission schemes were introduced in \cite{GF3GPP, MaZ2019IoT}, in which multiple users may occupy the same resource without the initial access request procedure.
Unlike the GB principle, no dedicated request transmission for granting access and allocating resource blocks is required for GF communications before starting a data transmission.
Although the GF scheme makes it possible to allow users to choose resource blocks independently and transmit data directly to reduce signaling overhead and latency effectively, collisions will become severe when multiple users select the same resource block to transmit simultaneously \cite{ZhangJ2020IoT}.
The collision issue can be resolved using massive multiple-input multiple-output (MIMO) or non-orthogonal multiple access (NOMA) technologies.
The former solution utilizes spatial degrees of freedom to mitigate multi-user collisions, while the latter focuses on spectrum sharing among multiple users with successive interference cancellation (SIC) \cite{Ding2017JSAC, ShahabMB2020Servey, AbbasR2019TCOM, ChenJ2022JSAC}.}

{Even though the massive connectivity can be supported through GF schemes, GB schemes are still desired, especially when strict quality of service (QoS) requirements exist \cite{ChenJ2022JSAC}.
The GB and GF transmission scheme must coexist in scenarios where URLLC applications are served by the GB scheme and mMTC applications in the same system are served by the GF scheme.
For example, a new hybrid access scheme was proposed in \cite{YangK2019Mag} to meet the various requirements of IoT networks wherein machine-type users with small data packets and delay-tolerant traffic utilized the GF scheme, and some users with large data packets and delay-sensitive traffic used the GB scheme.
NOMA-aided Semi-GF (SGF) transmission scheme was first explicitly introduced in \cite{DingZ2019TC} to alleviate the collisions and obtain massive connectivity.
A single GB user with multiple GF users to perform NOMA and two contention control mechanisms were proposed to suppress the interference on the GB user from the GF users.
Closed-form expressions for the outage probability (OP) of GF users were derived and the impact of different SIC decoding orders was investigated.
Their results demonstrated the superior performance of NOMA-aided SGF schemes.}
Based on the relationship between the GB user's targeted rate and channel conditions, an adaptive power allocation strategy was proposed to control the transmit power of GB users to ensure that the GB user's signals are always decoded in the second stage of SIC \cite{YangZ2020WCL}.
In \cite{ZhangC2020WCL}, the authors investigated the performance of an uplink SGF system with multiple uniformly distributed GF and GB users considered, in which the GF user whose received power is lower at the BS than that of the GB user was selected to pair with the connected GB user. Closed-form expressions for GB and GF users' exact and approximate ergodic rates were derived.
Further, the authors in \cite{ZhangC2021TWC} studied the effect of random locations of GF users on the performance of NOMA-assisted SGF systems by utilizing stochastic geometry. A dynamic threshold protocol was proposed to reduce the interference to GB users, and the outage performance was analyzed and compared with the open-loop protocol.
	
Relative to the SGF schemes proposed in \cite{DingZ2019TC}, a new QoS-guarantee scheme for NOMA-aided SGF systems was proposed in  \cite{DingZ2021TC} to ensure that the QoS of the GB user is the same as that when it solely occupies the channel.
Closed-form expressions were derived for the exact and asymptotic OP with {the best-user scheduling (BUS) scheme }and a hybrid SIC scheme. The results demonstrated that the proposed scheme could significantly improve the reliability of the GF users' transmissions.
Based on \cite{DingZ2021TC}, a new {adaptive} power control strategy was proposed to solve OP error floors entirely by adjusting the GF user's transmit power to change the decoding order of SIC in \cite{SunY2021TVT}.
{
	In \cite{LuH2021TWC}, the authors analyzed the outage performance of the NOMA-aided SGF systems with multiple randomly distributed GF users with fixed power and dynamic power schemes. 
	As discussed in \cite{LuH2021TWC}, the BUS scheme may lead to a fairness problem because the users closer to the base station may be scheduled more frequently due to weak path loss. To solve the fairness problem, a cumulative distribution function (CDF)-based user scheduling (CUS) scheme was proposed where the GF user with the maximal CDF value will be admitted to the channel. 
}
The analytical expressions for the OP with the CUS and BUS schemes were derived and the impacts of small-scale fading, path loss, and random user locations were jointly investigated.

%

{
	Recently, physical layer security for NOMA systems has attracted considerable attention \cite{LiuY2017TWC} - \cite{WangK2023TWC}.
	In \cite{LiuY2017TWC}, the authors investigated the secrecy performance of NOMA systems. Stochastic geometry was utilized to model the locations of legitimate and illegitimate receivers and the analytical expressions for the exact and asymptotic secrecy outage probability (SOP) for both single-antenna and multi-antenna scenarios were derived.
	In \cite{HeB2017JSAC}, the authors investigated the optimal decoding order, transmission rates, and power allocation in the design of NOMA systems. Two optimization problems were proposed and solved: the transmission power was minimized subject to the secrecy outage and QoS constraints and the minimum secrecy rate was maximized subject to the secrecy outage and transmit power constraints, respectively.
	Their results indicated that the optimal decoding order would not vary with the secrecy outage constraint in the considered problems and the power allocation ratio to the user must be increased as the secrecy constraint becomes more stringent.
	In \cite{LvL2019TIFS}, Lv \textit{et al.} proposed a new NOMA-inspired jamming and forwarding scheme to improve the security of cooperative communication systems and derived the analytical expressions for the lower bound of the ergodic secrecy sum rate (ESSR) and the asymptotic ESSR.
	Three relay selection schemes were proposed to enhance the secrecy performance of the multi-relay cooperative NOMA systems and the analytical expressions for the exact and asymptotic SOP were derived in \cite{LeiH2019TCOM}.
	In \cite{WangHM2020TCOM}, the authors proposed a novel downlink multi-user transmission scheme to meet the heterogeneous service requirements for the airborne NOMA systems consisting of security-sensitive users and QoS-sensitive users. 	The scenario where the QoS-sensitive users act as potential internal eavesdroppers were considered. The achievable secrecy rate was maximized through the joint optimization of user scheduling, power allocation, and trajectory design.
	In \cite{ZhaoN2020TVT}, two new schemes were proposed to enhance the security of airborne NOMA systems by the single user requiring and multiple users requiring security, respectively, and the effectiveness of the proposed schemes in ensuring secure transmissions were analyzed.
	In \cite{LeiH2020TCOM}, the relationship between the reliability and security of a two-user NOMA system was investigated. Considering different decoding capabilities at eavesdroppers and imperfect SIC, the analytical expressions of the SOP under the reliability outage probability constraint were derived.
	In \cite{WangK2023TWC}, the authors investigated the secrecy performance of a NOMA-based MEC system using the hybrid SIC decoding scheme.
	The latency was minimized by jointly optimizing the power allocation, task allocation, and computational resource allocation.
	A reinforcement learning-based and a matching-based algorithm were proposed to solve the optimization problems for the single-user and multi-user scenarios.
}

\subsection{Motivation and Contributions}

{
	Based on the authors' knowledge, there are two main differences between traditional NOMA and SGF schemes:
	1) In traditional NOMA systems, all the NOMA users can utilize the resource blocks, such as time slots or subcarriers.
	In NOMA-based SGF  systems, only the selected GF users based on scheduling schemes are allowed to opportunistically gain access to those resource blocks that GB users would exclusively occupy.
	2) For the conventional NOMA systems, the static SIC technology, either channel state information (CSI)-based SIC or QoS-based SIC, is utilized to cancel inter-user interference.
	The method in SGF systems to enhance spectral efficiency is through the hybrid (dynamic) SIC scheme.	
	For these reasons, although the secrecy performance of NOMA systems has been investigated in many works, the results are not applicable to NOMA-based SGF systems. This is the motivation for this work.
	Technically speaking, it is much more challenging to investigate the secrecy performance with a hybrid (dynamic) SIC scheme than that with a static SIC scheme.
}

We investigate a NOMA-aided SGF system with a single GF user, and then the results have been extended to SGF systems with multiple GF users.
The main contributions of this paper are summarized as follows.
\begin{enumerate}
	\item We analyze the secrecy performance of an uplink NOMA-aided SGF system with a single GF user as a benchmark. The analytical expression for the exact SOP of the GF user is derived. To obtain more insights, we derive asymptotic expressions for the SOP of the GF user in the high transmit signal-to-noise ratio (SNR) regime.	

    \item We further investigate the secrecy performance of NOMA-aided SGF systems with multiple GF users. The analytical expression for the exact and asymptotic SOP with {the BUS scheme} is developed based on order statistics to facilitate the performance analysis. {Monte Carlo simulation results are provided and compared with two different scheduling schemes. The effects of system parameters on the SOP of the considered system are demonstrated and the accuracy of the developed analytical results is verified.}

    \item In contrast to the metrics, such as OP and ergodic rate, derived in \cite{DingZ2019TC}-\cite{LuH2021TWC}, the secrecy performance of SGF systems is investigated in this work. Note that it is much more challenging to obtain the analytical expressions of the SOP relative to that of the OP for SGF systems, especially in the presence of multiple GF users.
\end{enumerate}
	
\subsection{Organization}
The rest of this paper is organized as follows. Section \ref{sec:SystemModel1} describes the considered system model. The SOP of the SGF systems with a single GF user and multiple GF users are analyzed in Sections \ref{sec:SOP} and \ref{sec:SystemModel2}, respectively. Section \ref{sec:RESULTS} presents the numerical and simulation results to demonstrate the analysis and the paper is concluded in Section \ref{sec:Conclusion}.
The notations utilized in this paper are summarized in Table I, which is shown at the top of this page.
\begin{table}[!t]
	\caption{\textit{List of Notations}}
	\begin{center}
		\begin{tabular}{|c| c | }
			\hline
			\textbf{Notation}   	& \textbf{Description}\\
			\hline
			$K$                 		& Number of the GF users \\
			\hline
			${N}$              			&{The number of antenna on $E$}\\
			\hline
			$g_B (g_F)$             & Channel coefficient between $U_B (U_F)$ and  $S$ \\
			\hline
			$g_k$                 	& Channel coefficient between $k$-th $U_F$ and  $S$ \\
			\hline
			${{\left| {{H_E}} \right|}^2}$  & Channel gains between $U_k$ and  $E$  \\
			\hline
			{
				$g_{{k_i}}$ }            & {Channel coefficient between $k$-th GF user and $i$-th receive antenna at $E$}\\
			\hline
			{
				$r_F (r_B)$}             & {The distance from $U_k$ ($U_B$) to  $S$  }\\
			\hline
			{
				$r_E$  }                 & {The distance from $U_k$  to   $E$  }\\
			\hline
			{
				$\alpha$ }               & {Path loss exponent } \\
			\hline
			$R_B$                	& Target rate of $U_B$\\
			\hline
			$R_{th}$             	& Secrecy target rate of $U_F$\\
			\hline
			${\sigma ^2}$
			& The noise power\\
			\hline
			$P_B(P_F)$
			& Transmit power of $U_B(U_F)$\\
			\hline
			$\rho_B(\rho_F)$
			& Transmit SNR of $U_B(U_F)$\\
			\hline
			${f_X}\left( \cdot \right)$
			& Probability density function of $X$\\
			\hline
			${F_X}\left(\cdot \right)$
			& Cumulative distribution function of $X$\\
			\hline
		\end{tabular}
	\end{center}
	\label{table1}
\end{table}

\section{System Model}
\label{sec:SystemModel1}

\subsection{NOMA-aided Semi-GF Systems}

\begin{figure}[!t]
	\centering
	\includegraphics[width = 2.45in]{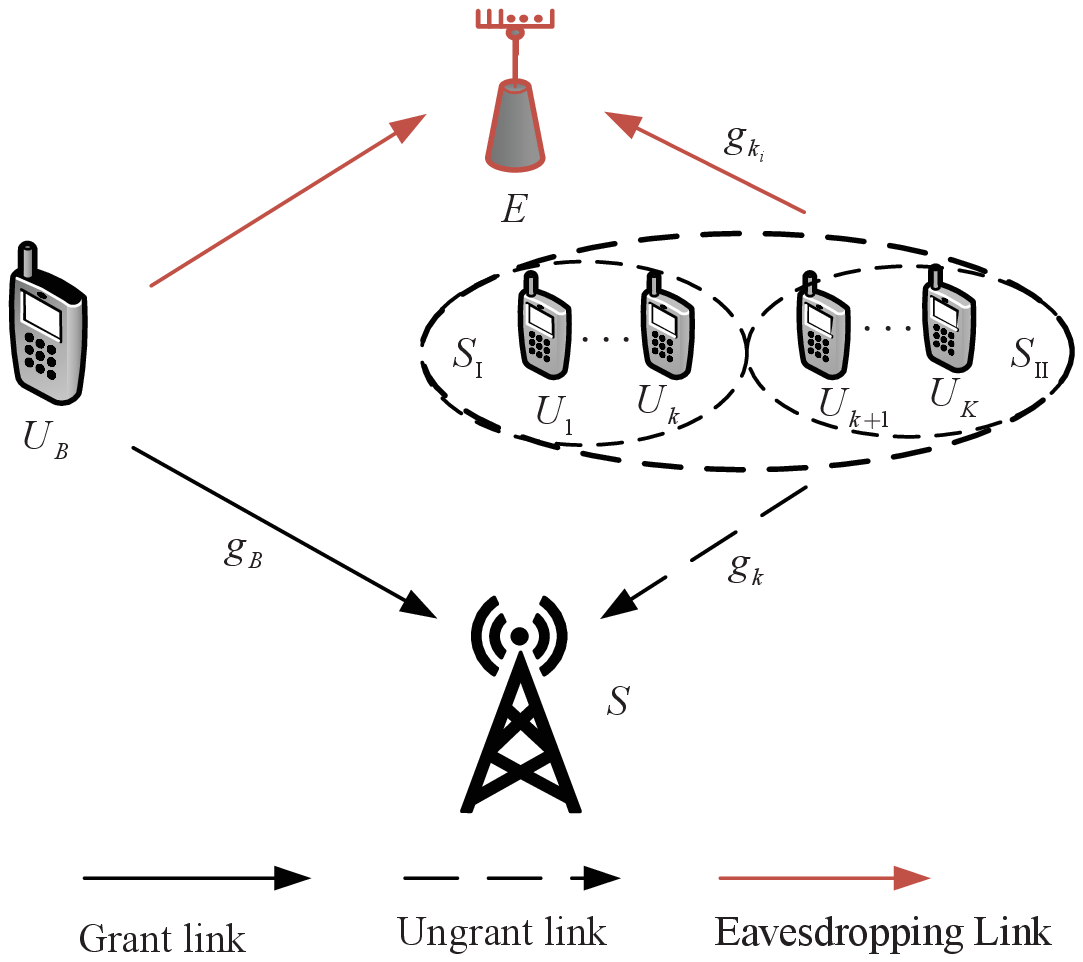}
	\caption{System model consisting of a BS ($S$), a GB user ($U_B$),  $K$ GF users ($U_k$), and an eavesdropper ($E$) {with $N$ antennas. The other  nodes are equipped with single antenna}.}
	\label{figmodel2}
\end{figure}
Consider an uplink SGF system illustrated in Fig. \ref{figmodel2}, a GB user ($U_B$) transmits signals to the BS ($S$), and the channel is re-used by $K$ GF user ($U_k, k=1,\cdots, K$) in SGF mode.
In other words, $U_k$ is allowed to utilize the resource block that would be solely occupied by $U_B$ employing NOMA technology while $U_B$'s QoS experience is the same as when it occupies the channel alone.
All the GF users are assumed to transmit signals with the same power ${\rho_F}$ and the channel gains are ordered as ${\left| {{h_1}} \right|^2} \le  \cdots  \le {\left| {{h_K}} \right|^2}$, where
{${\left| {{h_1}} \right|^2} = \mathop {\min }\limits_{1 \le k \le K} \left( {\frac{{{{\left| {{g_k}} \right|}^2}}}{{r_k^\alpha }}} \right)$
and
${\left| {{h_K}} \right|^2} = \mathop {\max }\limits_{1 \le k \le K} \left( {\frac{{{{\left| {{g_k}} \right|}^2}}}{{r_k^\alpha }}} \right)$}
where
$g_k$ denotes $U_k$'s channel coefficient,
{
	$r_k$ denotes the distance between $U_k$ and $S$, and 
	$\alpha $ signifies the path loss exponent.}
All the channels are assumed to undergo an independent identically and quasi-static Rayleigh fading model.
{To facilitate performance analysis, it is assumed that all the GF users are located in a small size cluster, such that
	the distances between $U_k$ and $S$ are same $\left( {{r_k} = {r_F}} \right)$.
}

The received signal at $S$ is expressed as ${y_B} = \sqrt {{P_B}} {h_B}{x_B} + \sqrt {{P_F}} {h_k}{x_F} + n$,
where
$P_i$ ($i \in \left\{ {B,F} \right\}$) denotes the transmit power,
{${{\left| {{h_B}} \right|^2}} = \frac{{{{\left| {{g_B}} \right|}^2}}}{{r_B^\alpha }}$,
	$g_B$ denotes $U_B$'s channel coefficient,
	$r_B$ denotes the distance between $U_B$ and $S$, }
$x_i$ is the signals from $U_i$ with unit power, i.e., $\mathbb{E}\left[ {{{\left| {{x_i}} \right|}^2}} \right] = 1$,
and $n$ is the additive white Gaussian noise (AWGN) with zero mean and variance ${\sigma ^2}$.

In this work, {the BUS scheme}
is considered, which means the GF user achieving the maximum data rate is scheduled to transmit signals \cite{DingZ2021TC, LuH2021TWC}.
The admission procedure consists of the following steps \cite{LuH2021TWC}:	
1) The $S$ sends pilot signals,
2) Each user estimates its own channel state information (CSI),
3) $U_B$ feedbacks its transmit SNR, target rate, and CSI to  $S$,
4) The $S$ calculates $U_B$'s decoding threshold and broadcasts $U_B$'s effective received SNR and decoding threshold to all GF users,
5) Each GF user calculates its transmit data rate, and
6) Each GF user sets its back-off time, which is a strictly decreasing function of the user's data rate. Then the GF user with the maximal data rate will be admitted to transmitting through distributed contention control protocol \cite{DingZ2019TC}.

To ensure the $U_B$'s QoS, there must have ${\log }_2 \left( {1 + \frac{{{\rho_B}{{\left| {{h_B}} \right|}^2}}}{{1 + \tau \left( {{{\left| {{h_B}} \right|}^2}} \right)}}} \right) \ge {R_B}$,
{where ${\rho _B} = \frac{{{P_B}}}{{{\sigma ^2}}}$,}
${R_B}$ denotes the target data of $U_B$ and $\tau \left( {{{\left| {{h_B}} \right|}^2}} \right) = \max \left\{ {0,{\tau _{B}}} \right\}$ denotes the maximum interference power tolerated when $U_B$'s signal is decoded during the first stage of SIC \cite{DingZ2021TC}
\footnote{
{
	The availability of perfect CSI is crucial in deciding the decoding order and the implementation of hybrid SIC.
The imperfect power gain caused by imperfect CSI could lead to an inappropriate SIC decoding order being selected and SIC decoding failures occurring \cite{DingZ2020CL2}.
}

},
$ {\tau _{B}} = \frac{{{{\left| {{h_B}} \right|}^2}}}{{{\alpha _B}}} - 1$,
${\alpha _B} = \frac{{\varepsilon _B}}{{{\rho_B}}}$,
${\varepsilon _B} = {\theta _B} - 1$,
and
${\theta _B}= {2^{{R_B}}} $.

$S$ first broadcasts $\tau \left( {{{\left| {{h_B}} \right|}^2}} \right)$ before scheduling.
By comparing {their received power of GF's signals on $S$} to $\tau \left( {{{\left| {{h_B}} \right|}^2}} \right)$, all the GF users are divided into two groups (${{\cal S}_{\mathrm{I}}}$ and ${{\cal S}_{\mathrm{II}}}$).
\begin{itemize}	
	
\item For ${U_k} \in {{\cal S}_{\mathrm{I}}}$ $\left( {1 \le k \le \left| {{{\cal S}_{\mathrm{I}}}} \right| \le K} \right)$, they experience ${\rho_F}{\left| {{h_k}} \right|^2} > \tau \left( {{{\left| {{h_B}} \right|}^2}} \right)$ {with ${\rho _F} = \frac{{{P_F}}}{{{\sigma ^2}}}$}, which will lead to ${\log }_2 \left( {1 + \frac{{{\rho_B}{{\left| {{h_B}} \right|}^2}}}{{1 + \tau \left( {{{\left| {{h_B}} \right|}^2}} \right)}}} \right) < {R_B}$. This signifies that ${U_k}$'s signals must be decoded before decoding $U_B$'s signals to guarantee that $U_B$'s QoS experience is the same as when it occupies the channel alone
\footnote{
	Since the signals from GF users were decoded before decoding those from the GB user in this case, additional latency for GB users will occur. Thus, the GF scheme in the NOMA-aided SGF systems are suitable for such applications with more stringent QoS than latency requirements. In other words, the NOMA-aided SGF systems aim to make a channel reserved by a GB user that can be shared by GF users, improving connectivity and spectral efficiency through collaboration between GF transmission and conventional GB schemes.}.	
{
	Then, the achievable rate of $U_B$ and $U_k$ are expressed as ${R_B^{\mathrm{I}}} = {\log _2}\left( {1 + {\rho _B}{{\left| {{h_B}} \right|}^2}} \right)$ and ${R_k^{\mathrm{I}}} = {\log }_2 \left( {1 + \frac{{{\rho_F}{{\left| {{h_k}} \right|}^2}}}{{1 + {\rho_B}{{\left| {{h_B}} \right|}^2}}}} \right)$, respectively.
}

\item For those GF users in ${U_k} \in {{\cal S}_{\mathrm{II}}}$ $\left( {1 \le k \le \left| {{{\cal S}_{\mathrm{II}}}} \right| \le K} \right)$, they experience ${\rho_F}{\left| {{h_k}} \right|^2} < \tau \left( {{{\left| {{h_B}} \right|}^2}} \right)$, which will lead to ${\log }_2 \left( {1 + \frac{{{\rho_B}{{\left| {{h_B}} \right|}^2}}}{{1 + \tau \left( {{{\left| {{h_B}} \right|}^2}} \right)}}} \right) > {R_B}$.
This signifies that the GF user's signal in this group will be decoded at either the first or the second stage of SIC. Accordingly, ${U_k}$ will achieve a data of ${R_k^{\mathrm{I}}} = {\log }_2 \left( {1 + \frac{{{\rho_F}{{\left| {{h_k}} \right|}^2}}}{{1 + {\rho_B}{{\left| {{h_B}} \right|}^2}}}} \right)$ or ${R_k^{{\mathrm{II}}}} = {\log }_2 \left( {1 + {\rho_F}{{\left| {{h_k}} \right|}^2}} \right)$. Due to ${R_k^{\mathrm{II}}} > {R_k^{\mathrm{I}}}$, to achieve the maximum data rate at the GF user, $U_B$'s signal must be decoded during the first stage of SIC \cite{DingZ2021TC, LuH2021TWC}.
{
	Thus, the achievable rate of $U_B$ and $U_k$ are expressed as
	${R_B^{\mathrm{II}}} = {\log _2}\left( {1 + \frac{{{\rho _B}{{\left| {{h_B}} \right|}^2}}}{{1 + {\rho _F}{{\left| {{h_k}} \right|}^2}}}} \right)$ and
	$R_k^{{\mathrm{II}}} = {\log _2}\left( {1 + {\rho _F}{{\left| {{h_k}} \right|}^2}} \right)$, respectively.
}
\end{itemize}
Then, the achievable rate of ${U_k}$ ($1 \le k \le K - 1$) is expressed as
\begin{small}
\begin{equation}
	{R_k} 
	= \left\{ {\begin{array}{*{20}{c}}
			{R_K^{\mathrm{I}},}&{\left| {{{\cal S}_{\mathrm{II}}}} \right| = 0,}\\
			{R_K^{\mathrm{II}},}&{\left| {{{\cal S}_{\mathrm{II}}}} \right| = K,}\\
			{\max \left\{ {R_K^{\mathrm{I}},R_k^{\mathrm{II}}} \right\},}&{\left| {{{\cal S}_{\mathrm{II}}}} \right| = k.}
	\end{array}} \right.
	\label{ratek}
\end{equation}
\end{small}

{{It must be noted that only one GF user is selected to access the channel.
The grouping stated before is logically grouped for analysis of the achievable rate of the selected GF user.
Specifically, the signals from the users in different groups have different decode orders at the base station.}}

\begin{remark}
	It must be noted the SGF scheme only guarantees that admitting the GF user is transparent to the GB user whose QoS experience is the same as when it occupies the channel alone. In other words, the SGF scheme does not always guarantee no outage for the GB user. Further, the outage of the GB user in this case $\left( {{{\left| {{h_B}} \right|}^2} < \alpha _B} \right)$ does not signify outage of the GF user.
\end{remark}
\begin{remark}
	$\tau \left( {{{\left| {{h_B}} \right|}^2}} \right) = \max \left\{ {0,{\tau _{B}}} \right\}$ denotes the maximum interference power tolerated when $U_B$'s signal is decoded during the first stage of SIC. Based on the definition of ${\tau _{B}}$, it can be observed that $\alpha _B$ is the threshold when $U_B$ occupies the channel alone. Specifically, { due to $ {\tau _B} = \frac{{{{\left| {{h_B}} \right|}^2}}}{{{\alpha _B}}} - 1 < 0 \Leftrightarrow {\left| {{h_B}} \right|^2} < {\alpha _B}$, ${\alpha _B}$ signifies the reliability threshold when $U_B$ occupies the channel alone.  ${\left| {{h_B}} \right|^2} < {\alpha _B}$ denotes reliability outage occurs on $U_B$ due to the weakness of the GB link} and {${\tau _B} > 0 \Leftrightarrow {\left| {{h_B}} \right|^2} > {\alpha _B}$} denotes the channels can be shared with $U_F$ under SGF scheme.		
\end{remark}	

In this work, we consider the worst-case security scenario wherein $E$ is equipped with $N$ antennas using maximal ratio combining (MRC) scheme to fully decode the users' information \footnote{
	{In this case, it is assumed that the eavesdropper has a powerful multi-user detection capability (e.g. parallel interference cancellation) so that the received data stream can be distinguished and the interference generated by the superimposed signals can be subtracted \cite{Verdu1998Book}. As stated in \cite{LiuY2017TWC, LeiH2020TCOM}, this case is the worst-case scenario where the decoding capability of the eavesdropper has been overestimated, which makes the analysis and design robust for the practical scenario and is sensible and desirable from a security perspective.}}.
Then, the eavesdropping rate is expressed as ${R_E} = {\log _2}\left( {1 + {\rho_F}{{\left| {{H_E}} \right|}^2} } \right)$, where
$ {{{\left| {{H_{E}}} \right|}^2}} \buildrel \Delta \over = \sum\limits_{i = 1}^N {{{\left| {{h_{{k_i}}}} \right|}^2}}$,
{
${\left| {{h_{{k_i}}}} \right|^2} = \frac{{{{\left| {{g_{{k_i}}}} \right|}^2}}}{{r_{E}^\alpha }}$,
${{\left| {{g_{{k_i}}}} \right|}^2}$ denotes channel coefficient between $k$-th GF user and $i$-th receive antenna at $E$, and $r_{E}$ denotes the distance between the GF users and $E$.
}

\subsection{Statistical Properties of Channel Power Gains}

This subsection provides the statistical law of channel power gains, laying the performance analysis foundation.
The probability density function (PDF)  of $ {{{\left| {{H_{E}}} \right|}^2}} $ is  expressed as
{
${f_{{H_E}}}(x) = \frac{{r_E^{N\alpha }}}{{\Gamma \left( N \right)}}{x^{N - 1}}{e^{ - r_E^{N\alpha }x}}$,
}
where
$\Gamma(z)=\int_{0}^{\infty} e^{-t} t^{z-1} dt$ is the  Gamma function as defined by \cite[(8.310.1)]{Gradshteyn2007Book}.
The CDF of ${\left| {{h_K}} \right|^2} $ is expressed as
{
${F_{{{\left| {{h_K}} \right|}^2}}}\left( x \right) = \sum\limits_{i = 0}^K {{\varphi _i}} {e^{ - ir_F^\alpha x}}$,
}
where
${\varphi _i} = \left( {_{\,\,i}^K} \right){\left( { - 1} \right)^i}$,
and $\left( {_{\,\, i}^K} \right) = \frac{{K!}}{{i!\left( {K - i} \right)!}}$.

The joint PDF of ${\left| {{h_i}} \right|^2}$ and ${\left| {{h_j}} \right|^2}$$\left( {1 \leqslant i < j \leqslant K} \right)$ is expressed as \cite{DingZ2021TC}
\begin{equation}
	{f_{{{\left| {{h_i}} \right|}^2},{{\left| {{h_j}} \right|}^2}}}\left( {x,y} \right)
	= \sum\limits_{n = 0}^{j - i - 1} {\sum\limits_{m = 0}^{i - 1} {{\phi _1}{e^{ - {\phi _2}x - {\phi _3}y}}} },
	\label{pdfh1k}
\end{equation}
where
{
${\phi _1} = \frac{{K!{{\left( { - 1} \right)}^{m + n}}\left( {_{\quad n}^{j - i - 1}} \right)\left( {_{\,m}^{i - 1}} \right)r_F^{2\alpha }}}{{\left( {i - 1} \right)!\left( {K - j} \right)!\left( {j - i - 1} \right)!}}$,
${\phi _2} =r_F^\alpha \left( {m + j - i - n} \right)$,
}
and
{
${\phi _3} = r_F^\alpha \left( {K - j + n + 1} \right)$.
}
Then, the joint CDF of ${\left| {{h_i}} \right|^2}$ and ${\left| {{h_j}} \right|^2}$ $\left( {1 \leqslant i < j \leqslant K} \right)$ is obtained as
	\begin{equation}
		\begin{aligned}
			{F_{{{\left| {{h_i}} \right|}^2},{{\left| {{h_j}} \right|}^2}}}\left( {x,y} \right) 
			&= \sum\limits_{n = 0}^{j - i - 1} {\sum\limits_{m = 0}^{i - 1} {\frac{{{\phi _1}}}{{{\phi _3}}}} } \left( {\frac{{{e^{ - \left( {{\phi _2} + {\phi _3}} \right)x}}}}{{{\phi _2} + {\phi _3}}} + \frac{{{\phi _3}{e^{ - \left( {{\phi _2} + {\phi _3}} \right)y}}}}{{{\phi _2}\left( {{\phi _2} + {\phi _3}} \right)}} - \frac{{{e^{ - {\phi _2}x - {\phi _3}y}}}}{{{\phi _2}}}} \right).
		\end{aligned}
		\label{iandj}
	\end{equation}
When $i = 1, j = K$, we obtain
{
\begin{equation}
	{f_{{{\left| {{h_1}} \right|}^2},{{\left| {{h_K}} \right|}^2}}}\left( {x,y} \right) = \sum\limits_{n = 0}^{K - 2} {{\mu _0}{e^{ - r_F^\alpha \left( {K - n - 1} \right)x}}{e^{ - r_F^\alpha \left( {n + 1} \right)y}}} ,
	\label{pdfh1K}
\end{equation}
}
and
{
\begin{equation}
	\begin{aligned}
		{F_{{{\left| {{h_1}} \right|}^2},{ {\left| {{h_K}} \right|}^2}}}\left( {x,y} \right) 
		& = \sum\limits_{n = 0}^{K - 2} {{\mu _1}{e^{ - Kr_F^\alpha x}} + {\mu _2}{e^{ - Kr_F^\alpha y}} - {\mu _3}{e^{ - \left( {K - n - 1} \right)r_F^\alpha x}}{e^{ - \left( {n + 1} \right)r_F^\alpha y}}},
		\label{cdfh1K}
	\end{aligned}
\end{equation}
}
respectively,
where
{
${\mu _0} = {\rm{ }}\frac{{K!{{\left( { - 1} \right)}^n}\left( {_{\;n}^{K - 2}} \right)r_F^{2\alpha }}}{{\left( {K - 2} \right)!}}$,
${\mu _1} = \frac{{{\mu _0}}}{{r_F^{2\alpha }K\left( {n + 1} \right)}}$,
${\mu _2} = \frac{{{\mu _0}}}{{r_F^{2\alpha }K\left( {K - n - 1} \right)}}$,
${\mu _3} = \frac{{{\mu _0}}}{{r_F^{2\alpha }\left( {n + 1} \right)\left( {K - n - 1} \right)}}$.
}

The joint PDF and CDF of ${{{\left| {{h_k}} \right|}^2}}$, ${{{\left| {{h_{k + 1}}} \right|}^2}}$ $\left( {1 \le k \le K - 2} \right)$, and ${{{\left| {{h_K}} \right|}^2}}$ is given as \cite{DingZ2021TC}
\begin{equation}
	\begin{aligned}
		{f_{{{\left| {{h_k}} \right|}^2},{{\left| {{h_{k + 1}}} \right|}^2},{{\left| {{h_K}} \right|}^2}}}\left( {x,y,z} \right)
		= \sum\limits_{n = 0}^{K - k - 2} {\sum\limits_{m = 0}^{k - 1} {{\varsigma _0}{e^{ - {A_0}x}}{e^{ - {B_0}y}}{e^{ - {C_0}z}}} },
		\label{pdfthree}
	\end{aligned}
\end{equation}
and
\begin{equation}
	\begin{aligned}
		{F_{{{\left| {{h_k}} \right|}^2},{{\left| {{h_{k + 1}}} \right|}^2},{{\left| {{h_K}} \right|}^2}}}\left( {x,y,z,w} \right)
		& = \sum\limits_{n = 0}^{K - k - 2} {\sum\limits_{m = 0}^{k - 1} {\sum\limits_{i = 1}^6 {{\varsigma _i}{e^{ - \left( {{A_i}x + {B_i}y + {C_i}z + {W_i}w} \right)}}} } } ,
		\label{cdfthree}
	\end{aligned}
\end{equation}
respectively,
where
{
${\varsigma _0} = \frac{{K!{{\left( { - 1} \right)}^{m + n}}\left( {_n^{K - k - 2}} \right)\left( {_m^{k - 1}} \right)r_F^{3\alpha }}}{{\left( {K - k - 2} \right)!\left( {k - 1} \right)!}}$,	${A_0} = r_F^\alpha \left( {m + 1} \right)$,
${B_0} = r_F^\alpha \left( {K - k - n - 1} \right)$,
${C_0} = r_F^\alpha \left( {n + 1} \right)$,
}
${W_0} = {B_0} + {C_0}$,
${\varsigma _1} =  - {\varsigma _2} = \frac{{\varsigma _0} }{{{A_0}{B_0}{C_0}}}$,
${\varsigma _3} =  - {\varsigma _4} =  - \frac{{\varsigma _0} }{{{A_0}{B_0}{W_0}}}$,
${\varsigma _5} =  - {\varsigma _6} =  - \frac{{\varsigma _0} }{{{A_0}{C_0}{W_0}}}$,
${A_1} = {A_3} = {A_5} = 0$,
${A_2} = {A_4} = {A_6} = {A_0}$,
${B_1} = {B_3} = {B_5} = {A_0}$,
${B_2} = {B_4} = {B_6} = 0$,
${C_1} = {C_2} = {B_0}$,
${C_3} = {C_4} = 0$,
${C_5} = {C_6} = {W_0}$,
${W_1} = {W_2} = {C_0}$,
${W_3} = {W_4} = {W_0}$,
and
${W_5} = {W_6} = 0$.

For $k = K - 1$, we have ${\left| {{h_{k + 1}}} \right|^2} = {\left| {{h_K}} \right|^2}$, the joint PDF and CDF of ${\left| {{h_{K - 1}}} \right|^2}$ and ${\left| {{h_K}} \right|^2}$ are expressed as
{
\begin{equation}
	{f_{{{\left| {{h_{K - 1}}} \right|}^2},{{\left| {{h_K}} \right|}^2}}}\left( {x,y} \right) = \sum\limits_{n = 0}^{K - 2} {{\mu _0}{e^{ - {C_0}x}}{e^{ - r_F^\alpha y}}},
	\label{pdftwo}
\end{equation}
}
and
\begin{equation}
	\begin{aligned}
		{F_{{{\left| {{h_{K - 1}}} \right|}^2},{{\left| {{h_K}} \right|}^2}}}\left( {x,y,z,w} \right)
		& = \sum\limits_{n = 0}^{k - 2} {\sum\limits_{j = 1}^4 {\frac{{{\mu_0}}}{{{C_0}}}{{\left( { - 1} \right)}^{j + 1}}{e^{ - \left( {{a_j}x + {b_j}y + {c_j}z + {q_j}w} \right)}}} },
		\label{cdftwo}
	\end{aligned}
\end{equation}
respectively,
where
${a_1} = {a_4} = 0,{a_2} = {a_3} = {C_0},{b_1} = {b_4} = {C_0},{b_2} = {b_3} = 0,{c_1} = {c_2} = 0$,
${{c_3} = {c_4} = r_F^\alpha}$,
{${q_1} = {q_2} = r_F^\alpha $},
and
${q_3} = {q_4} = 0$.
	
\section{Secrecy Outage Probability Analysis with A Single Grant-Free User}
\label{sec:SOP}

In this section, the secrecy performance of the SGF systems with a single GF user is investigated to pay the road to the performance analysis of SGF systems with multiple GF users.
When $K = 1$, there is no need to consider scheduling.
{It must be noted that this scenario can also be viewed as the multiple-GF-user SGF systems using a random user scheduling (RUS) scheme.}
The achievable rate of ${U_F}$ in (\ref{ratek}) is rewritten as
\begin{equation}
	{R_F} = \left\{ {\begin{array}{*{20}{c}}
			{R_F^{\mathrm{I}},}&{ {{\rho_F}{\left| {{h_F}} \right|^2} > \tau \left( {{{\left| {{h_B}} \right|}^2}} \right)},}\\
			{R_F^{\mathrm{II}},}&{ {{\rho_F}{\left| {{h_F}} \right|^2} < \tau \left( {{{\left| {{h_B}} \right|}^2}} \right)},}\\
	\end{array}} \right.
	\label{rateF}
\end{equation}
where
$R_F^{\mathrm{I}} = {\log _2}\left( {1 + \frac{{{\rho_F}{{\left| {{h_F}} \right|}^2}}}{{1 + {\rho_B}{{\left| {{h_B}} \right|}^2}}}} \right)$
and
$R_F^{\mathrm{II}} = {\log _2}\left( {1 + {\rho_F}{{\left| {{h_F}} \right|}^2}} \right)$,
which denote the achievable rate at $U_F$ in scenarios when ${U_F}$'s signal is decoded at the first and second stages of the SIC, respectively.
It must be noted that when there is an outage on $U_B$, the ${U_F}$' signals must be decoded at the first stage of the SIC.

{The user $U_j$'s achievable secrecy rate is expressed as $R_{s, j}^i = {\left[ {R_j^i - {R_E}} \right]^ + }$\cite{Bloch2008TIT},
where
$j \in \left\{ {F,B} \right\}$,
$i \in \left\{ {{\mathrm{I}},{\mathrm{II}}} \right\}$ and ${\left[ x \right]^ + } = \max \left\{ {x,0} \right\}$.
SOP denotes the probability that the maximum achievable secrecy rate is less than a target secrecy rate \cite{Bloch2008TIT}.
Based on (\ref{rateF}), the SOP for $U_F$ is given as
\begin{equation}
		\begin{aligned}	
			{P_{out, F}} = \underbrace {\Pr \left\{ {R_{s, F}^{\mathrm{I}} < {R_{th}},{\rho_F}{{\left| {{h_F}} \right|}^2} > \tau \left( {{{\left| {{h_B}} \right|}^2}} \right)} \right\}}_{P_{out}^{\mathrm{I}}}
			+ \underbrace {\Pr \left\{ {R_{s, F}^{\mathrm{II}} < {R_{th}},{\rho_F}{{\left| {{h_F}} \right|}^2} < \tau \left( {{{\left| {{h_B}} \right|}^2}} \right)} \right\}}_{P_{out}^{\mathrm{II}}},
		\label{SOPGF}
		\end{aligned}	
\end{equation}
where}
${{R_{th}}}$ represents the secrecy threshold rate,
${P_{out}^{\mathrm{I}}}$ denotes ${U_F}$'s signal is decoded at the first stage,
and
${P_{out}^{\mathrm{II}}}$ denotes ${U_F}$'s signal is decoded at the second stage.

{Similarly, the SOP for $U_B$ is expressed as
\begin{equation}
		{P_{out,B}} = \Pr \left\{ {R_{s,B}^{\mathrm{I}} < {R_{th}},{\rho _F}{{\left| {{h_F}} \right|}^2} > \tau \left( {{{\left| {{h_B}} \right|}^2}} \right)} \right\} + \Pr \left\{ {R_{s,B}^{{\mathrm{II}}} < {R_{th}},{\rho _F}{{\left| {{h_F}} \right|}^2} < \tau \left( {{{\left| {{h_B}} \right|}^2}} \right)} \right\}.
		\label{SOPGB}
\end{equation}}

{It can be observed that the analysis of the secrecy outage probability of the GB user is similar to that of the GF user, expressed in Eq. (\ref{SOPGF}).
	Due to space limitations, the analysis of the $U_B$'s secrecy outage probability is regrettably omitted here.
	In this work, the SOP of the NOMA-aided SGF system is equivalent to the SOP of the GF user, unless stated otherwise. }

\begin{remark}
{
	It must be noted that $\rho_B$ affects the SNR/SINR of $U_B$ and the maximum interference that $U_B$ can tolerate when $U_B$'s signal is decoded during the first stage of SIC simultaneously.
	In the lower-$\rho_B$ region, the signals from $U_F$ must be decoded in the first stage of SIC.
	With the increase of $\rho_B$, the interference to $U_F$ increases and the secrecy performance worsens.
	In the larger-$\rho_B$ region, the signals from $U_F$ will be decoded in the second stage of SIC. There is no interference from $U_B$ to $U_F$. Then, the SOP decreases to a constant.
	Thus, there is a worst $\rho_B$ for the security of $U_F$.
}
\end{remark}

\begin{remark}
{
	In contrast, $\rho_F$ affects the SNR/SINR of $U_F$ and $E$ simultaneously.
	In the lower-$\rho_F$ region, the signals from $U_B$ will be decoded in the first stage of SIC.
	For a small $\rho_F$, there is $\Pr \left\{ {{\rho_F}{{\left| {{h_F}} \right|}^2} < \tau \left( {{{\left| {{h_B}} \right|}^2}} \right)} \right\} > \Pr \left\{ {{\rho_F}{{\left| {{h_F}} \right|}^2} > \tau \left( {{{\left| {{h_B}} \right|}^2}} \right)} \right\}$. Thus, ${P_{out}^{\mathrm{II}}}$ is the main part of ${P_{out}}$ in the lower-$\rho_F$ region while ${P_{out}^{\mathrm{I}}}$ is the main part of ${P_{out}}$ in the larger-$\rho_F$ region.
	Based on the results in \cite{LeiH2017CL}, increasing $\rho_F$ will enhance the security of $U_F$ in the lower-$\rho_F$ region.
	In the larger-$\rho_F$ region, the signals from $U_F$ will be decoded in the first stage of SIC.
	Although both the SINR of $U_F$ and SNR of $E$ improve with increasing $\rho_F$, the SINR of $U_F$ improves slower than the SNR of $E$, so the security of the SGF system deteriorates. Thus, there is an optimal $\rho_F$ to minimize the SOP of $U_F$.
}
\end{remark}

\begin{remark}
{
Furthermore, the effect from $r_B$ on $P_{out}^{\mathrm{I}}\left( {P_{out}^{{\mathrm{II}}}} \right)$ is the opposite of the effect from $\rho_B$,
while the effect of $\rho_B$ and $r_B$ on the secrecy performance of the SGF systems are similar.
$r_F$ only affects the SINR/SNR of $U_F$. Larger $r_F$ denotes stronger path loss on $U_F$ thereby higher SOP.
$r_E$ only affects the SNR of $E$ where larger $r_E$ denotes stronger path loss on $E$ and hence lower SOP.
}
\end{remark}

The following theorem provides an exact expression for the SOP achieved applicable to the considered SGF scheme.
\begin{theorem}
The SOP of ${U_F}$ is expressed as
\begin{equation}
	\begin{aligned}	
		{P_{out}} = \left\{ {\begin{array}{*{20}{c}}
				{P_{out}^{\mathrm{I},1} + P_{out}^{\mathrm{I},21} + P_{out}^{\mathrm{II}},}&{{\varepsilon _B}{\varepsilon _{th}} < 1,}\\
				{P_{out}^{\mathrm{I},1} + P_{out}^{\mathrm{I},22} + P_{out}^{\mathrm{II}},} &{{\varepsilon _B}{\varepsilon _{th}} > 1,}
		\end{array}} \right.
		\label{Theorem1}		
	\end{aligned}
\end{equation}
where
{
$P_{out}^{\mathrm{I},1} =1 - {e^{ - r_B^\alpha {\alpha _B}}} - \frac{{r_B^\alpha r_E^{N\alpha }{e^{ - r_F^\alpha {\alpha _{th}}}}{\omega _1}\left( {{\lambda _1},{\lambda _2},{\lambda _3}} \right)}}{{\Gamma \left( N \right)}}$,
$P_{out}^{{\mathrm{I}},{\mathrm{21}}} =\frac{{{e^{ - r_B^\alpha {\alpha _B}}}r_B^\alpha \Gamma \left( {N,r_E^\alpha {\alpha _1}} \right)}}{{{\varepsilon _2}\Gamma \left( N \right)}} + {e^{\frac{{r_F^\alpha }}{{{P_F}}}}}r_B^\alpha r_E^{N\alpha }\frac{{{\omega _3}\left( {0,{\varepsilon _2},r_E^\alpha } \right)}}{{\Gamma \left( N \right)}}\\ - {e^{ - r_F^\alpha {\alpha _{th}}}}r_B^\alpha r_E^{N\alpha }\frac{{{\omega _2}\left( {{\lambda _1},{\lambda _2},{\lambda _3}} \right) + {\omega _3}\left( {{\lambda _1},{\lambda _2},{\lambda _3}} \right)}}{{\Gamma \left( N \right)}}$,
$P_{out}^{{\mathrm{I}},{\mathrm{22}}} = r_B^\alpha \frac{{{e^{ - r_B^\alpha {\alpha _B}}}}}{{{\varepsilon _2}}} - r_B^\alpha r_E^{N\alpha }\frac{{{e^{ - r_F^\alpha {\alpha _{th}}}}{\omega _4}\left( {{\lambda _1},{\lambda _2},{\lambda _3}} \right)}}{{\Gamma \left( N \right)}}$,
$P_{out}^{{\mathrm{II}}} = \frac{{r_F^\alpha {e^{ - r_B^\alpha {\alpha _B}}}}}{{r_B^\alpha {P_{\rm{F}}}{\alpha _B} + r_F^\alpha }}\\ - \frac{{r_E^{N\alpha }r_F^\alpha {e^{ - \left( {r_F^\alpha {\alpha _{th}} + r_B^\alpha {\varepsilon _1}} \right)}}}}{{\left( {r_F^\alpha {\rho _F}{\alpha _B} + r_B^\alpha } \right){{\left( {r_B^\alpha {\rho _F}{\varepsilon _1} + {\lambda _3}} \right)}^N}}}$,
}
${\alpha _{th}} = \frac{{\varepsilon _{th}}}{{{\rho_F}}}$,
${\varepsilon _{th}} = {\theta _{th}} - 1$,
${\theta _{th}}= {2^{{R_{th}}}} $,
{
${\lambda _1} = r_F^\alpha {\rho _B}{\theta _{th}}$,
${\lambda _2} = r_F^\alpha {\rho _B}{\alpha _{th}} + r_B^\alpha $,
${\lambda _3} = r_F^\alpha {\theta _{th}} + r_E^\alpha $,
${\alpha _1} = \frac{{1 - {\varepsilon _B}{\varepsilon _{th}}}}{{{\rho_F}{\theta _{th}}{\varepsilon _B}}}$,
${\varepsilon _1} = {\alpha _B}{\theta _{th}}$,
${\varepsilon _2} = \frac{{r_F^\alpha }}{{{P_F}{\alpha _B}}} + r_B^\alpha $,
}
$\Gamma \left( { \cdot , \cdot } \right)$ is the upper incomplete Gamma function, as defined by \cite[(8.350.2)]{Gradshteyn2007Book},
${\omega _1}\left( {a,b,c} \right) = \frac{{{b^{N - 1}}\Gamma \left( N \right)}}{{{a^N}}}{e^{\frac{{bc}}{a}}}$ $\times \left( {\Gamma \left( {1 - N,\frac{{bc}}{a}} \right) - \Gamma \left( {1 - N,b{\alpha _B} + \frac{{bc}}{a}} \right)} \right)$,
${\omega _2}\left( {a,b,c} \right) = \frac{{{b^{N - 1}}\Gamma \left( N \right)}}{{{a^N}}}{e^{\frac{{bc}}{a}}}  \Gamma \left( {1 - N,b{\alpha _B} + \frac{{bc}}{a}} \right) - {e^{ - b{\alpha _B}}}\Delta $,
${\Delta} = \frac{{\pi {\alpha _1}}}{{2R}}\sum\limits_{r = 1}^R {\frac{{\sqrt {1 - \ell _r^2} }}{{a{\hbar _r} + b}}\hbar _r^{N - 1}{e^{ - \left( {a{\alpha _B} + c} \right){\hbar _r}}}}$,
${\omega _3}\left( {a,b,c} \right) = \frac{{{b^{N - 1}}\Gamma \left( N \right)}}{{{a^N}}}{e^{\frac{{bc}}{a}}}\Gamma \left( {1 - N,b{\alpha _B} + \frac{{bc}}{a}} \right) - {\omega _2}\left( {a,b,c} \right) - {e^{\frac{b}{{{P_B}}} - a{\alpha _3}}}\frac{{\pi {\alpha _1}}}{{2L}}\sum\limits_{l = 1}^L {\frac{{\sqrt {1 - \vartheta _l^2} }}{{a{v_l} + b}}v_l^{N - 1} \times }
{e^{\left( {\frac{a}{{{P_B}}} - c} \right){v_l} - \frac{{{\alpha _3}\left( {a{\alpha _1} + b} \right)}}{{{\alpha _1} - {v_l}}}}}$,
${\omega _4}\left( {a,b,c} \right) = \frac{{{b^{N - 1}}\Gamma \left( N \right)}}{{{a^N}}}{e^{\frac{{bc}}{a}}}\Gamma \left( {1 - N,b{\alpha _B} + \frac{{bc}}{a}} \right)$,
$R$ and $L$ is the summation item, which reflects accuracy vs. complexity,
${\ell _r} = \cos \left( {\frac{{2r - 1}}{{2R}}\pi } \right)$, ${\hbar _r} = \frac{{{\alpha _1}}}{2}\left( {{\ell _r} + 1} \right)$,
${\vartheta _l} = \cos \left( {\frac{{2l - 1}}{{2L}}\pi } \right)$,
and
${v_l} = \frac{{{\alpha _1}}}{2}\left( {{\vartheta _l} + 1} \right)$.
\end{theorem}

\begin{proof}
	See Appendix \ref{appendicesA}.
\end{proof}

\begin{remark} 
	Based on (\ref{poutI01}), one can observe that $\Pr \left\{ {{\rho_F}{{\left| {{h_F}} \right|}^2} > 0} \right\} = 1$, which is independent of $\rho_F$. With the help of the result in \cite{LeiH2017CL}, secrecy capacity improves with increasing transmit SNR then gradually tends to a constant. So, $P_{out}^{\mathrm{I},1}$ decreases and gradually tends to a constant for a given $\alpha _B$.
	Furthermore, $\Pr \left\{ {{{\left| {{h_F}} \right|}^2} > \frac{{{\tau _B}}}{{{\rho_F}}}} \right\}$ increases gradually tending to 1 with increasing $\rho_F$. Thus, for a given $\alpha _B$, $P_{out}^{\mathrm{I},2}$ increases with increasing $\rho_F$ until gradually tending to a constant and independent of $\rho_F$.	
\end{remark}

\begin{remark}
	Based on (\ref{ppoutI02}), it must be noted that the relationship between ${{\omega _0}\left( {{{\left| {{h_B}} \right|}^2},{{\left| {{H_E}} \right|}^2}} \right)}$ and ${\frac{{{\tau _B}}}{{{\rho_F}}}}$ act as the constraint for the GF link. More specifically, the former is constraint on security while the latter is constraint on decoding order. The relationship between constraint on security and on decoding order directly affects the SOP of $U_F$.
\end{remark}

%
\begin{remark}
	The analysis in (\ref{relationship1}) demonstrates that SOP of $U_F$ depends on the relationship between ${\varepsilon _B}{\varepsilon _{th}} $ and 1, which determines the relationship between the constraint on decoding order and the constraint on security. When ${\varepsilon _B}{\varepsilon _{th}} > 1$, the constraint on decoding order is always less than that on security. 			
	{$\epsilon_B $$\epsilon_{th} < 1$ means that $R_{th}$ needs to be small for a given $R_B$, which is a generalized condition in practice since SGF is invoked to encourage spectrum sharing between a GB user and a GF user with a low secrecy threshold data rate. However, for $\epsilon_B $$\epsilon_{th} > 1$, it also offers secrecy outage performance achieved by the SGF scheme will be worse.}
\end{remark}

The analytical expression provided in (\ref{Theorem1}) is complicated because many factors affect the secrecy performance of the GF user, specifically, the decoding order, the target data rate of $U_B$, the target secrecy rate of $U_F$, and the quality of the eavesdropping channel.
We derive asymptotic expressions of the SOP in the high transmit SNR regime to obtain more insights.	

\begin{corollary}
When ${\rho_B} \to \infty $, the asymptotic SOP of ${U_F}$ is approximated as
\begin{equation}
	{
	P_{out}^{{\rho_B} \to \infty } \approx P_{out}^{{\rm{II,}}{\rho _B} \to \infty } = 1 - {e^{ - r_F^\alpha {\alpha _{th}}}}{\left( {1 + {{\left( {\frac{{{r_F}}}{{{r_E}}}} \right)}^\alpha }{\theta _{th}}} \right)^{ - N}}.
	\label{poutA1PB}
}
\end{equation}
\begin{proof}
	See Appendix \ref{appendices1}.
\end{proof}
\end{corollary}

\begin{remark}
	The increasing $\rho_B$ leads to larger $\tau \left( {{{\left| {{h_B}} \right|}^2}} \right)$, which means it is easy to guarantee the QoS of $U_B$. Then, the probability of decoding the $U_F$'s signals in the second stage of SIC increases. The final result is ${P_{out}} \approx \Pr \left\{ {R_s^{\mathrm{II}} < {R_{th}}} \right\}$ which simply depends on $\rho_F$, $R_{th}$, {$r_F$, $r_E$, and $N$.}
\end{remark}

\begin{corollary}
When $ {\rho_F} \to \infty $, the asymptotic SOP of ${U_F}$ is approximated as
\begin{equation}
	{
	P_{out}^{{\rho _F} \to \infty } \approx P_{out}^{{\rm{I}},{\rho _F} \to \infty } = 1 - {\left( {\frac{{{r_B}}}{{{r_F}}}} \right)^{N\alpha }}{\left( {\frac{{r_E^\alpha }}{{{\rho _B}{\theta _{th}}}}} \right)^N}\Gamma \left( {1 - N,\frac{{r_B^\alpha }}{{{\rho _B}{\theta _{th}}}}\left( {{\theta _{th}} + {{\left( {\frac{{{r_E}}}{{{r_F}}}} \right)}^\alpha }} \right)} \right).
	\label{poutA2PF}
}
\end{equation}
\begin{proof}
	See Appendix \ref{corollary2}.
\end{proof}

\end{corollary}

\begin{remark}
	The increasing $\rho_F$ leads to $\Pr \left\{ {{\rho_F}{{\left| {{h_F}} \right|}^2} < \tau \left( {{{\left| {{h_B}} \right|}^2}} \right)} \right\} \to 0$, which leads to $P_{{\mathrm{out}}}^{{\mathrm{II}}} \to 0$. Then, the probability of decoding the $U_F$'s signals in the first stage of SIC increases. The final result is ${P_{out}} \approx \Pr \left\{ {R_s^{\mathrm{I}} < {R_{th}}} \right\}$, which depends on $\rho_B$, $R_{th}$, {$r_B$, $r_E$, and $N$.}
\end{remark}

\begin{corollary}
When ${\rho_B} = {\rho_F} \to \infty $, the asymptotic SOP of ${U_F}$ is approximated as
\begin{equation}
	{
   \begin{aligned}
	P_{out}^{ \infty } &= P_{out}^{{\mathrm{I}}, \infty} + P_{out}^{{\mathrm{II}}, \infty} = 1 - {\left( {1 + {{\left( {\frac{{{r_B}}}{{{r_F}}}} \right)}^\alpha }{\varepsilon _B}} \right)^{ - 1}}{\left( {1 + {\theta _{th}}{{\left( {\frac{{{r_F}}}{{{r_E}}}} \right)}^\alpha } + {\varepsilon _B}{\theta _{th}}{{\left( {\frac{{{r_B}}}{{{r_E}}}} \right)}^\alpha }} \right)^{ - N}}.
	\label{poutA3PBF}
	\end{aligned}
}
\end{equation}

\begin{proof}
	See Appendix \ref{corollary3}.
\end{proof}

\end{corollary}

\begin{remark}
	{In this scenarios with ${\rho_B} = {\rho_F} \to \infty$, it must be noted that there is $\Pr \left\{ {{\rho_F}{{\left| {{h_F}} \right|}^2} < \tau \left( {{{\left| {{h_B}} \right|}^2}} \right)} \right\} = \Pr \left\{ {{{\left| {{h_F}} \right|}^2} < \frac{{{{\left| {{h_B}} \right|}^2}}}{{{\varepsilon _B}}}} \right\}$.
	The decoding order depends on the relationship between $\frac{{{{\left| {{g_F}} \right|}^2}}}{{r_F^\alpha }}$ and $\frac{{{{\left| {{h_B}} \right|}^2}}}{{{\varepsilon _B}}} = \frac{{{{\left| {{g_B}} \right|}^2}}}{{r_B^\alpha \left( {{2^{{R_B}}} - 1} \right)}}$.
	Then, we have $P_{out}^{\mathrm{I}} = \Pr \left\{ {R_s^{\mathrm{I}} < {R_{th}},{{\left| {{h_F}} \right|}^2} > \frac{{{{\left| {{h_B}} \right|}^2}}}{{{\varepsilon _B}}}} \right\}$ and $P_{out}^{{\mathrm{II}}} = \Pr \left\{ {R_s^{{\mathrm{II}}} < {R_{th}},{{\left| {{h_F}} \right|}^2} < \frac{{{{\left| {{h_B}} \right|}^2}}}{{{\varepsilon _B}}}} \right\}$, which are constants independent of ${\rho_B}$ and ${\rho_F} $ depends on $R_B$, $R_{th}$, $r_B$, $r_F$, and $r_E$.}
\end{remark}

\section{Secrecy Outage Probability Analysis with Multiple Grant-Free Users}
\label{sec:SystemModel2}

{In this section, the secrecy performance of the multiple-GF-user SGF systems with BUS scheme is investigated.}

\subsection{Secrecy Outage Probability Analysis}
When $K > 1$, both user scheduling and decoding order issues should be considered simultaneously.
It should be noted that ${\left| {{{\cal S}_{\mathrm{II}}}} \right| = K}$
denotes that the signals from GF users should be decoded on the secondary stage of SIC to maximize the achievable rate.
Then $U_K$ is selected to transmit signals.
The same for  ${\left| {{{\cal S}_{\mathrm{II}}}} \right| = 0}$.
Based on (\ref{ratek}), the SOP of $U_k$ is given by
\begin{small}
\begin{equation}
	\begin{aligned}
		{P_{out}} &= \underbrace {\Pr \left\{ {R_K^{\mathrm{I}} - {R_E} < {R_{th}},\left| {{S_{{\mathrm{II}}}}} \right| = 0} \right\}}_{ \buildrel \Delta \over = {P_{out,1}}} + \underbrace {\Pr \left\{ {R_K^{{\mathrm{II}}} - {R_E} < {R_{th}},\left| {{S_{{\mathrm{II}}}}} \right| = K} \right\}}_{ \buildrel \Delta \over = {P_{out,2}}}\\
		&+ \underbrace {\sum\limits_{k = 1}^{K - 1} {\Pr \left\{ {\max \left\{ {R_K^{\mathrm{I}},R_k^{{\mathrm{II}}}} \right\} - {R_E} < {R_{th}},\left| {{S_{{\mathrm{II}}}}} \right| = k} \right\}} }_{ \buildrel \Delta \over = {P_{out,3}}},
	\end{aligned}
\label{SOP03}
\end{equation}
\end{small}
where
${{P_{out,1}}}$ denotes the SOP for $U_k$ when groups ${{\cal S}_{\mathrm{II}}}$ are empty,
${{P_{out,2}}}$ signifies the SOP for $U_k$ when groups ${{\cal S}_{\mathrm{I}}}$ are empty,
and
${{P_{out,3}}}$ denotes the SOP for $U_k$ when there are $k$ GF users in groups ${{\cal S}_{\mathrm{II}}}$.
In the first two terms, $U_K$ is always selected to transmit signals.
The following theorem provides the exact expression for the SOP of the considered SGF scheme with multiple GF users.
	
\begin{theorem}
	The SOP of ${U_F}$ is expressed as
	\begin{equation}
		\begin{aligned}	
			{P_{out}} = \left\{ {\begin{array}{*{20}{c}}
					 {P_{out,1}^1} + {P_{out,1}^{21}}+ {P_{out,2}}+{P_{out,3}},&{{\varepsilon _B}{\varepsilon _{th}} < 1,}\\
					 {P_{out,1}^1} + {P_{out,1}^{22}}+ {P_{out,2}}+{P_{out,3}},&{{\varepsilon _B}{\varepsilon _{th}} > 1,}
			\end{array}} \right.
			\label{SOP02}		
		\end{aligned}
	\end{equation}
where
{
${P_{out,1}^1}  = 1 - {e^{ - r_B^\alpha {\alpha _B}}} - \frac{{r_B^\alpha r_E^{N\alpha }{e^{ - r_F^\alpha {\alpha _{th}}}}{\omega _1}\left( {{\lambda _1},{\lambda _2},{\lambda _3}} \right)}}{{\Gamma \left( N \right)}}$,
$P_{out,1}^{21} = \frac{{r_B^\alpha r_E^{N\alpha }}}{{\Gamma \left( N \right)}}\sum\limits_{n = 0}^{K - 2} {\sum\limits_{i = 2}^{i = 3} {\left( {{\mu _1}{e^{\frac{{Kr_F^\alpha }}{{{\rho _F}}}}}{\omega _i}\left( {0,{\alpha _4},r_E^\alpha } \right)} \right.} } \\
	\left. { + {\mu _2}{e^{ - Kr_F^\alpha {\alpha _{th}}}}{\omega _i}\left( {{\eta _1},{\eta _2},{\eta _3}} \right) - {\mu _3}{e^{\frac{{Kr_F^\alpha  - {C_0}}}{{{\rho _F}}} - {C_0}{\alpha _{th}}}}{\omega _i}\left( {{\eta _4},{\eta _5},{\eta _6}} \right)} \right)$,
$P_{out,1}^{22} = \frac{{r_B^\alpha r_E^{N\alpha }}}{{\Gamma \left( N \right)}}\sum\limits_{n = 0}^{K - 2} {\left( {{\mu _1}{e^{\frac{{Kr_F^\alpha }}{{{\rho _F}}}}} \times } \right.} \\
	\left. {{\omega _4}\left( {0,{\alpha _4},r_E^\alpha } \right) + {\mu _2}{e^{ - Kr_F^\alpha {\alpha _{th}}}}{\omega _4}\left( {{\eta _1},{\eta _2},{\eta _3}} \right) - {\mu _3}{e^{\frac{{Kr_F^\alpha  - {C_0}}}{{{\rho _F}}} - {C_0}{\alpha _{th}}}}{\omega _4}\left( {{\eta _4},{\eta _5},{\eta _6}} \right)} \right)$,
${P_{out,2}} = \\
	 \sum\limits_{i = 0}^K {\left( {\frac{{{\varphi _i}}}{{ir_F^\alpha  + r_B^\alpha {\rho _F}{\alpha _B}}}\left( {\frac{{ir_F^\alpha r_E^{N\alpha }{e^{ - \left( {ir_F^\alpha {\alpha _{th}} + r_B^\alpha {\varepsilon _1}} \right)}}}}{{{{\left( {ir_F^\alpha {\theta _{th}} + r_B^\alpha {\rho _F}{\varepsilon _1} + r_E^\alpha } \right)}^N}}} + {\rho _F}{\alpha _B}r_B^\alpha {e^{ - r_B^\alpha {\alpha _B}}}} \right)} \right)}$,
${P_{out,3}} =
	r_B^\kappa r_E^{N\kappa }\sum\limits_{n = 0}^{K - k - 2} {\sum\limits_{m = 0}^{k - 1} {\left( {\sum\limits_{i = 1}^4 {\frac{{{\varsigma _i}}}{{\Gamma \left( N \right)}}  } } \right.} } \\
	\times\left. {\left( {{e^{ - {\xi _1}}}{\Delta _1} + {e^{ - {\xi _4}}}{\Delta _3}} \right) + \sum\limits_{i = 5}^6 {{\varsigma _i}\left( {{e^{ - {\xi _1}}}{\Delta _2} + {e^{ - {\xi _4}}}{\Delta _4}} \right)} } \right)
	+ r_B^\kappa r_E^{N\kappa }\sum\limits_{n = 0}^{K - 2} {\frac{{{\mu _0}}}{{{C_0}}}\left( {\sum\limits_{j = 1}^2 {\frac{{{{\left( { - 1} \right)}^{j + 1}}}}{{\Gamma \left( N \right)}}} \left( {{e^{ - {\zeta _1}}}{\Delta _5} + {e^{ - {\zeta _4}}}{\Delta _7}} \right)} \right.} \\
	\left. { + \sum\limits_{j = 3}^4 {{{\left( { - 1} \right)}^{j + 1}}\left( {{e^{ - {\zeta _1}}}{\Delta _6} + {e^{ - {\zeta _4}}}{\Delta _8}} \right)} } \right)$,
${\alpha _4} = \frac{{Kr_F^\alpha }}{{{\rho _F}{\alpha _B}}} + r_B^\alpha$,
${\eta _1} = Kr_F^\alpha {\rho _B}{\theta _{th}}$,
${\eta _2} = Kr_F^\alpha {\rho _B}{\alpha _{th}} + r_B^\alpha $,
${\eta _3} = Kr_F^\alpha {\theta _{th}} + r_E^\alpha$,
}
$ {\eta _4} = {C_0}{\rho _B}{\theta _{th}}$,
${\eta _5} = {C_0}{\rho _B}{\alpha _{th}} + \frac{{Kr_F^\alpha  - {C_0}}}{{{\rho _F}{\alpha _B}}} + r_B^\alpha $,
$ {\eta _6} = {C_0}{\theta _{th}} + r_E^\alpha $, 	
${\Delta _1}  = \frac{{\xi _3^{N - 1}\Gamma \left( {1 - N,{\xi _3}{\alpha _B} + \frac{{{\xi _2}{\xi _3}}}{{{u_1}}}} \right)}}{{u_1^N}}{e^{\frac{{{\xi _2}{\xi _3}}}{{{u_1}}}}} - \frac{{{e^{ - {\xi _3}{\varepsilon _1}}}{\omega _5}\left( {{u_1},{\xi _3},{v_1},{l_1}} \right)}}{{\Gamma \left( N \right)}}$,
${\Delta _2} = \frac{{{e^{ - {\xi _3}{\alpha _B}}}}}{{\xi _2^N{\xi _3}}} - \frac{{{e^{ - {\xi _3}{\varepsilon _1}}}}}{{{\xi _3}{{\left( {{\rho_F}{\xi _3}{\varepsilon _1} + {\xi _2}} \right)}^N}}}$,
${\xi _1} = {W_i}{\alpha _{th}} - \frac{{{B_i} + {C_i}}}{{\rho_F}}$,
{
${\xi _2} = {W_i}{\theta _{th}} + r_E^\alpha $,
${\xi _3} = {W_i}{\rho_B}{\alpha _{th}} + \frac{{{B_i} + {C_i}}}{{P_{F}{\alpha _B}}} + r_B^\alpha $,
}
${u_1} = {W_i}{\rho_B}{\theta _{th}}$,
${v_1} = {u_1}{\rho_F}{\varepsilon _1}$,
${l_1} = {u_1}{\varepsilon _1} + {\rho_F}{\xi _3}{\varepsilon _1} + {\xi _2}$,
${\xi _4} = \left( {{B_i} + {W_i}} \right){\alpha _{th}} - \frac{{{C_i}}}{{\rho_F}}$,
${\omega _5} \left( {a,b,c,f} \right)
= \frac{{{f^{N - 1}}}}{{b}}H_{1,0:1,1:0,1}^{1,0:1,1:1,0}\left[ {_{\quad  - }^{\left( {0;1,2} \right)}\left| {_{\left( {0,1} \right)}^{\left( {0,1} \right)}\left| {_{\left( {0,1} \right)}^ - \left| {\frac{a}{{bf}},\frac{c}{{{f^2}}}} \right.} \right.} \right.} \right]$,	
${\Delta _3} = \frac{{{e^{ - {\xi _6}{\varepsilon _1}}}{\omega _6}\left( {{u_1},{\xi _6},{v_2},{l_2}} \right)}}{{\Gamma \left( N \right)}}$,
${\Delta _4} = \frac{{{e^{ - {\xi _6}{\varepsilon _1}}}}}{{{\xi _6}{{\left( {{\rho_F}{\xi _6}{\varepsilon _1} + {\xi _5}} \right)}^N}}}$,
${v_2} = {u_1}{\rho_F}{\varepsilon _1}$,
${l_2} = {u_1}{\varepsilon _1} + {\xi _6}{\rho_F}{\varepsilon _1} + {\xi _5}$,
{
${\xi _5} = \left( {{B_i} + {W_i}} \right){\theta _{th}} + r_E^\alpha $,
${\xi _6} = {W_i}{\alpha _{th}}{\rho_B} + \frac{{{C_i}}}{{P_{F}{\alpha _B}}} + r_B^\alpha $,
}
${\zeta _1} = {q_j}{\alpha _{th}} - \frac{{{b_j} + {c_j}}}{{\rho_F}}$,
{
${\zeta _2} = {q_j}{\theta _{th}} + r_E^\alpha$,
${\zeta _3} = {q_j}{\alpha _{th}}{\rho_B} + \frac{{{b_j} + {c_j}}}{{P_{F}{\alpha _B}}} + r_B^\alpha$,
}
${\Delta _5} = \frac{{\zeta _3^{N - 1}}}{{u_2^N}}{e^{\frac{{{\zeta _2}{\zeta _3}}}{{{u_2}}}}}\Gamma \left( {1 - N,{\zeta _3}{\alpha _B} + \frac{{{\zeta _2}{\zeta _3}}}{{{u_2}}}} \right)
- \frac{{{e^{ - {\xi _3}{\varepsilon _1}}}{\omega _5}\left( {{u_2},{\zeta _3},{v_3},{l_3}} \right)}}{{\Gamma \left( N \right)}}$,
${u_2} = {q_j}{\rho_B}{\theta _{th}}$,
${v_3} = {u_2}{\rho_F}{\varepsilon _1}$,
${l_3} = {u_2}{\varepsilon _1} + {\zeta _3}{\rho_F}{\varepsilon _1} + {\zeta _2}$,
${\Delta _6} = \frac{{{e^{ - {\zeta _3}{\alpha _B}}}}}{{\zeta _2^N{\zeta _3}}} - \frac{{{e^{ - {\zeta _3}{\varepsilon _1}}}}}{{{\zeta _3}{{\left( {{P_F}{\zeta _3}{\varepsilon _1} + {\zeta _2}} \right)}^N}}}$,
${\zeta _4} = \left( {{b_j} + {q_j}} \right){\alpha _{th}} - \frac{{{c_i}}}{{\rho_F}}$,
{
${\zeta _5} = \left( {{b_j} + {q_j}} \right){\theta _{th}} + r_E^\alpha$,
${\zeta _6} = {q_i}{\alpha _{th}}{\rho_B} + \frac{{{c_j}}}{{P_{F}{\alpha _B}}} + r_B^\alpha$,
}
${\Delta_7} = \frac{{{e^{ - {\zeta _6}{\varepsilon _1}}}}{\omega _5}\left( {{u_2},{\zeta _6},{v_4},{l_4}} \right)}{{\Gamma \left( N \right)}}$,
${v_4} = {u_2}{\rho_F}{\varepsilon _1}$,
${l_4} = {u_2}{\varepsilon _1} +{\zeta _6} {\rho_F}{\varepsilon _1} + {\zeta _5}$,
and
${\Delta _8} =\frac{{{e^{ - {\zeta _6}{\varepsilon _1}}}}}{{{\zeta _6}{{\left( {{\rho_F}{\zeta _6}{\varepsilon _1} + {\zeta _5}} \right)}^N}}}$.

\end{theorem}

\begin{proof}
	See Appendix \ref{appendicesE}.
\end{proof}

\begin{remark}
	{
		Based on (\ref{pout4k2}), it can be observed that the number of the users in Groups I and II depends on the relationship between ${\left| {{h_k}} \right|^2}$ and $\frac{{{\tau _B}}}{{{\rho _F}}} = \frac{{{\rho _B}}}{{{\rho _F}}}\frac{{{{\left| {{h_B}} \right|}^2}}}{{{2^{{R_B}}} - 1}} - \frac{1}{{{\rho _F}}}$ related to $\rho_B$ and $\rho_F$. }
\end{remark}

\normalsize
Relative to SGF systems with a single GF user, the expression of SOP presented in \textbf{Theorem 2} is exceptionally complicated, and the main reason is that in addition to the factors highlighted in \textbf{Theorem 1}, the number of users in each group has a significant effect on the secrecy performance.

\subsection{Asymptotic Secrecy Outage Probability Analysis}

To obtain more insights, we derive asymptotic expressions of the SOP in the high transmit SNR regime.	
	
\textbf{Corollary 4} When  $ {\rho_B}={\rho_F}  \to \infty $, the SOP of $U_k$ is approximated at high SNR as
{
\begin{equation}
	\begin{aligned}
		{P_{out}^{ \infty }} &\approx 	P_{out,1}^{2,\infty } + P_{out,2}^\infty  + P_{out,3}^\infty,
       \label{sopmasyBF}
	\end{aligned}
\end{equation}	
where
$P_{out,1}^{2,\infty } \approx \sum\limits_{n = 0}^{K - 2} {\frac{{{\varepsilon _B}{\mu _1}}}{{K{{\left( {\frac{{{r_F}}}{{{r_B}}}} \right)}^\alpha } + {\varepsilon _B}}}}$,
$P_{out,2}^\infty  \approx \sum\limits_{i = 0}^K {\frac{{{\varphi _i}{\varepsilon _B}}}{{i{{\left( {\frac{{{r_F}}}{{{r_B}}}} \right)}^\alpha } + {\varepsilon _B}}}}  + \sum\limits_{i = 0}^K {\frac{{i{\varphi _i}{{\left( {i{\chi _1} + {\chi _2}} \right)}^{ - N}}}}{{i + {\varepsilon _B}{{\left( {\frac{{{r_B}}}{{{r_F}}}} \right)}^\alpha }}}}$,	
${\chi _1} = {\theta _{th}}{\left( {\frac{{{r_F}}}{{{r_E}}}} \right)^\alpha }$,
${\chi _2} = {\varepsilon _B}{\theta _{th}}{\left( {\frac{{{r_B}}}{{{r_E}}}} \right)^\alpha } + 1$,
$P_{out,3}^\infty  = \sum\limits_{k = 1}^{K - 2} {P_{out,3}^{k,\infty }}  + P_{out,3}^{K - 1,\infty } = \sum\limits_{k = 1}^{K - 2} {\left( {I_3^\infty  + I_4^\infty } \right)}  + P_{out,3}^{K - 1,\infty }$,
$I_3^\infty \approx \sum\limits_{n = 0}^{K - k - 2} {\sum\limits_{m = 0}^{k - 1} {\sum\limits_{i = 5}^6 {\frac{{{\varsigma _i}{\varepsilon _B}\left( {1 - {\chi _3}} \right)}}{{\left( {K + {\varpi _i}} \right){{\left( {\frac{{{r_F}}}{{{r_B}}}} \right)}^\alpha } + {\varepsilon _B}}}} } } $,
$I_4^\infty \approx \sum\limits_{n = 0}^{K - k - 2} {\sum\limits_{m = 0}^{k - 1} {\sum\limits_{i = 5}^6 {\frac{{{\varsigma _i}{\varepsilon _B}{\chi _3}}}{{\left( {K - k} \right){{\left( {\frac{{{r_F}}}{{{r_B}}}} \right)}^\alpha } + {\varepsilon _B}}}} } } $,
${\chi _3} = {\left( {\left( {K + {\varpi _i}} \right){\chi _1} + {\chi _2}} \right)^{ - N}}$,
$\varpi  = \left[ {0,0,n + 2,1,m + 1 - k, - k} \right]$,
$P_{out,3}^{K - 1,\infty } \approx \sum\limits_{n = 0}^{K - 2} {{\mu _4}\sum\limits_{j = 3}^4 {{{\left( { - 1} \right)}^{j + 1}}{\varepsilon _B}\left( {\frac{{1 - {\chi _4}}}{{{\varpi _j}r_B^{ - \alpha } + r_F^{ - \alpha }{\varepsilon _B}}}} \right.} {\mkern 1mu}  + \left. {\frac{{{\chi _4}}}{{r_B^{ - \alpha } + r_F^{ - \alpha }{\varepsilon _B}}}} \right)} $,
${\mu _4} = \frac{{K!{{\left( { - 1} \right)}^n}\left( {_{\;\;n}^{K - 2}} \right)}}{{\left( {K - 2} \right)!\left( {n + 1} \right)}}$,
and
${\chi _4} = {\left( {{\varpi _j}{\chi _1} + {\chi _2}} \right)^{ - N}}$.
}
\begin{proof}
	See Appendix \ref{appendicesF}.
\end{proof}	

\begin{remark}
	{
		For the SGF systems with multiple GF users, when ${\rho_B} = {\rho_F} \to \infty$, it must be noted that there is $\Pr \left\{ {{\rho_F}{{\left| {{h_k}} \right|}^2} < \tau \left( {{{\left| {{h_B}} \right|}^2}} \right)} \right\} = \Pr \left\{ {{{\left| {{h_k}} \right|}^2} < \frac{{{{\left| {{h_B}} \right|}^2}}}{{{\varepsilon _B}}}} \right\}$. 
		The number of the users in Groups I and II depends on the relationship between $\frac{{{{\left| {{g_k}} \right|}^2}}}{{r_F^\alpha }}$ and $\frac{{{{\left| {{h_B}} \right|}^2}}}{{{\varepsilon _B}}} = \frac{{{{\left| {{g_B}} \right|}^2}}}{{r_B^\alpha \left( {{2^{{R_B}}} - 1} \right)}}$ irrelated to $\rho_B$ and $\rho_F$. 
	}
\end{remark}	

\begin{remark}
	{
		Based on \textbf{\textit{Corollary 4}}, one can observe that when ${\rho_B} = {\rho_F} \to \infty$, the SOP of the SGF systems with multiple GF users is a constant, which depends on $R_B$, $R_{th}$, $r_B$, $r_F$, $r_E$, and $K$. Further, $P_{out,3}$ is the main part of the SOP.
	}
\end{remark}	

\begin{remark}
{
	Based on \textbf{Corollaries} 3 and 4, one can find that varying transmit power at $U_B$ and $U_F$ can only improve the secrecy performance of the SGF systems within a certain range.
	In contrast, improving the system's secrecy performance is more effective by varying the distance. Specifically, reducing the distance between the GF users and the base station as much as possible or making the GF users far from the eavesdroppers.
	All the parameters must be carefully chosen to maximize the secrecy performance of the considered SGF systems, such as the target rate of $U_B$, the secrecy threshold rate of $U_F$, and the distance between the transmitters and receivers.
}
\end{remark}

\section{Numerical Results and Discussions}
\label{sec:RESULTS}

This section presents Monte-Carlo simulations and numerical results to prove the secrecy performance analysis on the NOMA-aided SGF systems through varying parameters, such as transmit SNR, target data rate, and the number of antennas,
{
	etc.
	The main parameters are set as
	$R_{th}$ = 0.1, $R_B$ = 0.9, $N = 2$, $\alpha = 2.2$,  $r_B = r_F = r_E = 10 $ m,
	unless stated otherwise.
	In all the figures, ``Sim'', ``Ana'', and `` Asy'' denote the simulation, numerical results,  and asymptotic analysis respectively.
}
The results in all the figures demonstrate that the analytical results perfectly match the simulation results, verifying our analysis's accuracy.

\subsection{{SOP of the NOMA-aided SGF system with a single-GF-user}}

\begin{figure}[t]
	\centering
	\subfigure[SOP for varying $R_{th}$ and $R_B$.  ]{
		\label{fig03a}
		\includegraphics[width = 2.5in]{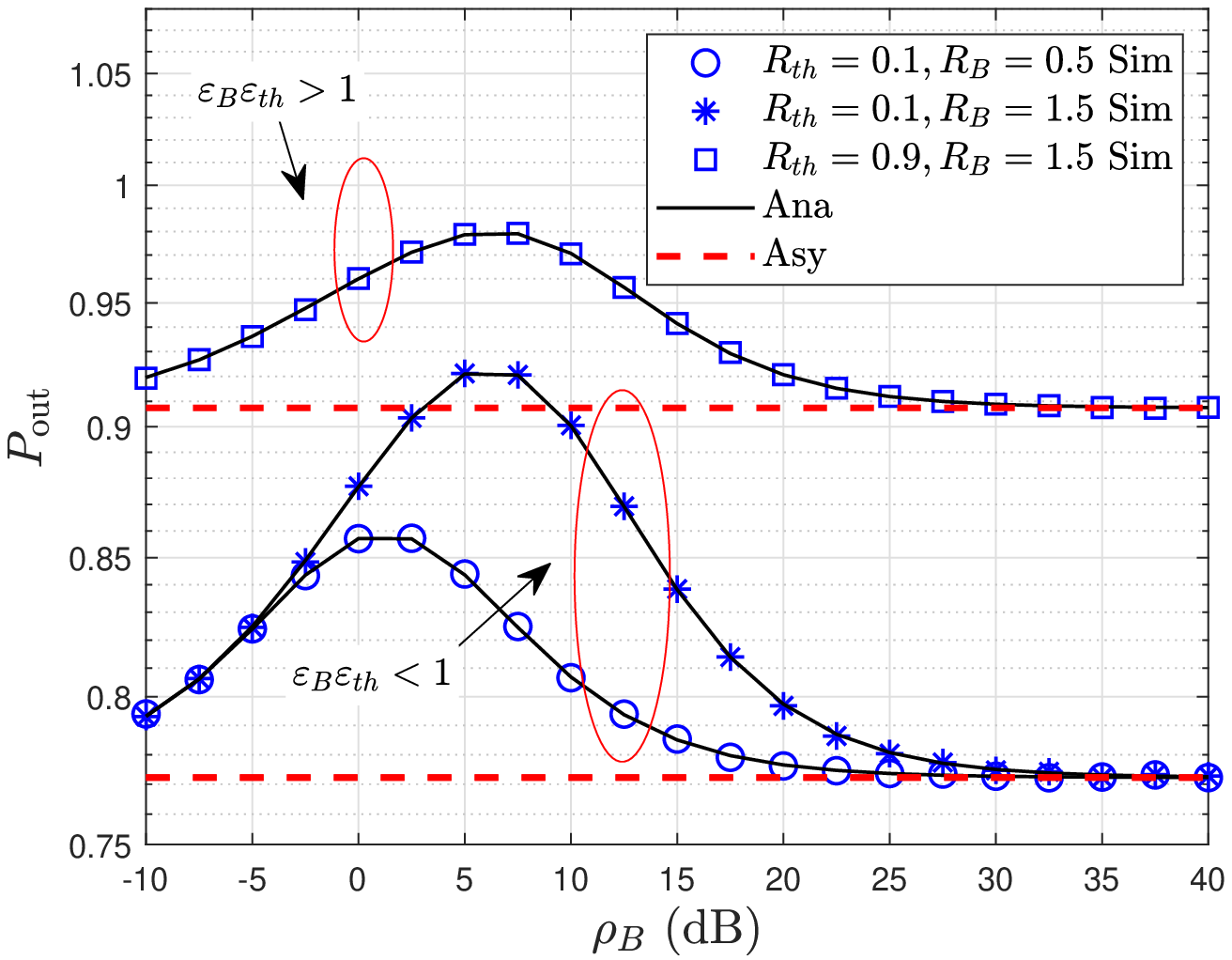}}
	\subfigure[SOP for varying $\rho_F$.]{
		\label{fig03b} 
		\includegraphics[width = 2.5in]{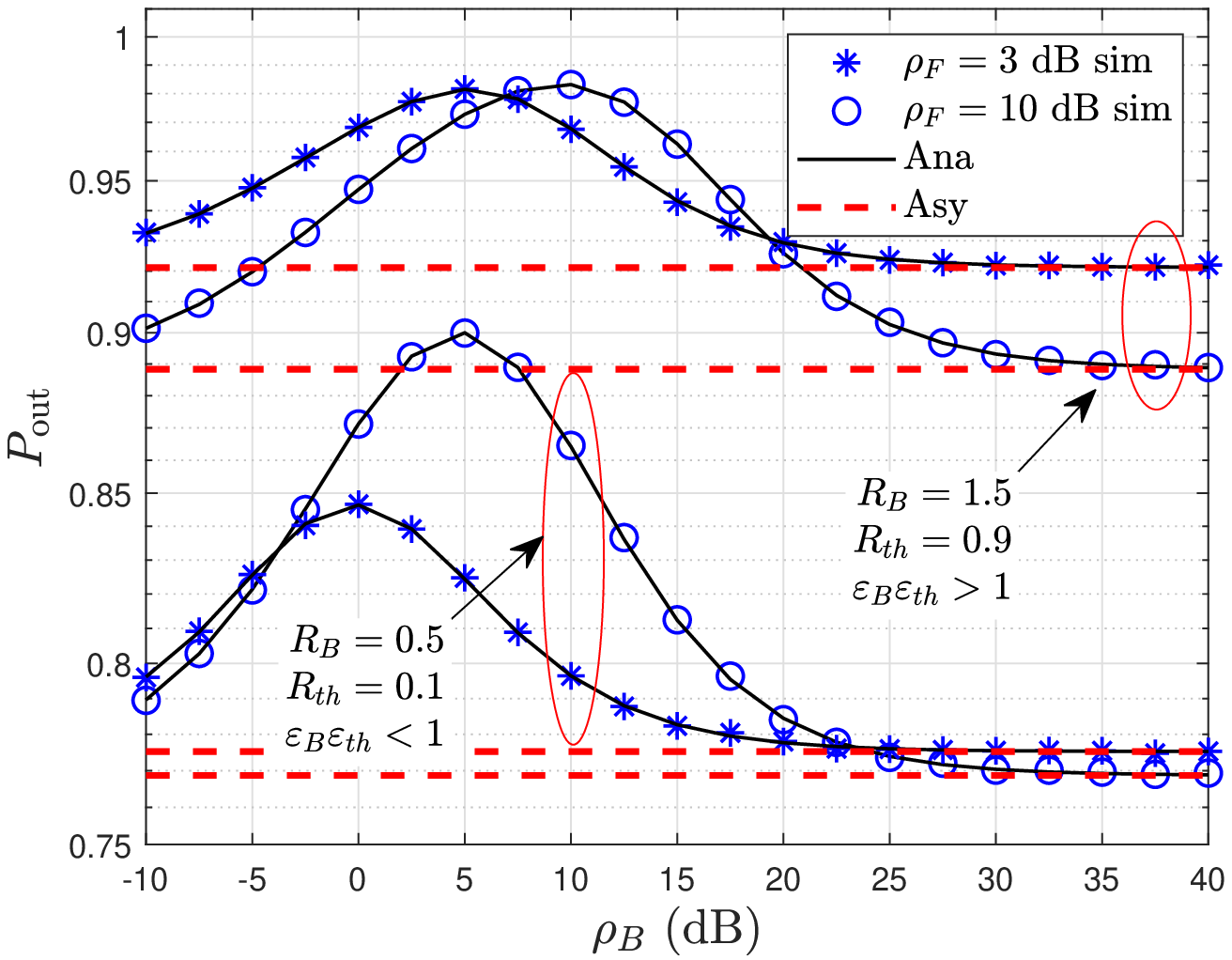}}
	\caption{SOP of the single-GF-user NOMA-aided SGF system with respect to ${\epsilon_B}{\epsilon_{th}}$ under varying $\rho_B$.}
	\label{fig03}
\end{figure}

Fig. \ref{fig03} demonstrates the SOP of the single-GF-user NOMA-aided SGF system with varying $\rho_B$. One can easily observe that the SOP increases initially and subsequently decreases with increasing $\rho_B$.
This is because $\alpha_B$ decreases as $\rho_B$ increases, then the probability that the signals from $U_F$ are decoded first decreases for a given $\rho_F$.
In the lower-$\rho_B$ region, the interference power tolerated for the $U_B$ is limited, so the signals from $U_F$ are always decoded first to ensure the QoS of $U_B$. The achievable rate for $U_F$ ($R_F^{\mathrm{I}}$) decreases with increasing of $\rho_B$ while the eavesdropping rate is independent of $\rho_B$; thus, the secrecy performance deteriorates.
As the $\rho_B$ increases, $\tau _B$ increases, whereas the probability of decoding signals from $U_F$ during the first stage of SIC decreases. In the larger-$\rho_B$ region, SOP tends to be a constant, independent of $\rho_B$ but depends on $\rho_F$ and $R_{th}$. Moreover, the effect from $R_{th}$ is relatively larger than that from $R_B$ since $R_B$ only affects $\tau_B$, i.e., the probability of decoding $x_F$ first, while $R_{th}$ not only affects the probability of decoding $x_F$ first but also affects the achievable rate for $U_F$ ($R_F^{\mathrm{I}}$).

\begin{figure}[t]
	\centering
	\subfigure[SOP for varying $R_{th}$ and $R_B$.  ]{
		\label{fig04a}
		\includegraphics[width = 2.5in]{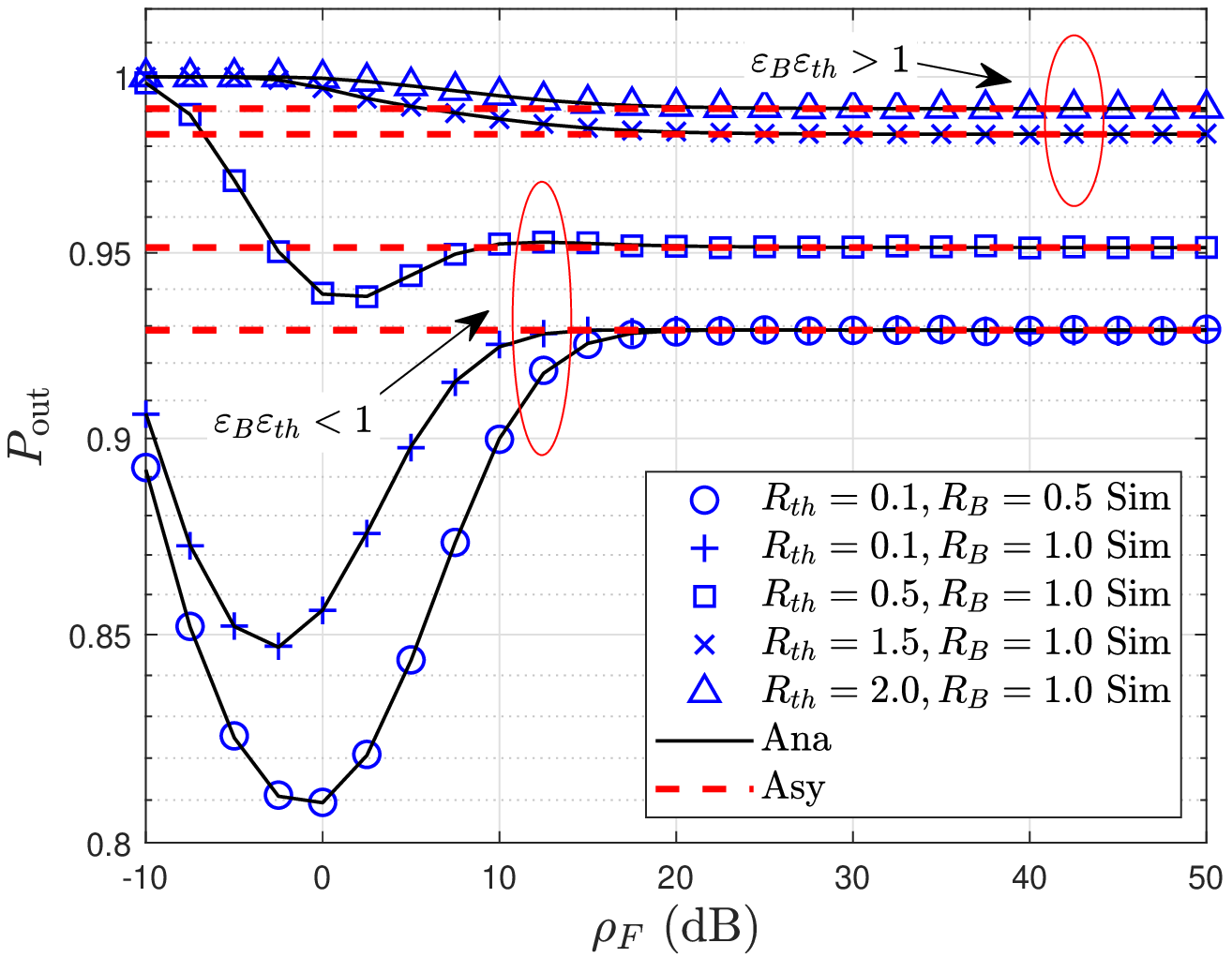}}
	\subfigure[SOP for varying $\rho_B$.]{
		\label{fig04b} 
		\includegraphics[width = 2.5in]{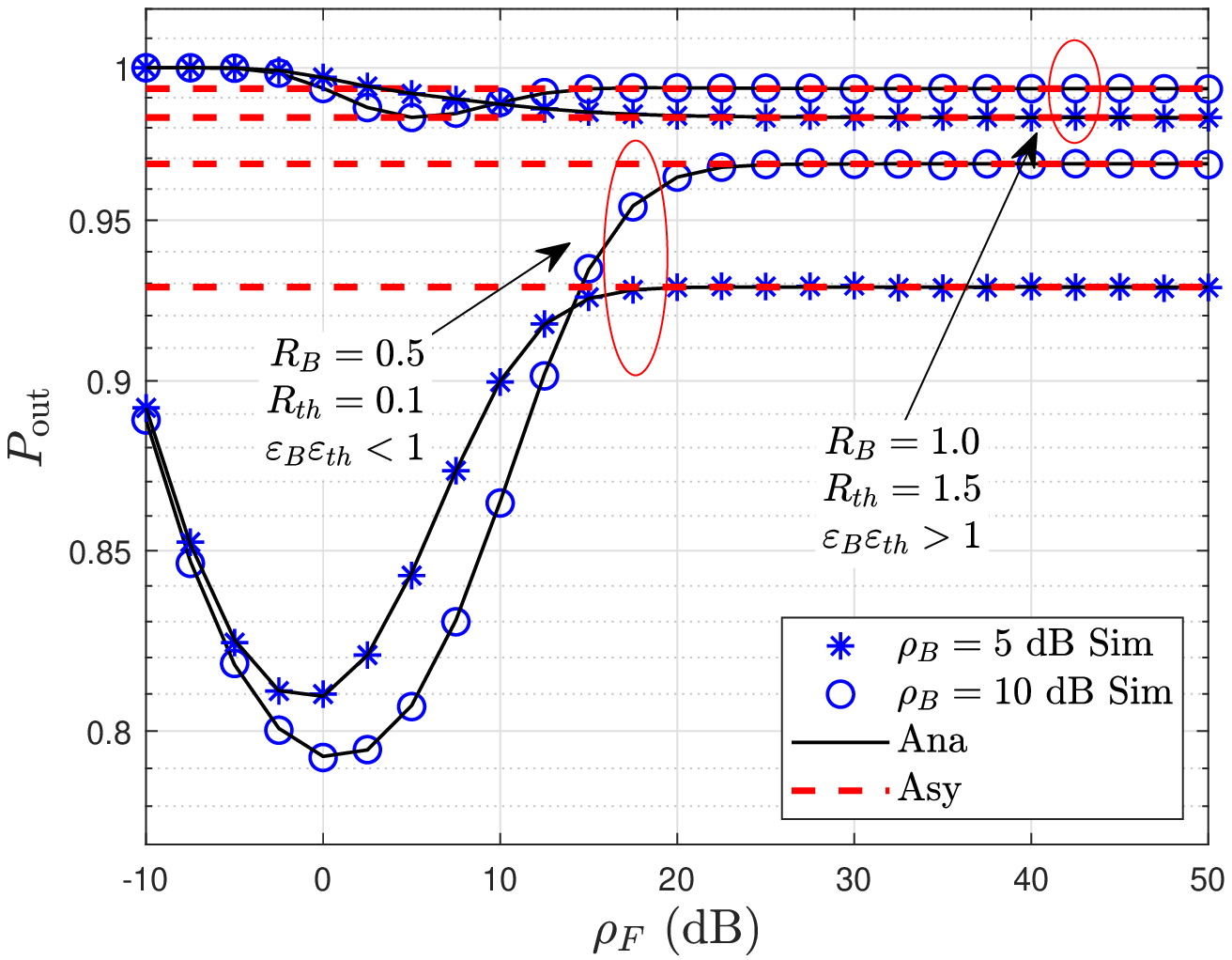}}
	\caption{SOP of the single-GF-user NOMA-aided SGF system with respect to ${\epsilon_B}{\epsilon_{th}}$ under varying $\rho_F$ .}
	\label{fig04}
\end{figure}
Figs. \ref{fig04} describes SOP of $U_F$ with varying $\rho_F$.
One can observe that SOP in the larger-$\rho_F$ region tends to be a constant. This is because the probability of decoding signals from $U_F$ during the first stage of SIC increases with increasing $\rho_F$, i.e. $\Pr \left\{ {{\rho_F}{{\left| {{h_F}} \right|}^2} > \tau \left( {{{\left| {{h_B}} \right|}^2}} \right)} \right\} \to 1$. Thus, we have $P_{out}^{\mathrm{II}} \to 0$ and $P_{out} = P_{out}^{\mathrm{I}} = \Pr \left\{ {R_s^{\mathrm{I}} < {R_{th}}} \right\}$, which depends on $R_{th}$ and $\rho_B$. Further, the SOP trends in the lower-$\rho_F$ region vary with different ${\varepsilon _B}{\varepsilon _{th}}$.
Specifically, when ${\varepsilon _B}{\varepsilon _{th}} > 1$, SOP decreases with increasing $\rho_F$. For the case with ${\varepsilon _B}{\varepsilon _{th}} < 1$ in lower-$\rho_F$ region, SOP firstly decreases, then increases to a constant.
An important factor is the probability of decoding during the first or second stage, which depends on $\rho_F$, $\rho_B$, and $\alpha_B$.
As $\rho_F$ increases or/and $\alpha_B$ increases, the probability of first decoding increases, and SOP increases.
As $\rho_B$ and $\tau_B$ increase, the probability of first decoding decreases, then SOP decreases, as shown in Fig. \ref{fig04}.

\begin{figure}[t]
	\centering
	\includegraphics[width = 2.5 in]{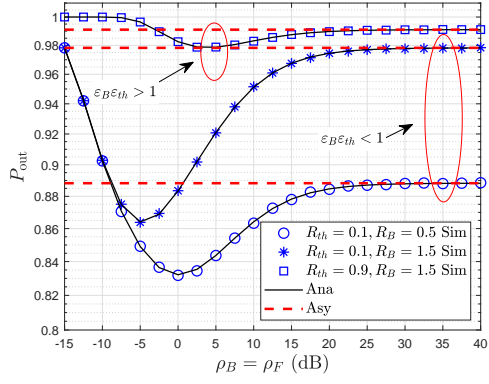}
	\caption{SOP of the single-GF-user NOMA-aided SGF system with respect to ${\epsilon_B}{\epsilon_{th}}$ under increasing $\rho_B$ = $\rho_F$.}
	\label{fig05}
\end{figure}
Fig. \ref{fig05} demonstrates SOP versus varying $\rho_B = \rho_F$ simultaneously.
One can observe that SOP of $U_F$ is enhanced and then becomes worse until it tends to a constant depending on ${R_B}$ and ${R_{th}}$ with increasing $\rho_B = \rho_F$.
This is because $\rho_F$ affects both the signal-to-interference-noise ratio (SINR) at $U_B$ and SNR at $E$ while $\rho_B$ only influences the SINR at $U_B$. Thus, $\rho_F$ has a stronger effect on SOP relative to $\rho_B$ when $\rho_B = \rho_F$ vary simultaneously.
Furthermore, when $\rho_B = \rho_F$ vary in a smaller range simultaneously, SOP depends mainly on $R_{th}$.
There is an optimal transmit SNR depending on ${R_B}$ and ${R_{th}}$ to obtain the minimum SOP in these scenarios.

\subsection{{SOP of the NOMA-aided SGF system with multiple GF users}}

\begin{figure}[t]
	\centering
	\subfigure[SOP for varying $K$ . ]{
		\label{fig06a}
		\includegraphics[width = 2.5in]{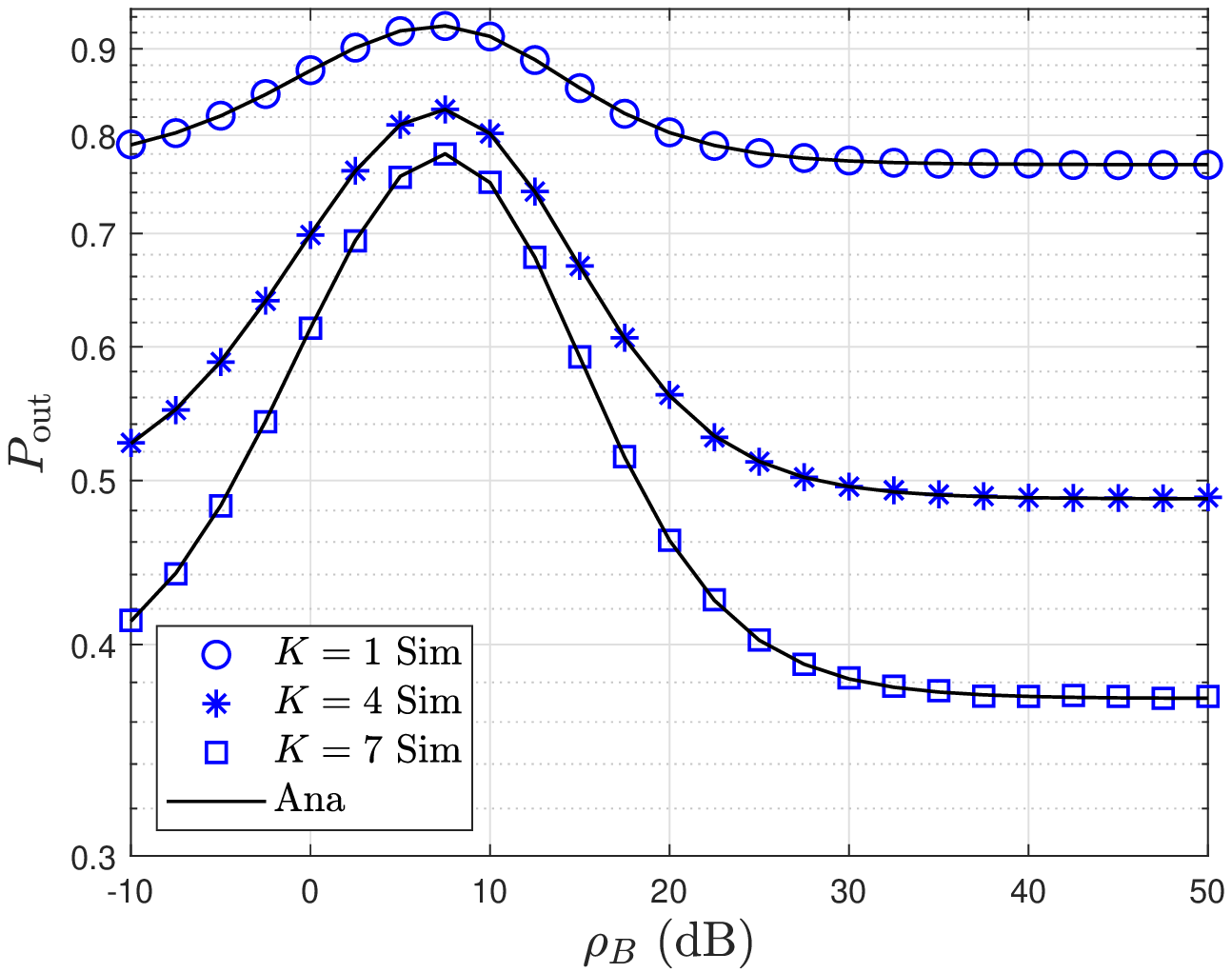}}
	\subfigure[ SOP for varying $R_{th}$ and $R_B$. ]{
		\label{fig06b} 
		\includegraphics[width = 2.5in]{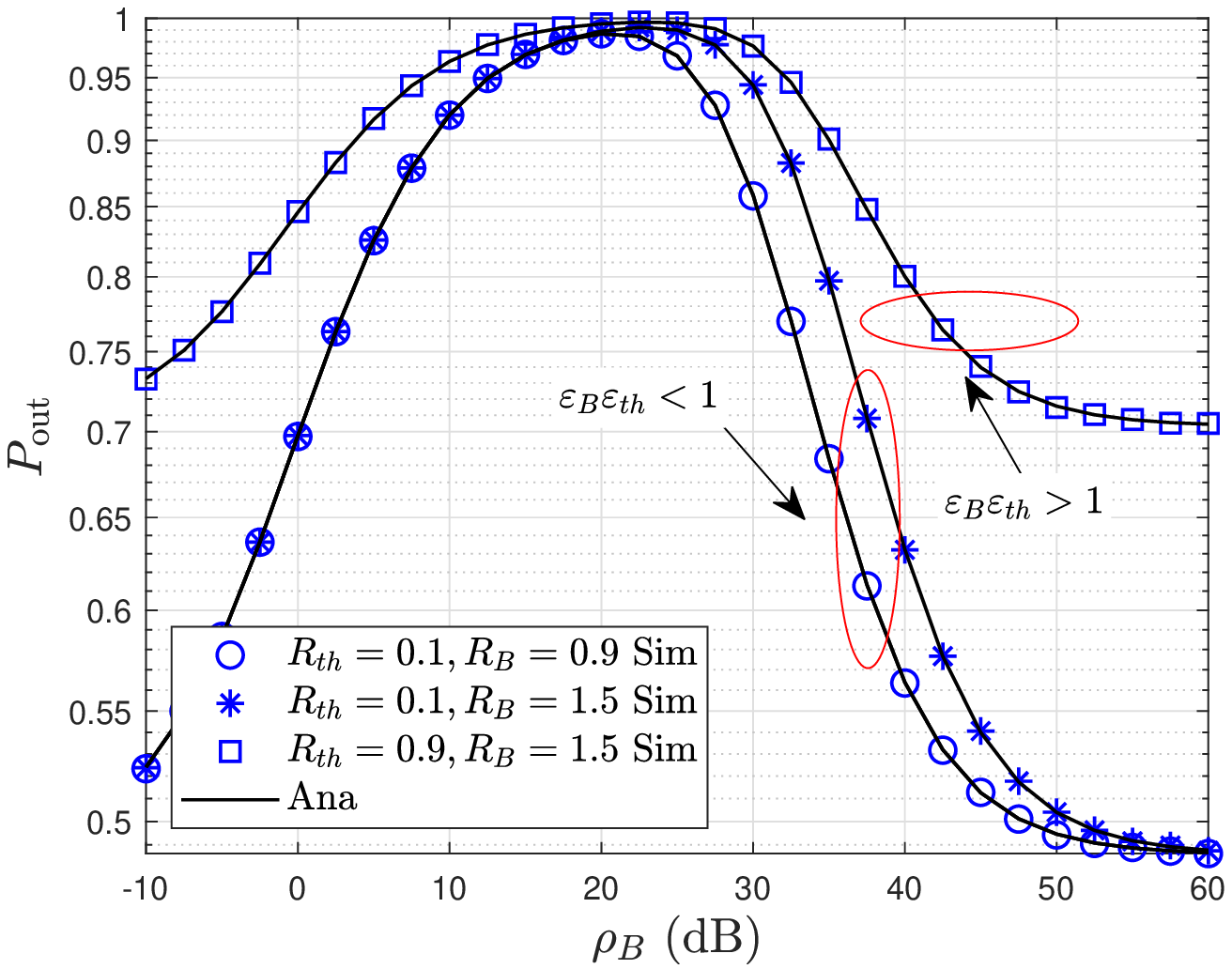}}
	\caption{SOP of the multiple-GF-user NOMA-aided SGF system experiencing $\rho_F$ = 10 dB .}
	\label{fig06}
\end{figure}
\begin{figure}[t]
	\centering
	\subfigure[SOP for varying $K$. ]{
		\label{fig07a}
		\includegraphics[width = 2.5in]{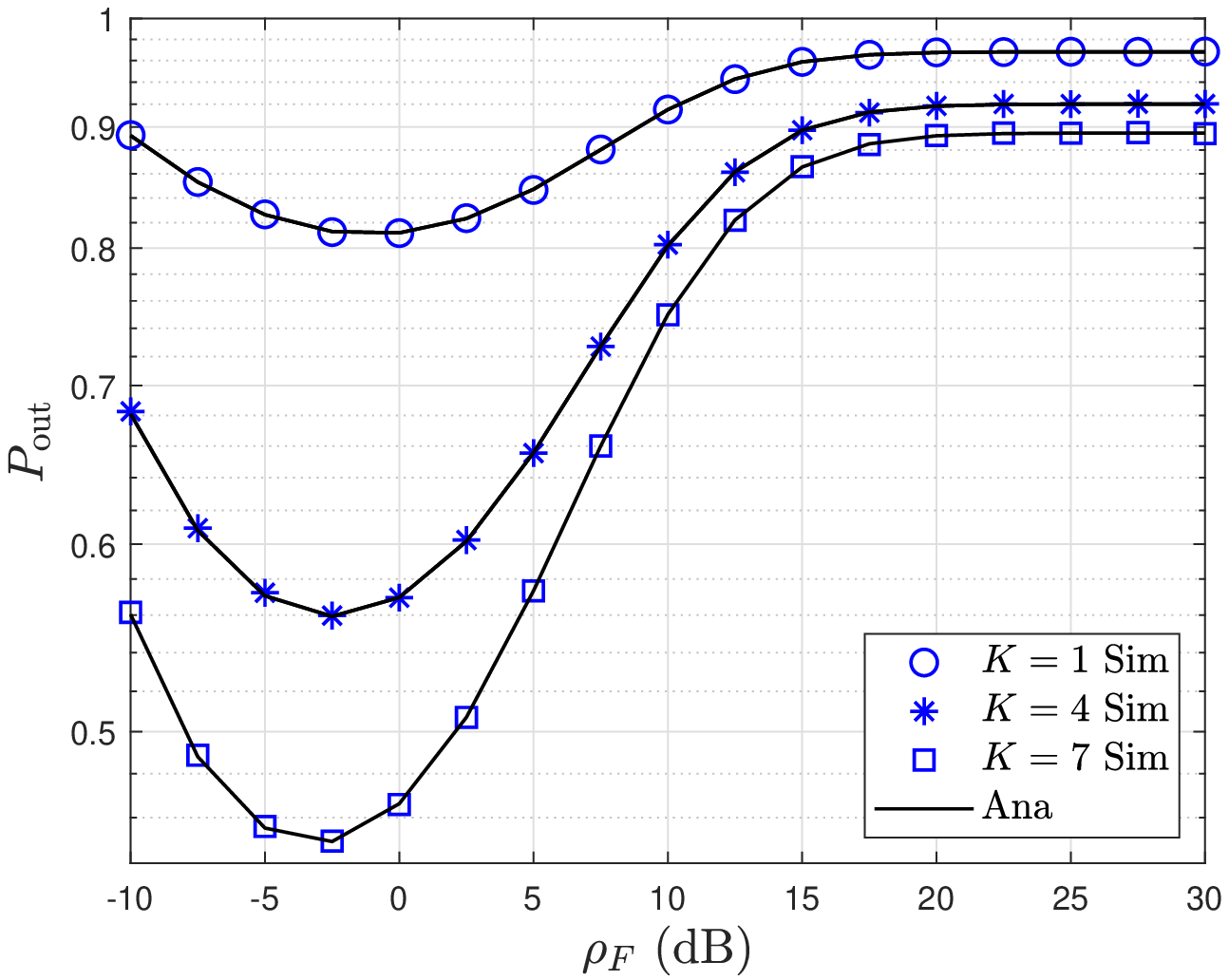}}
	\subfigure[SOP for varying $R_{th}$ and $R_B$.]{
		\label{fig07b} 
		\includegraphics[width = 2.5in]{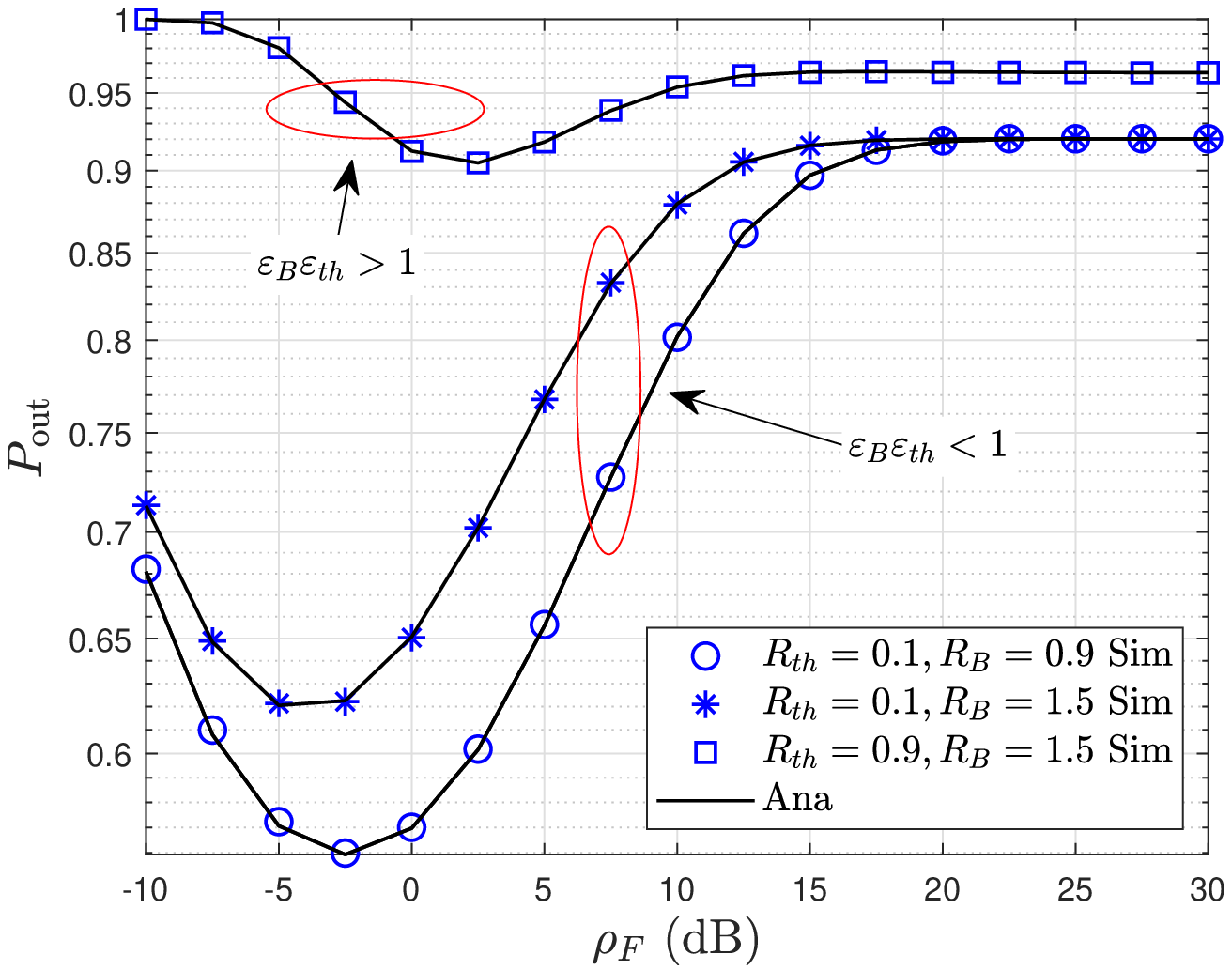}}
	\caption{SOP of the multiple-GF-user NOMA-aided SGF system experiencing $\rho_B$ = 10 dB.}
	\label{fig07}
\end{figure}
Figs. \ref{fig06} and \ref{fig07} demonstrate the impact of various $K$, ${R_B}$, and ${R_{th}}$ on SOP of $U_F$.
As can be observed from the figure, with the increase of $\rho_B$, SOP first increases and then decreases to a constant depending on $K$ and ${R_{th}}$.
Moreover, with an increase in $K$, the SOP improves since the better GF user is selected to access the channel, enhancing the secrecy performance.
Based on Figs. \ref{fig06} and \ref{fig07}, one can observe that the effect of the transmit SNR, $\rho_B$, and $\rho_F$, on the SOP with multiple GF users is similar to that in Figs. \ref{fig03} - \ref{fig04} with a single GF user.

\begin{figure*}[t]
	\centering
	\subfigure[SOP for varying $R_B$ and $R_{th}$.]{
		\label{fig007a}
		\includegraphics[width = 0.3 \textwidth]{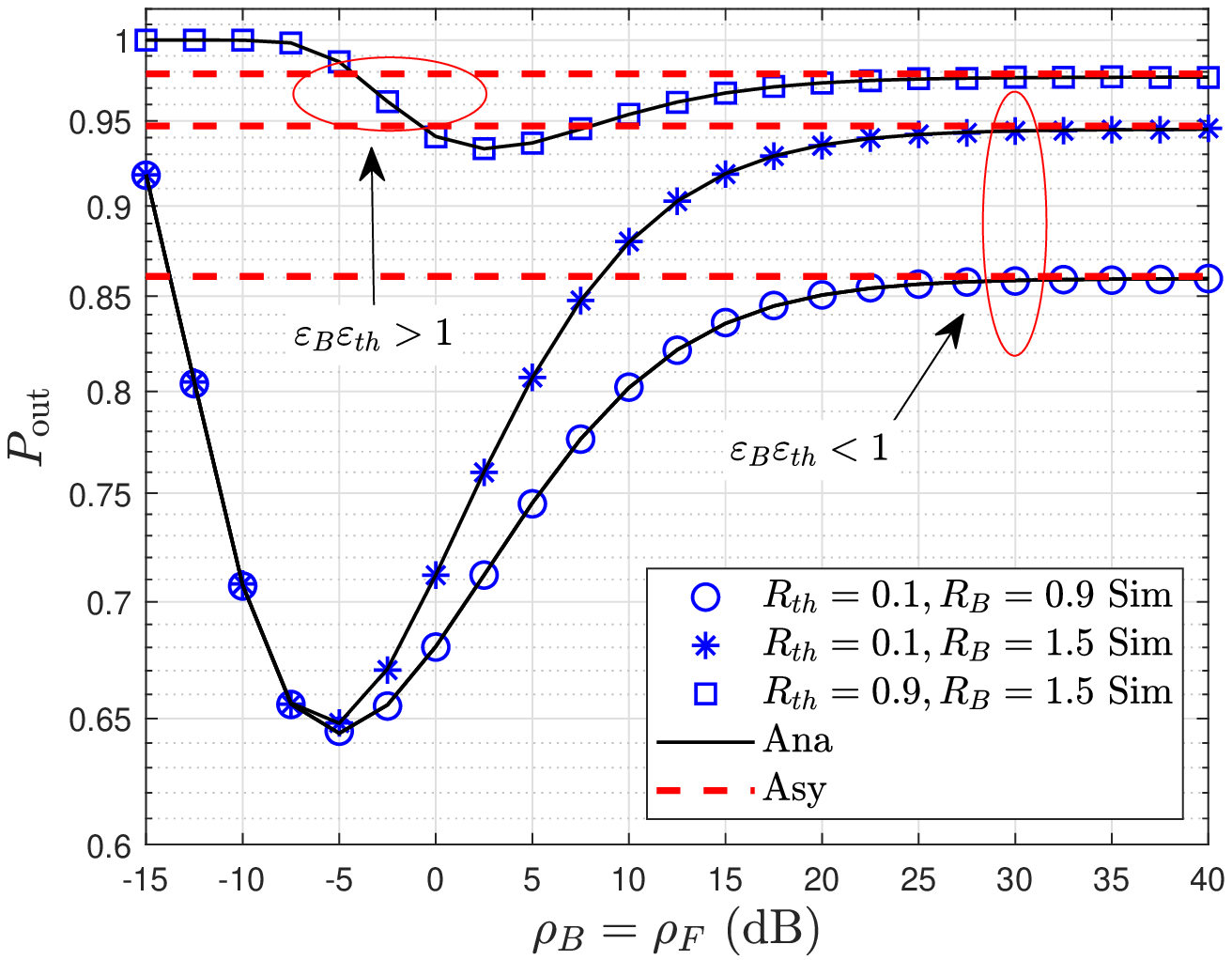}}
	\subfigure[SOP for varying $K$.]{
		\label{fig007b}
		\includegraphics[width = 0.3 \textwidth]{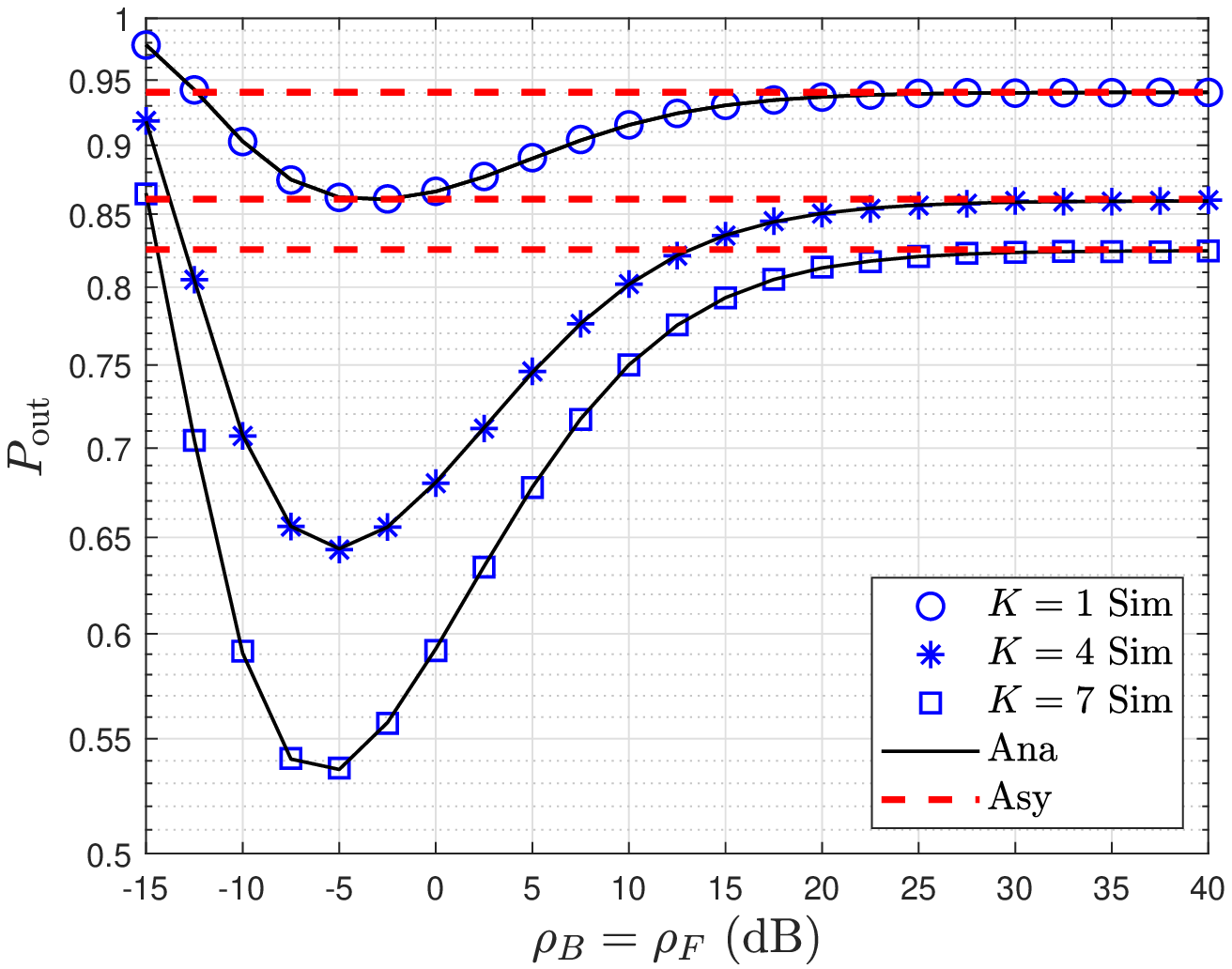}}	
	\subfigure[SOP for varying $N$.]{
		\label{fig007c}
		\includegraphics[width = 0.3 \textwidth]{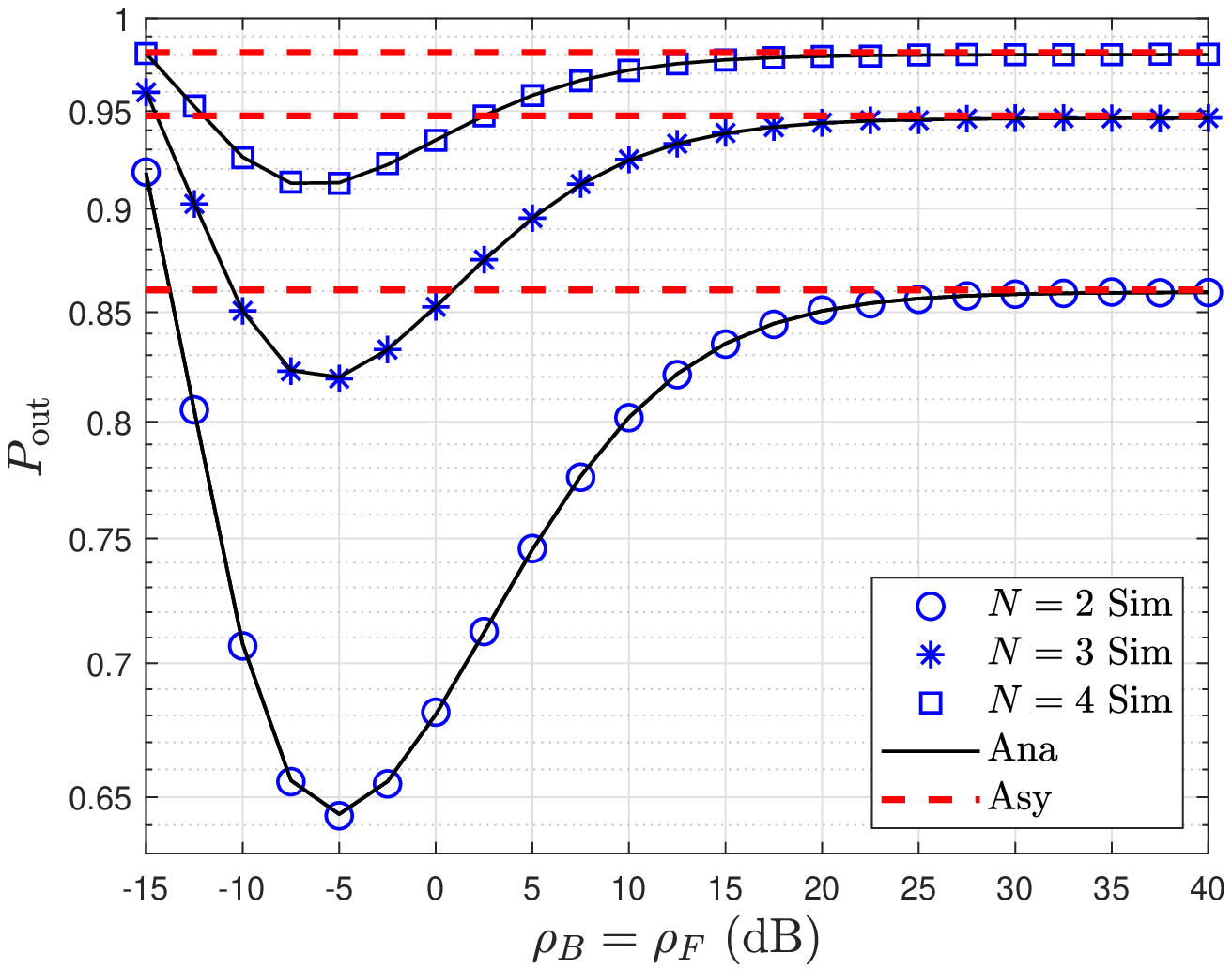}}
	\caption{SOP of the multiple-GF-user NOMA-aided SGF system versus varying $\rho_B = \rho_F$.}
	\label{fig007}
\end{figure*}

Fig. \ref{fig007} plots the effects of varying $K$, $R_B$, $R_{th}$, and $N$ on SOP versus varying $\rho_B = \rho_F$
One can observe that the curves of SOP in these scenarios are similar to those demonstrated in Fig. \ref{fig05}.
Moreover, from Fig. \ref{fig007c}, one can observe that SOP of $U_F$  becomes worse until it tends to be a constant depending on $N$.
This can be explained by the fact that  weakening diversity at $E$ implies a better security performance of the considered SGF system.

{
Comparing Figs. (\ref{fig03}) and (\ref{fig06}), (\ref{fig04}) and (\ref{fig07}), one interesting conclusion can be drawn that the transmit power of the GF and GB users has an opposite impact on the GF user' secrecy performance.
From the point of view of security of GF users, there exists an optimal $P_F$ and a worst $P_B$.}

\begin{figure*}[t]
	\centering
	\subfigure[SOP for varying $r_B$.]{
		\label{fig10a}
		\includegraphics[width = 0.318 \textwidth]{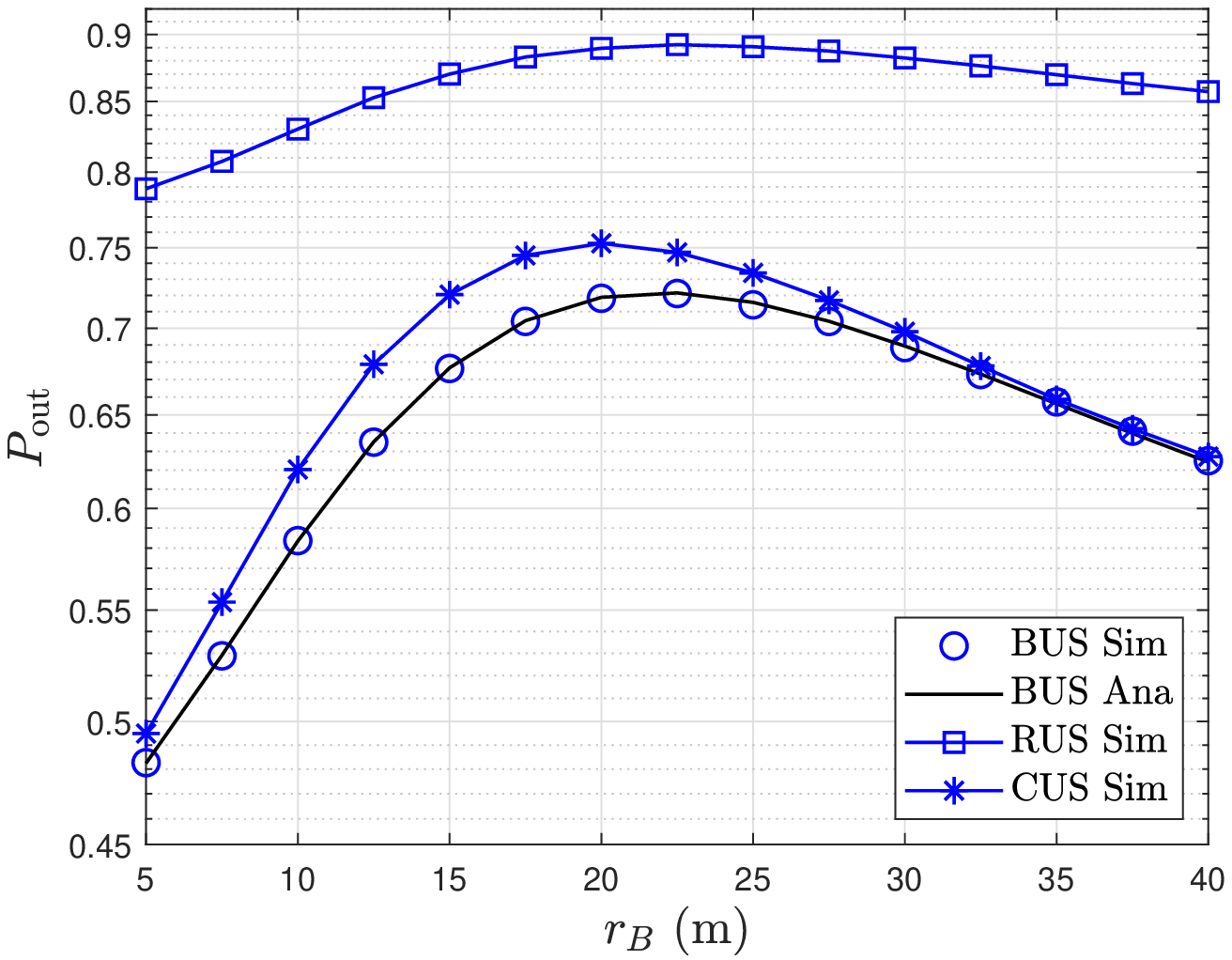}}
	\subfigure[SOP for varying $r_F$.]{
		\label{fig10b}
		\includegraphics[width = 0.318 \textwidth]{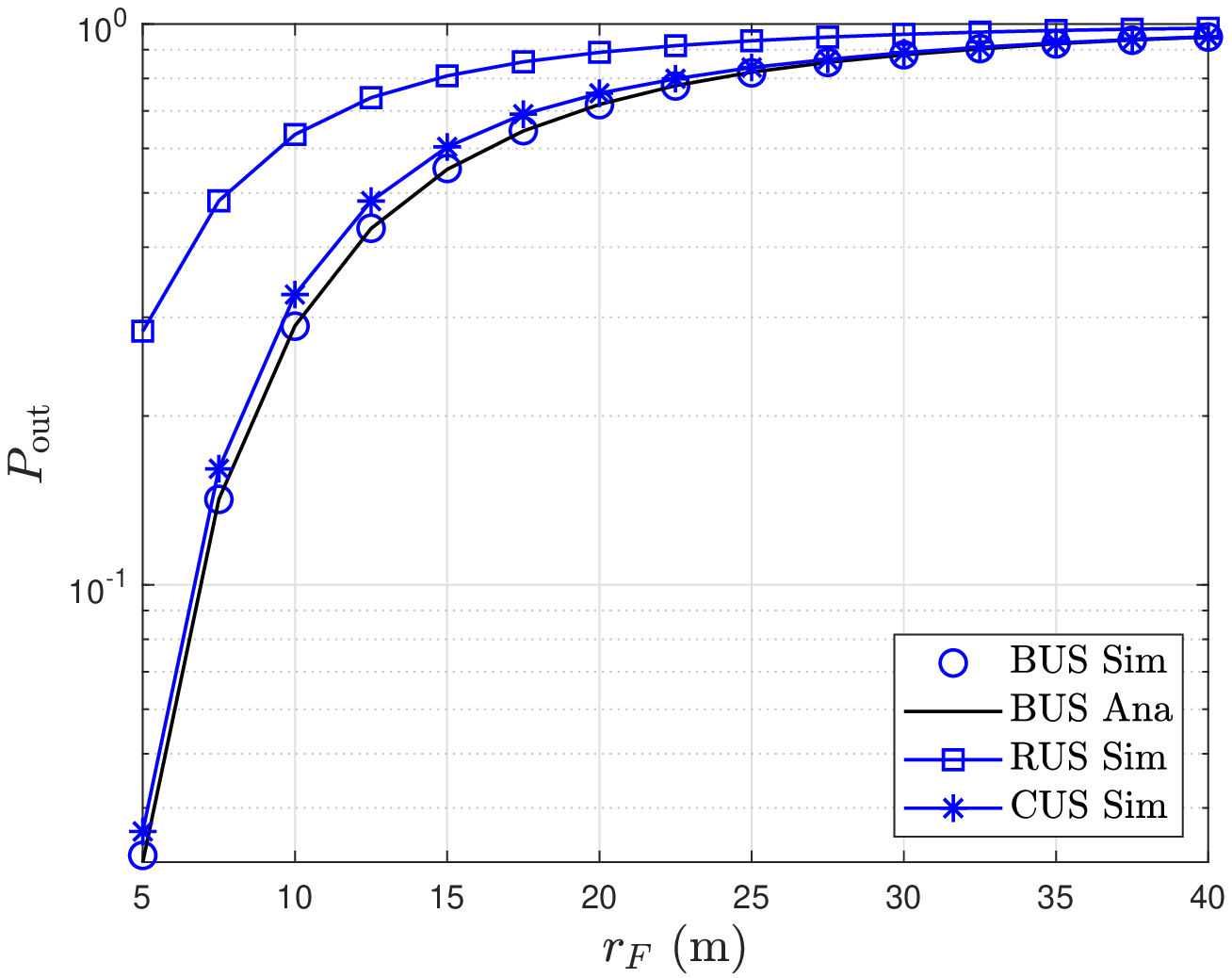}}
	\subfigure[SOP for varying $r_E$.]{
		\label{fig10c}
		\includegraphics[width = 0.318 \textwidth]{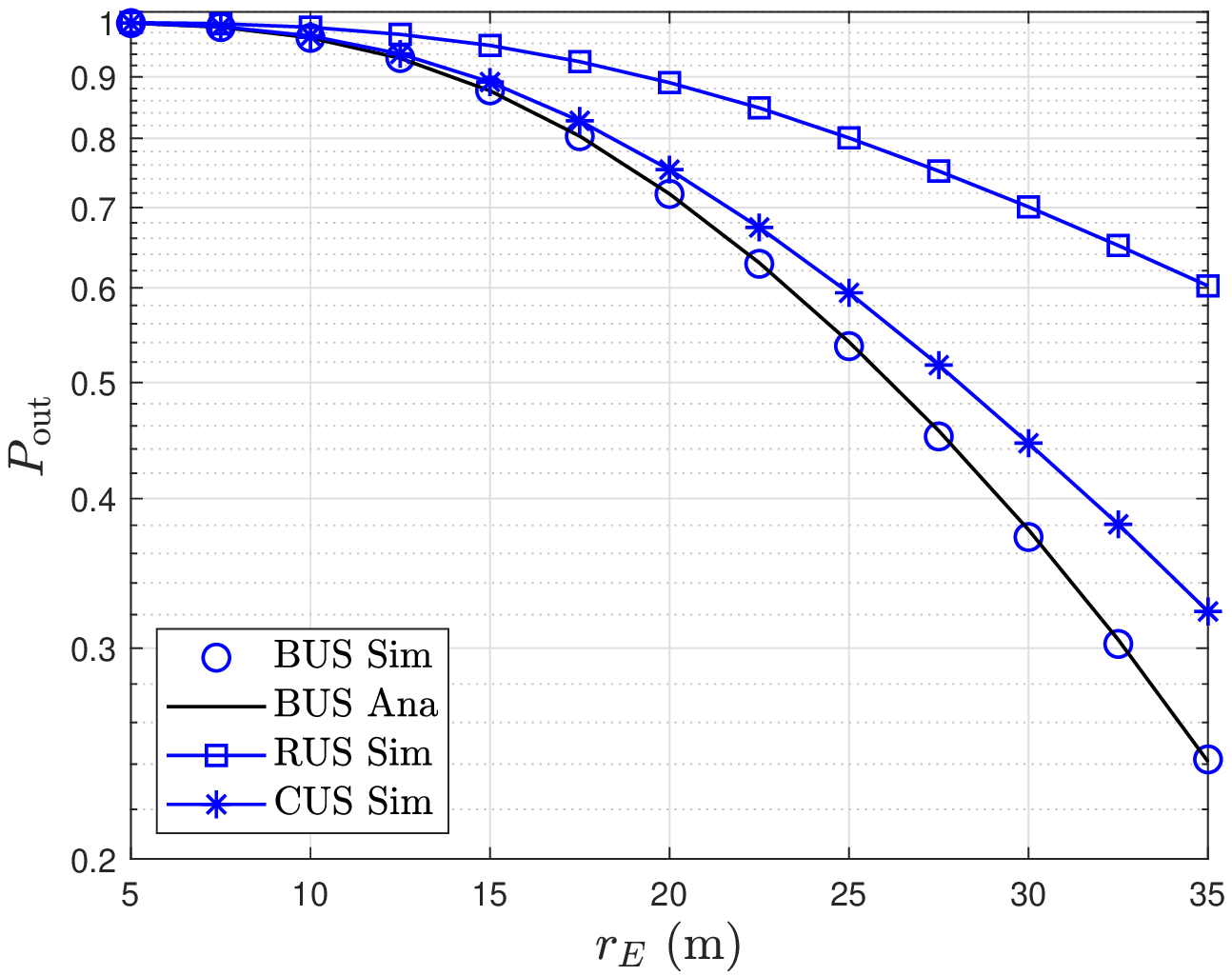}}
	\caption{SOP of the multiple-GF-user NOMA-aided SGF system for different user scheduling schemes with $\rho_B = \rho_F = 5$ dB.}
	\label{fig10}
\end{figure*}

{
Fig. \ref{fig10} demonstrates the NOMA-aided SGF system for different user scheduling schemes with varying $r_B$, $r_F$, and $r_E$.
From Fig. \ref{fig10a}, one can observe that the SOP increases initially and subsequently decreases with increasing $r_B$.
The achievable rate for $U_F$ decreases with increasing $r_B$ thereby the secrecy performance deteriorates.
As the $r_B$ increases, $\tau_B$ increases, whereas the probability of decoding signals from $U_F$ during the second stage of SIC increases. Thus, security of $U_F$ with all the schemes is enhanced.
Figs. \ref{fig10b} and \ref{fig10c} demonstrate that $r_F$ and $r_E$ have an opposite impact on the GF user's secrecy performance, which is easy to follow.
Furthermore, the BUS scheme obtains the best security while the RUS scheme obtains the worst secrecy performance.
This is because the GF user with maximum data rate is scheduled to transmit signals in the BUS scheme while a GF user is selected randomly in the RUS scheme.
Moreover, it can be observed that the difference between the secrecy performance with the BUS and CUS schemes is minor in the lower/larger-$r_B$ region (Fig. \ref{fig10a}) and lower/larger-$r_F$ (Fig. \ref{fig10b}).
The reason is as follows.
The CUS scheme is proposed to solve the fairness between GF users due to the difference in path loss in each group.
In the scenarios with lower/larger-$r_B$ region (Fig. \ref{fig10a}) or lower/larger-$r_F$ (Fig. \ref{fig10b}), the GF users belong to the same group with high probability.
Assuming the same distance between the GF user and the base station, the user with the maximum power gain leads to the maximum rate.
Thus, the secrecy performance with BUS and CUS schemes is equal.
}

\section{Conclusion}
\label{sec:Conclusion}
In this paper, we investigated the secrecy outage performance of the NOMA-aided SGF systems.
With the premise that GF users are entirely transparent for GB users, we first analyzed the NOMA-aided SGF system with a single GF user.
Subsequently, the secrecy performance of NOMA-aided SGF systems with multiple GF users was investigated.
The effects of all the parameters, such as the target data rate of GB users, the secrecy threshold rate of GF users, and transmit powers on GB and GF users, were discussed.
Monte-Carlo simulation results were presented to validate the correctness of the derived analytical expressions.

{
SIC and CSI are assumed to be perfect in this work, which is a typical assumption in many works, like \cite{AbbasR2019TCOM}-\cite{YangZ2020WCL}.
The performance results assuming perfect SIC can be seen as an upper bound of the case with imperfect SIC and worst-case SIC, respectively.
An exciting direction for future research is investigating the performance of NOMA-aided SGF systems with imperfect SIC and CSI.
In this work, it assumed that all users transmit at fixed power. However, the results in \cite{SunY2021TVT} and \cite{LuH2021TWC} showed that the system performance could be enhanced by carefully adjusting the transmit power of the GF and GB users.
As we analyzed previously, there exists an optimal $P_F$ and a worst $P_B$ for the security of GF users.
Thus, analyzing the secrecy performance of the NOMA-based SGF systems wherein both the transmit powers of the GB and GF users are dynamically adjusted in a coordinated manner will be exciting subsequent work.
To facilitate performance analysis, it is assumed that all the GF users are located in a small cluster, such that the distances between GF users and the base station are the same.
Another interesting problem is analyzing the performance of NOMA-aided SGF systems with multiple randomly distributed GB users, GF users, and eavesdroppers via stochastic geometry.
Furthermore, machine-type GF users in mMTC applications often have small data packets. Fairness is another issue that is as important as security.
Analyzing the secrecy performance of NOMA-based SGF systems for short-packet transmission with the different user scheduling schemes also is an exciting problem.
}
	
\begin{appendices}	
\section{Proof of Theorem 1  }
\label{appendicesA}
\setcounter{equation}{0}
\renewcommand\theequation{A.\arabic{equation}}	

\subsection{Derivation of ${P_{out}^{\mathrm{I}}}$}	
	
Based on the definition of $\tau \left( {{{\left| {{h_B}} \right|}^2}} \right)$, $P_{out}^{\mathrm{I}}$ is expressed as
	\begin{equation}
	\begin{aligned}
		P_{out}^{\mathrm{I}} &= \Pr \left\{ {R_s^{\mathrm{I}} < {R_{th}},{\rho _F}{{\left| {{h_F}} \right|}^2} > \tau \left( {{{\left| {{h_B}} \right|}^2}} \right),\tau \left( {{{\left| {{h_B}} \right|}^2}} \right) < 0} \right\}\\
		&+\Pr \left\{ {R_s^{\mathrm{I}} < {R_{th}},{\rho _F}{{\left| {{h_F}} \right|}^2} > \tau \left( {{{\left| {{h_B}} \right|}^2}} \right),\tau \left( {{{\left| {{h_B}} \right|}^2}} \right) > 0} \right\} \\
		&= \Pr \left\{ {R_s^{\mathrm{I}} < {R_{th}},{\rho_F}{{\left| {{h_F}} \right|}^2} > 0,{{\left| {{h_B}} \right|}^2} < {\alpha _B}} \right\} \\
		&+ \Pr \left\{ {R_s^{\mathrm{I}} < {R_{th}},{\rho_F}{{\left| {{h_F}} \right|}^2} > {\tau _B},{{\left| {{h_B}} \right|}^2} > {\alpha _B}} \right\}\\
		&= \underbrace {\Pr \left\{ {R_s^{\mathrm{I}} < {R_{th}},{{\left| {{h_B}} \right|}^2} < {\alpha _B}} \right\}}_{P_{out}^{\mathrm{I},1}}
		+ \underbrace {\Pr \left\{ {R_s^{\mathrm{I}} < {R_{th}},{{\left| {{h_F}} \right|}^2} > \frac{{{\tau _B}}}{{{\rho_F}}},{{\left| {{h_B}} \right|}^2} > {\alpha _B}} \right\}}_{P_{out}^{\mathrm{I},2}}.
		\label{poutI01}
	\end{aligned}
\end{equation}

In (\ref{poutI01}), ${P_{out}^{\mathrm{I},1}}$ denotes the SOP of $U_F$ when $U_B$ is in outage while accessing the channel alone.
In these scenarios, $S$ can not successfully decode the signals from $U_B$ while decoding signals from $U_F$ is a unique choice. ${P_{out}^{\mathrm{I},2}}$ signifies SOP of $U_F$ when $U_B$ is not in outage while accessing the channel alone. In this scenario, although $S$ can successfully decode the signals from $U_B$, the QoS of $U_B$ cannot be guaranteed because of the interference caused by $U_F$. Therefore, the signals from $U_F$ must be decoded at the first stage of SIC.

Substituting (\ref{rateF}) into (\ref{poutI01}) and after some algebraic manipulations, we obtain
	\begin{equation}
	\begin{aligned}
		P_{out}^{\mathrm{I},1} &= \Pr \left\{ {{{\log }_2}\left( {1 + \frac{{{\rho_F}{{\left| {{h_F}} \right|}^2}}}{{1 + {\rho_B}{{\left| {{h_B}} \right|}^2}}}} \right) - {{\log }_2}\left( {1 + {\rho_F}{{\left| {{H_E}} \right|}^2}} \right) < {R_{th}},{{\left| {{h_B}} \right|}^2} < {\alpha _B}} \right\} \\
		&= \Pr \left\{ {{{\left| {{h_F}} \right|}^2} < \underbrace {{\rho_B}{\theta _{th}}{{\left| {{h_B}} \right|}^2}{{\left| {{H_E}} \right|}^2} + {\theta _{th}}{{\left| {{H_E}} \right|}^2} + {\rho_B}{\alpha _{th}}{{\left| {{h_B}} \right|}^2} + {\alpha _{th}}}_{ \triangleq {\omega _0}\left( {{{\left| {{h_B}} \right|}^2},{{\left| {{H_E}} \right|}^2}} \right)},{{\left| {{h_B}} \right|}^2} < {\alpha _B}} \right\} \\
		&= \int_0^\infty  {\int_0^{{\alpha _B}} { {{F_{{{\left| {{h_F}} \right|}^2}}}\left( {{\omega _0}\left( {x,y} \right)} \right)} {f_{{{\left| {{h_B}} \right|}^2}}}\left( x \right)dx{f_{{{\left| {{H_E}} \right|}^2}}}\left( y \right)dy} }  \\
		&{= \frac{{r_B^\alpha r_E^{N\alpha }}}{{\Gamma \left( N \right)}}\int_0^\infty  {{y^{N - 1}}{e^{ - r_E^\alpha y}}\int_0^{{\alpha _B}} {{e^{ - r_B^\alpha x}}dxdy} }}\\
		& {- \frac{{r_B^\alpha r_E^{N\alpha }{e^{ - r_F^\alpha {\alpha _{th}}}}}}{{\Gamma \left( N \right)}}\int_0^\infty  {\int_0^{{\alpha _B}} {{y^{N - 1}}{e^{ - {\lambda _1}xy - {\lambda _2}x - {\lambda _3}y}}dxdy} } }\\
		&{= 1 - {e^{ - r_B^\alpha {\alpha _B}}} - \frac{{r_B^\alpha r_E^{N\alpha }{e^{ - r_F^\alpha {\alpha _{th}}}}{\omega _1}\left( {{\lambda _1},{\lambda _2},{\lambda _3}} \right)}}{{\Gamma \left( N \right)}},}
		\label{poutI0101}
	\end{aligned}
\end{equation}
where
${\omega _1}\left( {a,b,c} \right)   \mathop  = \limits^{\left( a \right)} \frac{{{b^{N - 1}}\Gamma \left( N \right)}}{{{a^N}}}{e^{\frac{{bc}}{a}}}\left( {\Gamma \left( {1 - N,\frac{{bc}}{a}} \right) - \Gamma \left( {1 - N,b{\alpha _B} + \frac{{bc}}{a}} \right)} \right)$,
${\alpha _{th}} = \frac{{\varepsilon _{th}}}{{{\rho_F}}}$,
${\varepsilon _{th}} = {\theta _{th}} - 1$,
${\theta _{th}}= {2^{{R_{th}}}} $,
{
${\lambda _1} = r_F^\alpha {\rho _B}{\theta _{th}}$,
${\lambda _2} = r_F^\alpha {\rho _B}{\alpha _{th}} + r_B^\alpha $,
${\lambda _3} = r_F^\alpha {\theta _{th}} + r_E^\alpha $,
}
and
$(a)$ is obtained via utilizing \cite[(3.383.10)]{Gradshteyn2007Book}.

Similarly, we obtain
\begin{equation}
	\begin{aligned}
		P_{out}^{\mathrm{I},2} &= \Pr \left\{ {{{\left| {{h_F}} \right|}^2} < {\omega _0}\left( {{{\left| {{h_B}} \right|}^2},{{\left| {{H_E}} \right|}^2}} \right),{{\left| {{h_F}} \right|}^2} > \frac{{{\tau _B}}}{{{\rho_F}}},{{\left| {{h_B}} \right|}^2} > {\alpha _B}} \right\}.
		\label{ppoutI02}
	\end{aligned}
\end{equation}

The relationship between ${{\omega _0}\left( {{{\left| {{h_B}} \right|}^2},{{\left| {{H_E}} \right|}^2}} \right)}$ and ${\frac{{{\tau _B}}}{{{\rho_F}}}}$ is expressed as
\begin{equation}
	\begin{aligned}
			\Pr \left\{ {\frac{{{\tau _B}}}{{{\rho_F}}} < {\omega _0} \left( {{{\left| {{h_B}} \right|}^2},{{\left| {{H_E}} \right|}^2}} \right)} \right\}
		& = \Pr \left\{ {1 - {\varepsilon _B}{\varepsilon _{th}} - {\varepsilon _B}{\rho _F}{\theta _{th}}{\left| {{H_E}} \right|^2} < 0} \right\}\\
		&+\Pr \left\{ {1 - {\varepsilon _B}{\varepsilon _{th}} - {\varepsilon _B}{\rho _F}{\theta _{th}}{\left| {{H_E}} \right|^2} > 0,{{\left| {{{h_B}}} \right|}^2} < {\alpha _2}} \right\}\\		
		& = \Pr \left\{ {{{\left| {{H_E}} \right|}^2} > {\alpha _1}} \right\} + \Pr \left\{ {{{{\left| {{H_E}} \right|}^2} } < {\alpha _1},{{\left| {{{h_B}}} \right|}^2} < {\alpha _2}} \right\} ,
		\label{relationship1}
	\end{aligned}
\end{equation}
where
${\theta _B}= {2^{{R_B}}} $,
${\alpha _1} = \frac{{1 - {\varepsilon _B}{\varepsilon _{th}}}}{{{\rho_F}{\theta _{th}}{\varepsilon _B}}}$,
${\alpha _2} = \frac{{{\rho_F}{\varepsilon _1}{{\left| {{H_E}} \right|}^2} + {\varepsilon _1}}}{{ - {\rho_F}{\theta _{th}}{\varepsilon _B}{{\left| {{H_E}} \right|}^2} + 1 - {\varepsilon _B}{\varepsilon _{th}}}} = \frac{{{\alpha _3}}}{{{\alpha _1} - {{\left| {{H_E}} \right|}^2}}} - \frac{1}{{{\rho_B}}}$,
${\varepsilon _1} = {\alpha _B}{\theta _{th}}$,
and
${\alpha _3} = \frac{{ {\theta _B}}}{{{\rho_F}{\theta _{th}}{\rho_B}{\varepsilon _B}}}$.
Eq. (\ref{relationship1}) is easy to follow, while the first item denotes the scenario that the eavesdropper link is too strong and the second term denotes that the eavesdropper link is relatively weak and there is a constraint on the GB link from the eavesdropper link.
Moreover, the relationship $ {\alpha _1}$ and 0 has important effect on the relationship between ${{\omega _0}\left( {{{\left| {{h_B}} \right|}^2},{{\left| {{H_E}} \right|}^2}} \right)}$ and ${\frac{{{\tau _B}}}{{{\rho_F}}}}$.

(i) When ${\varepsilon _B}{\varepsilon _{th}} < 1$, we have ${\alpha _1} > 0 $. Then, based on (\ref{ppoutI02}), ${P_{out}^{\mathrm{I},2}}$ is obtained as
\begin{equation}
	\begin{aligned}
		{P_{out}^{\mathrm{I},21}} &= \Pr \left\{ {{{\left| {{h_F}} \right|}^2} < {\omega _0}\left( {{{\left| {{h_B}} \right|}^2},{{\left| {{H_E}} \right|}^2}} \right),{{\left| {{h_F}} \right|}^2} > \frac{{{\tau _B}}}{{{\rho_F}}},{{\left| {{h_B}} \right|}^2} > {\alpha _B}} \right\}\\
		& = \Pr \left\{ {\frac{{{\tau _B}}}{{{\rho_F}}} < {{\left| {{h_F}} \right|}^2} < {\omega _0}\left( {{{\left| {{h_B}} \right|}^2},{{\left| {{H_E}} \right|}^2}} \right),{{\left| {{h_B}} \right|}^2} > {\alpha _B},{{\left| {{H_E}} \right|}^2} > {\alpha _1}} \right\}\\
		&+\Pr \left\{ {\frac{{{\tau _B}}}{{{\rho_F}}} < {{\left| {{h_F}} \right|}^2} < {\omega _0}\left( {{{\left| {{h_B}} \right|}^2},{{\left| {{H_E}} \right|}^2}} \right),{\alpha _B} < {{\left| {{h_B}} \right|}^2} < {\alpha _2},{{\left| {{H_E}} \right|}^2} < {\alpha _1}} \right\}\\
		&{= \frac{{{e^{ - r_B^\alpha {\alpha _B}}}r_B^\alpha \Gamma \left( {N,r_E^\alpha {\alpha _1}} \right)}}{{{\varepsilon _2}\Gamma \left( N \right)}} + {e^{\frac{{r_F^\alpha }}{{{P_F}}}}}r_B^\alpha r_E^{N\alpha }\frac{{{\omega _3}\left( {0,{\varepsilon _2},r_E^\alpha } \right)}}{{\Gamma \left( N \right)}}}\\
		& {- {e^{ - r_F^\alpha {\alpha _{th}}}}r_B^\alpha r_E^{N\alpha }\frac{{{\omega _2}\left( {{\lambda _1},{\lambda _2},{\lambda _3}} \right) + {\omega _3}\left( {{\lambda _1},{\lambda _2},{\lambda _3}} \right)}}{{\Gamma \left( N \right)}},}
		\label{poutI21}
	\end{aligned}
\end{equation}
where
\begin{equation}
	\begin{aligned}
		{\omega _2}\left( {a,b,c} \right) &= \int_{{\alpha _1}}^\infty  {\int_{{\alpha _B}}^\infty  {{y^{N - 1}}{e^{ - axy - bx - cy}}dxdy} }  \\
		&\mathop  = \limits^{\left( b \right)} \frac{{{b^{N - 1}}\Gamma \left( N \right)}}{{{a^N}}}{e^{\frac{{bc}}{a}}}\Gamma \left( {1 - N,b{\alpha _B} + \frac{{bc}}{a}} \right)  - {e^{ - b{\alpha _B}}}\Delta ,
		\label{H232}
	\end{aligned}
\end{equation}
\begin{equation}
	\begin{aligned}
     	{\Delta} &= \int_0^{{\alpha _1}} {\frac{{{y^{N - 1}}{e^{ - \left( {a{\alpha _B} + c} \right)y}}}}{{ay + b}}dy} \mathop {{\mathrm{ }} = }\limits^{\left( c \right)} \frac{{\pi {\alpha _1}}}{{2R}}\sum\limits_{r = 1}^R {\frac{{\sqrt {1 - \ell _r^2} }}{{a{\hbar _r} + b}}\hbar _r^{N - 1}{e^{ - \left( {a{\alpha _B} + c} \right){\hbar _r}}}},
		\label{deta1}
\end{aligned}
\end{equation}
and
\begin{equation}
	\begin{aligned}
		{\omega _3}\left( {a,b,c} \right) 
		&= \int_0^{{\alpha _1}} {\frac{{{y^{N - 1}}{e^{ - \left( {a{\alpha _B} + c} \right)y}}}}{{ay + b}}dy}  - {e^{\frac{b}{{{P_B}}} - a{\alpha _3}}}\int_0^{{\alpha _1}} {\frac{{{y^{N - 1}}{e^{\left( {\frac{a}{{{P_B}}} - c} \right)y - \frac{{{\alpha _3}\left( {a{\alpha _1} + b} \right)}}{{y - {\alpha _1}}}}}}}{{ay + b}}dy}\\
		&\mathop  = \limits^{\left( d \right)} \frac{{{b^{N - 1}}\Gamma \left( N \right)}}{{{a^N}}}{e^{\frac{{bc}}{a}}}\Gamma \left( {1 - N,b{\alpha _B} + \frac{{bc}}{a}} \right) - {\omega _2}\left( {a,b,c} \right) \\
		&- {e^{\frac{b}{{{P_B}}} - a{\alpha _3}}}\frac{{\pi {\alpha _1}}}{{2L}}\sum\limits_{l = 1}^L {\frac{{\sqrt {1 - \vartheta _l^2} }}{{a{v_l} + b}}v_l^{N - 1}{e^{\left( {\frac{a}{{{P_B}}} - c} \right){v_l} - \frac{{{\alpha _3}\left( {a{\alpha _1} + b} \right)}}{{{\alpha _1} - {v_l}}}}}},
		\label{H232}
	\end{aligned}
\end{equation}
{
${\varepsilon _2} = \frac{{r_F^\alpha }}{{{P_F}{\alpha _B}}} + r_B^\alpha $,
}
$(b)$ holds by applying \cite[(3.383.10)]{Gradshteyn2007Book},
and
$(c)$ and $(d)$ holds by
\cite[(3.383.10)]{Gradshteyn2007Book} and
applying Gaussian-Chebyshev quadrature \cite[(25.4.30)]{Abramowitz1972Book},
$R$ and $L$ is the summation item, which reflects accuracy vs. complexity,
${\ell _r} = \cos \left( {\frac{{2r - 1}}{{2R}}\pi } \right)$, ${\hbar _r} = \frac{{{\alpha _1}}}{2}\left( {{\ell _r} + 1} \right)$,
${\vartheta _l} = \cos \left( {\frac{{2l - 1}}{{2L}}\pi } \right)$,
and
${v_l} = \frac{{{\alpha _1}}}{2}\left( {{\vartheta _l} + 1} \right)$.

(ii)  When ${\varepsilon _B}{\varepsilon _{th}} > 1$, it has ${\alpha _1} < 0 $, then, $\Pr \left\{ {\frac{{{\tau _B}}}{{{\rho_F}}} < {\omega _0} \left( {{{\left| {{h_B}} \right|}^2},{{\left| {{H_E}} \right|}^2}} \right)} \right\}  = \Pr \left\{ {{{\left| {{H_E}} \right|}^2} > {\alpha _1}} \right\} = 1$. Thus, $P _{out}^{\mathrm{I},2}$ is expressed as
\begin{equation}
	\begin{aligned}
		P _{out}^{\mathrm{I},22} &= \Pr \left\{ {\frac{{{\tau _B}}}{{{\rho_F}}} < {{\left| {{h_F}} \right|}^2} < {\omega _0}\left( {{{\left| {{h_B}} \right|}^2},{{\left| {{H_E}} \right|}^2}} \right),{{\left| {{h_B}} \right|}^2} > {\alpha _B}} \right\}\\
		&= \int_0^\infty  {\int_{{\alpha _B}}^\infty  {\left( {{F_{{{\left| {{h_F}} \right|}^2}}}\left( {{\omega _0}\left( {x,y} \right)} \right) - {F_{{{\left| {{h_F}} \right|}^2}}}\left( {\frac{{{\tau _B}}}{{{\rho_F}}}} \right)} \right){f_{{{\left| {{h_B}} \right|}^2}}}\left( x \right)dx{f_{{{\left| {{H_E}} \right|}^2}}}\left( y \right)dy} } \\
		& {= \frac{{r_B^\alpha r_E^{N\alpha }{e^{\frac{{r_F^\alpha }}{{{P_{\rm{F}}}}}}}}}{{\Gamma \left( N \right)}}\int_0^\infty  {\int_{{\alpha _B}}^\infty  {{y^{N - 1}}{e^{ - {\varepsilon _2}x - r_E^\alpha y}}dxdy} }} \\
		& {- \frac{{r_B^\alpha r_E^{N\alpha }{e^{ - r_F^\alpha {\alpha _{th}}}}}}{{\Gamma \left( N \right)}}\int_0^\infty  {\int_{{\alpha _B}}^\infty  {{y^{N - 1}}{e^{ - {\lambda _1}xy - {\lambda _2}x - {\lambda _3}y}}dxdy} } }\\
		&{= \frac{{r_B^\alpha {e^{ - r_B^\alpha {\alpha _B}}}}}{{{\varepsilon _2}}} - \frac{{r_B^\alpha r_E^{N\alpha }{e^{ - r_F^\alpha {\alpha _{th}}}}{\omega _4}\left( {{\lambda _1},{\lambda _2},{\lambda _3}} \right)}}{{\Gamma \left( N \right)}}},
		\label{poutI22}
	\end{aligned}
\end{equation}	
where
\begin{equation}
	\begin{aligned}
		{\omega _4}\left( {a,b,c} \right) &= \int_0^\infty  {\int_{{\alpha _B}}^\infty  {{y^{N - 1}}{e^{ - axy - bx - cy}}} dxdy} \\
		& \mathop  = \limits^{\left( e \right)}  \frac{{{b^{N - 1}}\Gamma \left( N \right)}}{{{a^N}}}{e^{\frac{{bc}}{a}}}\Gamma \left( {1 - N,b{\alpha _B} + \frac{{bc}}{a}} \right),
		\label{H232}
	\end{aligned}
\end{equation}
step $(e)$ is obtained by applying  \cite[(3.383.10)]{Gradshteyn2007Book}.

\subsection{Derivation of $P_{out}^{\mathrm{II}}$}

Similar to (\ref{poutI01}), $P_{out}^{\mathrm{II}}$ is expressed as
\begin{equation}
	\begin{aligned}
		P_{out}^{\mathrm{II}} &= \Pr \left\{ {R_s^{\mathrm{II}} < {R_{th}},{\rho_F}{{\left| {h_F} \right|}^2} < 0,{{\left| {{h_B}} \right|}^2} < {\alpha _B}} \right\}  + \Pr \left\{ {R_s^{\mathrm{II}} < {R_{th}},{\rho_F}{{\left| {{h_F}} \right|}^2} < {\tau _B},{{\left| {{h_B}} \right|}^2} > {\alpha _B}} \right\}\\
		&= \Pr \left\{ {R_s^{\mathrm{II}} < {R_{th}},{{{\left| {{h_F}} \right|}^2} < \frac{{{\tau _B}}}{{{\rho_F}}}},{{\left| {{h_B}} \right|}^2} > {\alpha _B}} \right\}\\
		&= \Pr \left\{ {{{\left| {h_F} \right|}^2} < {\theta _{th}}{{\left| {{H_E}} \right|}^2} + {\alpha _{th}},{{\left| {h_F} \right|}^2} < \frac{{{\tau _B}}}{{{\rho_F}}},{{\left| {{h_B}} \right|}^2} > {\alpha _B}} \right\}.
		\label{poutII1}
	\end{aligned}
\end{equation}
In this case, the relationship between constraint on security (${{\theta _{th}}{{\left| {{H_E}} \right|}^2} + {\alpha _{th}}}$) and constraint on decoding order is considered as follow.
\begin{equation}
	\Pr \left\{ {\frac{{{\tau _B}}}{{{\rho_F}}} < {\theta _{th}}{{\left| {{H_E}} \right|}^2} + {\alpha _{th}}} \right\} = \Pr \left\{ {{{\left| {{h_B}} \right|}^2} < \left( {{\rho_F}{{\left| {{H_E}} \right|}^2} + 1} \right){\varepsilon _1}} \right\}.
\end{equation}
Due to ${\theta _{th}} = {2^{{R_{th}}}} \geqslant 1$, then $\Pr \left\{ {\left( {{\rho _F}{{\left| {{H_E}} \right|}^2} + 1} \right){\varepsilon _1} > {\alpha _B}} \right\} = 1$. Thus, $P_{out}^{\mathrm{II}}$ is obtained as
\begin{equation}
	\begin{aligned}
		P_{out}^{\mathrm{II}} &= \Pr \left\{ {{{\left| {h_F} \right|}^2} < {\theta _{th}}{{\left| {{H_E}} \right|}^2} + {\alpha _{th}},{{\left| {{h_B}} \right|}^2} > \left( {{\rho_F}{{\left| {{H_E}} \right|}^2} + 1} \right){\varepsilon _1}} \right\} \\
		&+ \Pr \left\{ {{{\left| {h_F} \right|}^2} < \frac{{{\tau _B}}}{{{\rho_F}}},{\alpha _B} < {{\left| {{h_B}} \right|}^2} < \left( {{\rho_F}{{\left| {{H_E}} \right|}^2} + 1} \right){\varepsilon _1}} \right\} \\
		&= \int_0^\infty  {\int_{\left( {{\rho_F}y + 1} \right){\varepsilon _1}}^\infty  {{F_{{{\left| {{h_F}} \right|}^2}}}\left( {{\theta _{th}}{y} + {\alpha _{th}}} \right)} } {f_{{{\left| {{h_B}} \right|}^2}}}\left( x \right)dx{f_{{{\left| {{H_E}} \right|}^2}}}\left( y \right)dy\\
		&+ \int_0^\infty  {\int_{{\alpha _B}}^{\left( {{\rho_F}y + 1} \right){\varepsilon _1}} {{F_{{{\left| {{h_F}} \right|}^2}}}\left( {\frac{{{\tau _B}}}{{{\rho_F}}}} \right)} } {f_{{{\left| {{h_B}} \right|}^2}}}\left( x \right)dx{f_{{{\left| {{H_E}} \right|}^2}}}\left( y \right)dy\\
		& {= \frac{{r_E^{N\alpha }r_B^\alpha }}{{\Gamma \left( N \right)}}\int_0^\infty  {\int_{\left( {{\rho _F}y + 1} \right){\varepsilon _1}}^\infty  {\left( {1 - {e^{ - r_F^\alpha \left( {{\theta _{th}}y + {\alpha _{th}}} \right)}}} \right){e^{ - r_B^\alpha x}}{y^{N - 1}}{e^{ - r_E^\alpha y}}} } dxdy}\\
		&{+ \frac{{r_E^{N\alpha }r_B^\alpha }}{{\Gamma \left( N \right)}}\int_0^\infty  {\int_{{\alpha _B}}^{\left( {{\rho _F}y + 1} \right){\varepsilon _1}} {\left( {1 - {e^{ - r_F^\alpha \left( {\frac{x}{{{P_{\rm{F}}}{\alpha _B}}} - \frac{1}{{{P_{\rm{F}}}}}} \right)}}} \right)} } {e^{ - r_B^\alpha x}}{y^{N - 1}}{e^{ - r_E^\alpha y}}dxdy}\\
		& {= \frac{{r_F^\alpha {e^{ - r_B^\alpha {\alpha _B}}}}}{{r_B^\alpha {P_{\rm{F}}}{\alpha _B} + r_F^\alpha }} - \frac{{r_E^{N\alpha }r_F^\alpha {e^{ - \left( {r_F^\alpha {\alpha _{th}} + r_B^\alpha {\varepsilon _1}} \right)}}}}{{\left( {r_F^\alpha {\rho _F}{\alpha _B} + r_B^\alpha } \right){{\left( {r_B^\alpha {\rho _F}{\varepsilon _1} + {\lambda _3}} \right)}^N}}}.}
		\label{EpoutII}
	\end{aligned}
\end{equation}
Substituting (\ref{poutI0101}), (\ref{poutI21}), (\ref{poutI22}), (\ref{EpoutII}) into (\ref{SOPGF}), we have (\ref{Theorem1}).

\section{Proof of Corollary 1}
\label{appendices1}
\setcounter{equation}{0}
\renewcommand\theequation{B.\arabic{equation}}

When ${\rho_B} \to \infty $, we have ${\alpha _B} = \frac{{\varepsilon _B}}{{{\rho_B}}} \to 0 $. Based on (\ref{SOPGF}), we obtain ${\tau _B} \to \infty $, then $P_{out}^{\mathrm{I}} \approx 0$ and
$P_{out}^{\mathrm{II}}$ is approximated as
\begin{equation}
\begin{aligned}
P_{out}^{{\mathrm{II}},{\rho_B} \to \infty} &\approx \Pr \left\{ {R_s^{{\mathrm{II}}} < {R_{th}}} \right\} \\
&= \Pr \left\{ {{{\left| {{h_F}} \right|}^2} < {\theta _{th}}{{\left| {{H_E}} \right|}^2} + {\alpha _{th}}} \right\}\\
&= \int_0^\infty  {{F_{{{\left| {{h_F}} \right|}^2}}}\left( {{\theta _{th}}x + {\alpha _{th}}} \right){f_{{{\left| {{H_E}} \right|}^2}}}\left( x \right)} dx\\
& = {1 - {e^{ - r_F^\alpha {\alpha _{th}}}}{\left( {1 + {\theta _{th}}{{\left( {\frac{{{r_F}}}{{{r_E}}}} \right)}^\alpha }} \right)^{ - N}}.}
\label{asyBPoutII}
\end{aligned}
\end{equation}


\section{Proof of Corollary 2}
\label{corollary2}
\setcounter{equation}{0}
\renewcommand\theequation{C.\arabic{equation}}

When ${\rho_F} \to \infty $, based on (\ref{SOPGF}), we can easily observe $P_{out}^{{\mathrm{II}}} \approx 0$ and
$P_{out}^{\mathrm{I}}$ is expressed as
\begin{equation}
	\begin{aligned}
		P_{out}^{\mathrm{I},{\rho_F} \to \infty} & = \Pr \left\{ {R_s^{\mathrm{I}} < {R_{th}}} \right\}\\
		& \approx \Pr \left\{ {{{\left| {{h_F}} \right|}^2} < \left( {{\rho _B}{\theta _{th}}{{\left| {{h_B}} \right|}^2} + {\theta _{th}}} \right){{\left| {{H_E}} \right|}^2}} \right\}\\
		&= \int_0^\infty  {\int_0^\infty  {{F_{{{\left| {{h_F}} \right|}^2}}}\left( {\left( {{\rho _B}{\theta _{th}}x + {\theta _{th}}} \right)y } \right)} } {f_{{{\left| {{h_B}} \right|}^2}}}\left( x \right){f_{{{\left| {{H_E}} \right|}^2}}}\left( y \right)dxdy\\
		& ={1 - {\left( {\frac{{{r_B}}}{{{r_F}}}} \right)^{N\alpha }}{\left( {\frac{{r_E^\alpha }}{{{\rho _B}{\theta _{th}}}}} \right)^N}\Gamma \left( {1 - N,\frac{{r_B^\alpha }}{{{\rho _B}{\theta _{th}}}}\left( {{\theta _{th}} + {{\left( {\frac{{{r_E}}}{{{r_F}}}} \right)}^\alpha }} \right)} \right).}
		\label{H232}
	\end{aligned}
\end{equation}
%
%

\section{Proof of Corollary 3}
\label{corollary3}
\setcounter{equation}{0}
\renewcommand\theequation{D.\arabic{equation}}

When ${\rho_B} = {\rho_F} \to \infty $, we have $\frac{{{\tau _B}}}{{{\rho_F}}} \to \frac{{{{\left| {{h_B}} \right|}^2}}}{{{\varepsilon _B}}}$. Based on (\ref{SOPGF}), $P_{out}^{\mathrm{I}} $ is approximated as
%
\begin{equation}\small
\begin{aligned}
P_{out}^{{\rm{I}}, \infty } &\approx \Pr \left\{ {\frac{{{{\left| {{h_B}} \right|}^2}}}{{{\varepsilon _B}}} < {{\left| {{h_F}} \right|}^2} < {\omega _0}\left( {{{\left| {{h_B}} \right|}^2},{{\left| {{H_E}} \right|}^2}} \right)} \right\}\\
&= \int_0^\infty  {\int_0^\infty  {\left( {{F_{{{\left| {{h_F}} \right|}^2}}}\left( {{\omega _0}\left( {x,y} \right)} \right) - {F_{{{\left| {{h_F}} \right|}^2}}}\left( {\frac{x}{{{\varepsilon _B}}}} \right)} \right){f_{{{\left| {{h_B}} \right|}^2}}}\left( x \right)dx{f_{{{\left| {{H_E}} \right|}^2}}}\left( y \right)dy} } \\
&{\mathop  \approx \limits^{\left( f \right)} 1 - \frac{1}{{1 + {\varepsilon _B}{{\left( {\frac{{{r_B}}}{{{r_F}}}} \right)}^\alpha }}}.}
\label{pou11Iinf}
\end{aligned}
\end{equation}
where $(f)$ holds by applying \cite[(3.383.10)]{Gradshteyn2007Book} and $\Gamma \left( {a,x} \right)\mathop  \to \limits^{x \to \infty } 0$.
Based on (\ref{EpoutII}), $P_{out}^{\mathrm{II}} $ is approximated as
\begin{equation}
\begin{aligned}
P_{out}^{{\mathrm{II}}, \infty} &\approx \Pr \left\{ {R_s^{{\mathrm{II}}} < {R_{th}},{{\left| {{h_F}} \right|}^2} < \frac{{{{\left| {{h_B}} \right|}^2}}}{{{\varepsilon _B}}}} \right\}\\
&= \Pr \left\{ {{{\left| {{h_F}} \right|}^2} < {\theta _{th}}{{\left| {{H_E}} \right|}^2},{{\left| {{h_B}} \right|}^2} > {\varepsilon _B}{\theta _{th}}{{\left| {{H_E}} \right|}^2}} \right\} \\
&+ \Pr \left\{ {{{\left| {{h_F}} \right|}^2} < \frac{{{{\left| {{h_B}} \right|}^2}}}{{{\varepsilon _B}}},{{\left| {{h_B}} \right|}^2} < {\varepsilon _B}{\theta _{th}}{{\left| {{H_E}} \right|}^2}} \right\}\\
&{= {\left( {1 + {{\left( {\frac{{{r_B}}}{{{r_F}}}} \right)}^\alpha }{\varepsilon _B}} \right)^{ - 1}}\left( {1 - {{\left( {{{\left( {\frac{{{r_F}}}{{{r_E}}}} \right)}^\alpha }{\theta _{th}} + {{\left( {\frac{{{r_B}}}{{{r_E}}}} \right)}^\alpha }{\varepsilon _B}{\theta _{th}} + 1} \right)}^{ - N}}} \right).}
\label{pou11IIinf}
\end{aligned}
\end{equation}
		
\section{Proof of Theorem 2}
\label{appendicesE}
\setcounter{equation}{0}
\renewcommand\theequation{E.\arabic{equation}}

1) \textit{ Derivation of ${P_{out,1}}$}

Based on (\ref{SOP03}) and $\Pr \left\{ {\tau \left( {{{\left| {{h_B}} \right|}^2}} \right) < 0} \right\} = \Pr \left\{ {{{\left| {{h_B}} \right|}^2} < {\alpha _B}} \right\}$, ${P_{out,1}}$ is rewritten as
\begin{equation}
	\begin{aligned}
	{P_{out,1}} 
	&= \underbrace {\Pr \left\{ {R_K^{\mathrm{I}} - {R_E} < {R_{th}},\left| {{S_{{\mathrm{II}}}}} \right| = 0,{{\left| {{h_B}} \right|}^2} < {\alpha _B}} \right\}}_{ \buildrel \Delta \over = P_{out,1}^1}\\
	& + \underbrace {\Pr \left\{ {R_K^{\mathrm{I}} - {R_E} < {R_{th}},\left| {{S_{{\mathrm{II}}}}} \right| = 0,{{\left| {{h_B}} \right|}^2} > {\alpha _B}} \right\}}_{ \buildrel \Delta \over = P_{out,1}^2},
	\label{pout1}
\end{aligned}
\end{equation}
where
${{\left| {{h_B}} \right|}^2} < {\alpha _B}$ denotes $U_B$ is reliability outage. Utilizing \cite[(3.383.10)]{Gradshteyn2007Book}, we obtain
\begin{equation}
	\begin{aligned}
		{P_{out,1}^1} &= \Pr \left\{ {{R_K^{\mathrm{I}} - {R_E} < {R_{th}}},{{\left| {{h_B}} \right|}^2} < {\alpha _B}} \right\}\\
		& = \Pr \left\{ {{{\log }_2}\left( {1 + \frac{{{\rho_F}{{\left| {{h_K}} \right|}^2}}}{{1 + {\rho_B}{{\left| {{h_B}} \right|}^2}}}} \right) - {{{\log }_2}\left( {1 + {\rho_F}{{\left| {{H_E}} \right|}^2} } \right)} < {R_{th}}, {{\left| {{h_B}} \right|}^2} < {\alpha _B}} \right\}\\
		&  = \Pr \left\{ {{{\left| {{h_K}} \right|}^2} < {\omega _0} \left( {{{\left| {{h_B}} \right|}^2},{{\left| {{H_E}} \right|}^2}} \right),{{\left| {{h_B}} \right|}^2} < {\alpha _B}} \right\}\\
		& = \int_0^\infty  {\int_0^{{\alpha _B}} {{F_{{{\left| {{h_K}} \right|}^2}}}\left( {{\omega _0} \left( {x,y} \right)} \right){f_{{{\left| {{h_B}} \right|}^2}}}\left( x \right)dx{f_{{{\left| {{H_E}} \right|}^2}}}\left( y \right)dy} }\\
		& {=1 - {e^{ - r_B^\alpha {\alpha _B}}} + \sum\limits_{i = 0}^K {\frac{{{\varphi _i}r_B^\alpha r_E^{N\alpha }}}{{\Gamma \left( N \right)}}{e^{ - ir_F^\alpha {\alpha _{th}}}}{\omega _1}\left( {i{\lambda _1},{\varepsilon _3},{\varepsilon _4}} \right)} ,}
	\label{pout11}
	\end{aligned}
\end{equation}
where
{
${{\varepsilon _3}} = ir_F^\alpha {\alpha _{th}}{\rho _B} + r_B^\alpha $
and
${\varepsilon _4} = ir_F^\alpha {\theta _{th}} + r_E^\alpha $.
}
Similarly, ${P_{out,1}^2}$ is expressed as
\begin{equation}
	\begin{aligned}
		{P_{out,1}^2} & =\Pr \left\{ {{{\log }_2}\left( {1 + \frac{{{\rho_F}{{\left| {{h_K}} \right|}^2}}}{{1 + {\rho_B}{{\left| {{h_B}} \right|}^2}}}} \right) - {{{\log }_2}\left( {1 + {\rho_F}{{\left| {{H_E}} \right|}^2} } \right)} } < {R_{th}},{{{\left| {{h_1}} \right|}^2} > \frac{{{\tau _B}}}{{{\rho_F}}},{{\left| {{h_B}} \right|}^2} > {\alpha _B}} \right\}\\
		& = \Pr \left\{ {{{\left| {{h_K}} \right|}^2} < {\omega _0} \left( {{{\left| {{h_B}} \right|}^2},{{\left| {{H_E}} \right|}^2}} \right),{{\left| {{h_1}} \right|}^2} > {\frac{{{\tau _B}}}{{{\rho_F}}}},{{\left| {{h_B}} \right|}^2} > {\alpha _B}} \right\}.
		\label{pout31}
	\end{aligned}
\end{equation}

Considering ${\left| {{h_1}} \right|^2} \le  \cdots  \le {\left| {{h_K}} \right|^2}$ and the relationship between ${{\omega _0}\left( {{{\left| {{h_B}} \right|}^2},{{\left| {{H_E}} \right|}^2}} \right)}$ and ${\frac{{{\tau _B}}}{{{\rho_F}}}}$, given in (\ref{relationship1}), two scenarios (${\varepsilon _B}{\varepsilon _{th}} < 1$ and ${\varepsilon _B}{\varepsilon _{th}} > 1$) are considered as follows.

(i) When ${\varepsilon _B}{\varepsilon _{th}} < 1$, we have ${\alpha _1} > 0 $. Due to ${\alpha _B} < {\alpha _2}$, based on (\ref{relationship1}), we obtain
	\begin{equation}
		\begin{aligned}
			P_{out,1}^{21} &=\Pr \left\{ {{{\left| {{h_K}} \right|}^2} < {\omega _0}\left( {{{\left| {{h_B}} \right|}^2},{{\left| {{H_E}} \right|}^2}} \right),} {{{\left| {{h_1}} \right|}^2} > \frac{{{\tau _B}}}{{{\rho_F}}},{{\left| {{h_B}} \right|}^2} > {\alpha _B}} \right\}\\
			& =\underbrace {\begin{array}{l}
					\Pr \left\{ {\frac{{{\tau _B}}}{{{\rho_F}}} < {{\left| {{h_1}} \right|}^2} < {{\left| {{h_K}} \right|}^2} < {\omega _0}\left( {{{\left| {{h_B}} \right|}^2},{{\left| {{H_E}} \right|}^2}} \right), {{\left| {{h_B}} \right|}^2} > {\alpha _B},{{\left| {{H_E}} \right|}^2} > {\alpha _1}} \right\}\end{array}}_{{I_1}}\\
			&+\underbrace {\begin{array}{l}
					\Pr \left\{ {\frac{{{\tau _B}}}{{{\rho_F}}} < {{\left| {{h_1}} \right|}^2} < {{\left| {{h_K}} \right|}^2} < {\omega _0}\left( {{{\left| {{h_B}} \right|}^2},{{\left| {{H_E}} \right|}^2}} \right), {\alpha _B} < {{\left| {{h_B}} \right|}^2} < {\alpha _2},{{\left| {{H_E}} \right|}^2}  < {\alpha _1}} \right\}
			\end{array}}_{{I_2}}.
			\label{pout321}
		\end{aligned}
\end{equation}

Based on (\ref{cdfh1K}), we obtain
	\begin{equation}\small
		\begin{aligned}
			{I_1} & = \int_{{\alpha _1}}^\infty  {\int_{{\alpha _B}}^\infty  {\left( {{F_{{{\left| {{h_1}} \right|}^2},{{\left| {{h_K}} \right|}^2}}}\left( {\frac{{{\tau _B}}}{{\rho_F}},{\omega _0} \left( {x,y} \right)} \right)} \right){f_{{{\left| {{h_B}} \right|}^2}}}\left( x \right)dx{f_{{{\left| {{H_E}} \right|}^2}}}\left( y \right)dy} } \\
			& {= \frac{{r_B^\alpha r_E^{N\alpha }}}{{\Gamma \left( N \right)}}\sum\limits_{n = 0}^{K - 2} {\left( {{\mu _1}{e^{\frac{{Kr_F^\alpha }}{{{\rho _F}}}}}{\omega _2}\left( {0,{\alpha _4},r_E^\alpha } \right) + {\mu _2}{e^{ - Kr_F^\alpha {\alpha _{th}}}}{\omega _2}\left( {{\eta _1},{\eta _2},{\eta _3}} \right) - {\mu _3}{e^{\frac{{Kr_F^\alpha  - {C_0}}}{{{\rho _F}}} - {C_0}{\alpha _{th}}}}{\omega _2}\left( {{\eta _4},{\eta _5},{\eta _6}} \right)} \right)} ,}
			\label{pout3I11}
		\end{aligned}
	\end{equation}
where
{
${\alpha _4} = \frac{{Kr_F^\alpha }}{{{\rho _F}{\alpha _B}}} + r_B^\alpha$,
${\eta _1} = Kr_F^\alpha {\rho _B}{\theta _{th}}$,
${\eta _2} = Kr_F^\alpha {\rho _B}{\alpha _{th}} + r_B^\alpha $,
${\eta _3} = Kr_F^\alpha {\theta _{th}} + r_E^\alpha$,
}
$ {\eta _4} = {C_0}{\rho _B}{\theta _{th}}$,
{
${\eta _5} = {C_0}{\rho _B}{\alpha _{th}} + \frac{{Kr_F^\alpha  - {C_0}}}{{{\rho _F}{\alpha _B}}} + r_B^\alpha $,
and
$ {\eta _6} = {C_0}{\theta _{th}} + r_E^\alpha $.
}
Similarly, we obtain
	\begin{equation}\small
		\begin{aligned}
			I_2 & = \int_0^{{\alpha _1}} {\int_{{\alpha _B}}^ {{\alpha _2}} {\left( {{F_{{{\left| {{h_1}} \right|}^2},{{\left| {{h_K}} \right|}^2}}}\left( {\frac{{{\tau _B}}}{{{\rho_F}}},{\omega _0} \left( {x,y} \right)} \right)} \right){f_{{{\left| {{h_B}} \right|}^2}}}\left( x \right)dx{f_{{{\left| {{H_E}} \right|}^2}}}\left( y \right)dy} }\\
			& {= \frac{{r_B^\alpha r_E^{N\alpha }}}{{\Gamma \left( N \right)}}\sum\limits_{n = 0}^{K - 2} {\left( {{\mu _1}{e^{\frac{{Kr_F^\alpha }}{{{\rho _F}}}}}{\omega _3}\left( {0,{\alpha _4},r_E^\alpha } \right) + {\mu _2}{e^{ - Kr_F^\alpha {\alpha _{th}}}}{\omega _3}\left( {{\eta _1},{\eta _2},{\eta _3}} \right) - {\mu _3}{e^{\frac{{Kr_F^\alpha  - {C_0}}}{{{\rho _F}}} - {C_0}{\alpha _{th}}}}{\omega _3}\left( {{\eta _4},{\eta _5},{\eta _6}} \right)} \right)}.}
			\label{pout3I2}
		\end{aligned}
	\end{equation}	

(ii) When ${\varepsilon _B}{\varepsilon _{th}} > 1$, we have ${\alpha _1} < 0 $, then $\Pr \left\{ {\frac{{{\tau _B}}}{{{\rho_F}}} < {\omega _0}\left( {{{\left| {{h_B}} \right|}^2},{{\left| {{H_E}} \right|}^2}} \right)} \right\} = 1$. Thus, we obtain
	\begin{equation}\small
		\begin{aligned}
			{P}_{out,1}^{22} &= \Pr \left\{ {\frac{{{\tau _B}}}{{{\rho_F}}} < {{\left| {{h_1}} \right|}^2} < {{\left| {{h_K}} \right|}^2} < {\omega _0}\left( {{{\left| {{h_B}} \right|}^2},{{\left| {{H_E}} \right|}^2}} \right),{{\left| {{h_B}} \right|}^2} > {\alpha _B}} \right\}\\
			&= \int_0^\infty  {\int_{{\alpha _B}}^\infty  {\left( {{F_{{{\left| {{h_1}} \right|}^2},{{\left| {{h_K}} \right|}^2}}}\left( {\frac{{{\tau _B}}}{{{\rho_F}}},{\omega _0}\left( {x,y} \right)} \right)} \right){f_{{{\left| {{h_B}} \right|}^2}}}\left( x \right)dx{f_{{{\left| {{H_E}} \right|}^2}}}\left( y \right)dy} }  \\
			& {= \frac{{r_B^\alpha r_E^{N\alpha }}}{{\Gamma \left( N \right)}}\sum\limits_{n = 0}^{K - 2} {\left( {{\mu _1}{e^{\frac{{Kr_F^\alpha }}{{{\rho _F}}}}}{\omega _4}\left( {0,{\alpha _4},r_E^\alpha } \right) + {\mu _2}{e^{ - Kr_F^\alpha {\alpha _{th}}}}{\omega _4}\left( {{\eta _1},{\eta _2},{\eta _3}} \right) - {\mu _3}{e^{\frac{{Kr_F^\alpha  - {C_0}}}{{{\rho _F}}} - {C_0}{\alpha _{th}}}}{\omega _4}\left( {{\eta _4},{\eta _5},{\eta _6}} \right)} \right)}.}
			\label{H23211}
		\end{aligned}
	\end{equation}

2) \textit{ {Derivation of ${P_{out,2}}$}}

When ${\left| {{{\cal S}_{\mathrm{II}}}} \right| = K}$, $U_B$'s signal must be decoded during the first stage of SIC and the signals of all the GF users will be decoded in the second stage of SIC. Utilizing the best-user scheduling scheme, $U_K$ will be selected. Then, SOP in this case is expressed as
\begin{equation}
	\begin{aligned}
		{P_{out,2}} &= \Pr \left\{ {{R_K^{\mathrm{II}} - {R_E} < {R_{th}}},\left| {{{\cal S}_{\mathrm{II}}}} \right| = K,{{\left| {{h_B}} \right|}^2} > {\alpha _B}} \right\}\\
		& =\Pr \left\{ {{{\log }_2}\left( {1 + {\rho_F}{{\left| {{h_K}} \right|}^2}} \right) -{{{\log }_2}\left( {1 + {\rho_F}{{\left| {{H_E}} \right|}^2} } \right)} < {R_{th}}, {{\left| {{h_K}} \right|}^2} < \frac{{{\tau _B}}}{{{\rho_F}}},{{\left| {{h_B}} \right|}^2} > {\alpha _B}} \right\}\\
		& = \Pr \left\{ {{{\left| {{h_K}} \right|}^2} < {\theta _{th}}{{\left| {{H_E}} \right|}^2} + {\alpha _{th}},{{\left| {{h_K}} \right|}^2} < \frac{{{\tau _B}}}{{{\rho_F}}},{{\left| {{h_B}} \right|}^2} > {\alpha _B}} \right\}\\
		&  = \Pr \left\{ {{{\left| {{h_K}} \right|}^2} < \min \left( {\frac{{{\tau _B}}}{{{\rho_F}}},{\theta _{th}}{{\left| {{H_E}} \right|}^2} + {\alpha _{th}}} \right),{{\left| {{h_B}} \right|}^2} > {\alpha _B}} \right\}\\
		&  = \Pr \left\{ {{{\left| {{h_K}} \right|}^2} < {\theta _{th}}{{\left| {{H_E}} \right|}^2} + {\alpha _{th}},{{\left| {{h_B}} \right|}^2} > \left( {{\rho_F}{{\left| {{H_E}} \right|}^2} + 1} \right){\varepsilon _1}} \right\}\\
		& + \Pr \left\{ {{{\left| {{h_K}} \right|}^2} < \frac{{{\tau _B}}}{{{\rho_F}}},{\alpha _B} < {{\left| {{h_B}} \right|}^2} < \left( {{\rho_F}{{\left| {{H_E}} \right|}^2} + 1} \right){\varepsilon _1}} \right\}.
		\label{pout12}
	\end{aligned}
\end{equation}
With some simple algebraic manipulations, we obtain
\begin{equation}
	\begin{aligned}
		{P_{out,2}}
		&  = \sum\limits_{i = 0}^K {{\varphi _i}\int_0^\infty  {{e^{ - i\left( {{\theta _{th}}y + {\alpha _{th}}} \right)}}{f_{{{\left| {{H_E}} \right|}^2}}}\left( y \right)\int_{\left( {{\rho_F}y + 1} \right){\varepsilon _1}}^\infty  {{f_{{{\left| {{h_B}} \right|}^2}}}\left( x \right)dxdy} } } \\
		& + \sum\limits_{i = 0}^K {{\varphi _i}\int_0^\infty  {{f_{{{\left| {{H_E}} \right|}^2}}}\left( y \right)\int_{{\alpha _B}}^{\left( {{\rho_F}y + 1} \right){\varepsilon _1}} {{e^{ - \frac{i}{{{\rho_F}}}\left( {\frac{x}{{{\alpha _B}}} - 1} \right)}}{f_{{{\left| {{h_B}} \right|}^2}}}\left( x \right)dxdy} } }  \\
		& {= \sum\limits_{i = 0}^K {\left( {\frac{{{\varphi _i}}}{{ir_F^\alpha  + r_B^\alpha {\rho _F}{\alpha _B}}}\left( {\frac{{ir_F^\alpha r_E^{N\alpha }{e^{ - \left( {ir_F^\alpha {\alpha _{th}} + r_B^\alpha {\varepsilon _1}} \right)}}}}{{{{\left( {ir_F^\alpha {\theta _{th}} + r_B^\alpha {\rho _F}{\varepsilon _1} + r_E^\alpha } \right)}^N}}} + {\rho _F}{\alpha _B}r_B^\alpha {e^{ - r_B^\alpha {\alpha _B}}}} \right)} \right)}.}
		\label{pout22}
	\end{aligned}
\end{equation}
			
3) \textit{ {Derivation of ${P_{out,3}}$}}
		
When both ${{\cal S}_{\mathrm{I}}}$ and ${{\cal S}_{\mathrm{II}}}$ are not empty, SOP is expressed as
\begin{equation}
	\begin{aligned}
			{P_{out,3}} &=\sum\limits_{k = 1}^{K - 2} \underbrace { {\Pr \left\{ {\max \left\{ {R_K^{\mathrm{I}},R_k^{{\mathrm{II}}}} \right\} - {R_E} < {R_{th}},\left| {{S_{{\mathrm{II}}}}} \right| = k} \right\}} }_{ \buildrel \Delta \over = P_{out,3}^k}\\
			&+ \underbrace {\Pr \left\{ {\max \left\{ {R_K^{\mathrm{I}},R_{K-1}^{{\mathrm{II}}}} \right\} - {R_E} < {R_{th}},\left| {{S_{{\mathrm{II}}}}} \right| = K - 1} \right\}}_{ \buildrel \Delta \over = P_{out,3}^{K - 1}}.	
	\end{aligned}
\end{equation}
Based on (\ref{ratek}), (\ref{pout1}), and (\ref{pout12}), we have
\begin{equation}
	\begin{aligned}	
		\Pr \left\{ {\left| {{{\cal S}_{\mathrm{II}}}} \right| = k,{{\left| {{h_B}} \right|}^2} > {\alpha _B}} \right\}
			&= \Pr \left\{ {{{\left| {{h_k}} \right|}^2} < \frac{{{\tau _B}}}{{{\rho_F}}} < {{\left| {{h_{k + 1}}} \right|}^2},{{\left| {{h_B}} \right|}^2} > {\alpha _B}} \right\},
		\label{pout4k2}
  \end{aligned}
\end{equation}
and
\begin{equation}
	\begin{aligned}
		\Pr \left\{ {\max \left\{ {R_K^{\mathrm{I}},R_k^{\mathrm{II}}} \right\} - {R_E} < {R_{th}}} \right\}
		 &= \Pr \left\{ {{{\left| {{h_k}} \right|}^2} < {\theta _{th}}{{\left| {{H_E}} \right|}^2} + {\alpha _{th}}, {{\left| {{h_K}} \right|}^2} < {\omega _0}\left( {{{\left| {{h_B}} \right|}^2},{{\left| {{H_E}} \right|}^2}} \right)} \right\},
		\label{pout4k1}
	\end{aligned}
\end{equation}
respectively.
Thus, $P_{out,3}^k$ is expressed as
\begin{equation}
	\begin{aligned}
		P_{out,3}^k  &=	\Pr \left\{ {{{\left| {{h_k}} \right|}^2} < \min \left( {\frac{{{\tau _B}}}{{{\rho_F}}},{\theta _{th}}{{\left| {{H_E}} \right|}^2} + {\alpha _{th}}} \right),{{\left| {{h_B}} \right|}^2} > {\alpha _B}} \right.,\\
		&\;\;\;\;\;\;\;\;\;\;\;\,\left. {\frac{{{\tau _B}}}{{{\rho_F}}} < {{\left| {{h_{k + 1}}} \right|}^2} < {{\left| {{h_K}} \right|}^2} < {\omega _0}\left( {{{\left| {{h_B}} \right|}^2},{{\left| {{H_E}} \right|}^2}} \right)} \right\}\\
			& = \underbrace {\begin{array}{l}
			\Pr \left\{ {{{\left| {{h_k}} \right|}^2} < \frac{{{\tau _B}}}{{{\rho_F}}} < {{\left| {{h_{k + 1}}} \right|}^2} < {{\left| {{h_K}} \right|}^2} < {\omega _0}\left( {{{\left| {{h_B}} \right|}^2},{{\left| {{H_E}} \right|}^2}} \right),} \right.\\
			\;\;\;\;\;\;\;\ \left. {{\alpha _B} < {{\left| {{h_B}} \right|}^2} < \left( {{\rho_F}{{\left| {{H_E}} \right|}^2} + 1} \right){\varepsilon _1}} \right\}\end{array}}_{{I_3}}\\
		   & + \underbrace {\begin{array}{l}
		     \Pr \left\{ {{{\left| {{h_k}} \right|}^2} < {\alpha _{th}} + {\theta _{th}}{{\left| {{H_E}} \right|}^2},{{\left| {{h_B}} \right|}^2} > \left( {{\rho_F}{{\left| {{H_E}} \right|}^2} + 1} \right){\varepsilon _1},} \right.\\
			 \;\;\;\;\;\;\ \left. {\frac{{{\tau _B}}}{{{\rho_F}}} < {{\left| {{h_{k + 1}}} \right|}^2} < {{\left| {{h_K}} \right|}^2} < {\omega _0}\left( {{{\left| {{h_B}} \right|}^2},{{\left| {{H_E}} \right|}^2}} \right)} \right\}
		\end{array}}_{{I_4}}.
		\label{pout43}
	\end{aligned}
\end{equation}
		
Based on (\ref{cdfthree}) and utilizing \cite[3.352.2]{Gradshteyn2007Book}, we obtain
\begin{equation}
	\begin{aligned}
		{I_3} &= \int_0^\infty  {{f_{{{\left| {{H_E}} \right|}^2}}}\left( t \right)dt\int_{{\alpha _B}}^{\left( {{\rho _F}t + 1} \right){\varepsilon _1}} {{F_{{{\left| {{h_k}} \right|}^2},{{\left| {{h_{k + 1}}} \right|}^2}{{\left| {{h_K}} \right|}^2}}}\left( {0,\frac{{{\tau _B}}}{{{\rho _F}}},\frac{{{\tau _B}}}{{{\rho _F}}},{\omega _0}\left( {s,t} \right)} \right){f_{{{\left| {{h_B}} \right|}^2}}}\left( s \right)ds} } \\
		&= \sum\limits_{n = 0}^{K - k - 2} {\sum\limits_{m = 0}^{k - 1} {\sum\limits_{i = 1}^6 {{\varsigma _i}\int_0^\infty  {\int_{{\alpha _B}}^{\left( {{\rho _F}t + 1} \right){\varepsilon _1}} {{e^{ - \left( {{B_i} + {C_i}} \right)\frac{{{\tau _B}}}{{{\rho _F}}} - {W_i}{\omega _0}\left( {s,t} \right)}}{f_{{{\left| {{h_B}} \right|}^2}}}\left( s \right){f_{{{\left| {{H_E}} \right|}^2}}}\left( t \right)dsdt} } } } }  \\
		&{= \frac{{r_B^\alpha r_E^{N\alpha }}}{{\Gamma \left( N \right)}}\sum\limits_{n = 0}^{K - k - 2} {\sum\limits_{m = 0}^{k - 1} {\sum\limits_{i = 1}^6 {{\varsigma _i}{e^{ - {\xi _1}}}\int_0^\infty  {t^{N - 1}}{{e^{ - {\xi _2}t}}\int_{{\alpha _B}}^{\left( {{\rho_F}t + 1} \right){\varepsilon _1}} {{e^{ - \left( {{u_1}t + {\xi _3}} \right)s}}dsdt} } } } } } \\
		&  {= r_B^\alpha r_E^{N\alpha }\sum\limits_{n = 0}^{K - k - 2} {\sum\limits_{m = 0}^{k - 1} {\left( {\sum\limits_{i = 1}^4 {\frac{{{\varsigma _i}{e^{ - {\xi _1}}}}}{{\Gamma \left( N \right)}}{\Delta _1}}  + \sum\limits_{i = 5}^6 {{\varsigma _i}{e^{ - {\xi _1}}}{\Delta _2}} } \right)} } ,}
	\label{II3}
	\end{aligned}
\end{equation}
where
${\Delta _1}  = \frac{{\xi _3^{N - 1}\Gamma \left( {1 - N,{\xi _3}{\alpha _B} + \frac{{{\xi _2}{\xi _3}}}{{{u_1}}}} \right)}}{{u_1^N}}{e^{\frac{{{\xi _2}{\xi _3}}}{{{u_1}}}}} - \frac{{{e^{ - {\xi _3}{\varepsilon _1}}}{\omega _5}\left( {{u_1},{\xi _3},{v_1},{l_1}} \right)}}{{\Gamma \left( N \right)}}$,
${\Delta _2} = \frac{{{e^{ - {\xi _3}{\alpha _B}}}}}{{\xi _2^N{\xi _3}}} - \frac{{{e^{ - {\xi _3}{\varepsilon _1}}}}}{{{\xi _3}{{\left( {{\rho_F}{\xi _3}{\varepsilon _1} + {\xi _2}} \right)}^N}}}$,
${\xi _1} = {W_i}{\alpha _{th}} - \frac{{{B_i} + {C_i}}}{{\rho_F}}$,
{
${\xi _2} = {W_i}{\theta _{th}} + r_E^\alpha $,
${\xi _3} = {W_i}{\rho_B}{\alpha _{th}} + \frac{{{B_i} + {C_i}}}{{P_{F}{\alpha _B}}} + r_B^\alpha $,
}
${u_1} = {W_i}{\rho_B}{\theta _{th}}$,
${v_1} = {u_1}{\rho_F}{\varepsilon _1}$,
${l_1} = {u_1}{\varepsilon _1} + {\rho_F}{\xi _3}{\varepsilon _1} + {\xi _2}$, and
${\omega _5} \left( {a,b,c,f} \right) = \int_0^\infty  {\frac{1}{{ax + b}}{e^{ - \left( {c{x^2} + fx} \right)}}dx}$.
By utilizing \cite[(10), (11)]{Adamchik1990}, \cite[(6.2.8)]{Springer1979}, and \cite[(2.3)]{RGBHFH1972}, we obtain
\begin{equation}
	\begin{aligned}
		{\omega _5} \left( {a,b,c,f} \right) 
		& = \frac{1}{b}\int_0^\infty  {H_{0,1}^{1,0}\left[ {fx\left| {_{\left( {0,1} \right)}^ - } \right.} \right]H_{1,1}^{1,1}\left[ {\frac{a}{b}x\left| {_{\left( {0,1} \right)}^{\left( {0,1} \right)}} \right.} \right]H_{0,1}^{1,0}\left[ {c{x^2}\left| {_{\left( {0,1} \right)}^ - } \right.} \right]dx} \\
		& = \frac{{{f^{N - 1}}}}{{b}}H_{1,0:1,1:0,1}^{1,0:1,1:1,0}\left[ {_{\quad  - }^{\left( {0;1,2} \right)}\left| {_{\left( {0,1} \right)}^{\left( {0,1} \right)}\left| {_{\left( {0,1} \right)}^ - \left| {\frac{a}{{bf}},\frac{c}{{{f^2}}}} \right.} \right.} \right.} \right].
		\label{fai3fun}
	\end{aligned}
\end{equation}
With the same method, we obtain
\begin{equation}
	\begin{aligned}
		{I_4} &= \int_0^\infty  {{f_{{{\left| {{H_E}} \right|}^2}}}\left( t \right)dt\int_{\left( {{\rho _F}t + 1} \right){\varepsilon _1}}^\infty  {{f_{{{\left| {{h_B}} \right|}^2}}}\left( s \right){F_{{{\left| {{h_k}} \right|}^2},{{\left| {{h_{k + 1}}} \right|}^2}{{\left| {{h_K}} \right|}^2}}}\left( {0,{\alpha _{th}} + {\theta _{th}}t,\frac{{{\tau _B}}}{{{\rho _F}}},{\omega _0}\left( {s,t} \right)} \right)ds} } \\
		&= \sum\limits_{n = 0}^{K - k - 2} {\sum\limits_{m = 0}^{k - 1} {\sum\limits_{i = 1}^6 {{\varsigma _i}} \int_0^\infty  {\int_{\left( {{\rho _F}t + 1} \right){\varepsilon _1}}^\infty  {{e^{ - {B_i}\left( {{\alpha _{th}} + {\theta _{th}}t} \right) - {C_i}\frac{{{\tau _B}}}{{{\rho _F}}} - {W_i}{\omega _0}\left( {s,t} \right)}}{f_{{{\left| {{h_B}} \right|}^2}}}\left( s \right){f_{{{\left| {{H_E}} \right|}^2}}}\left( t \right)dsdt} } } } 	 \\
		&  {= r_B^\alpha r_E^{N\alpha }\sum\limits_{n = 0}^{K - k - 2} {\sum\limits_{m = 0}^{k - 1} {\left( {\sum\limits_{i = 1}^4 {\frac{{{\varsigma _i}{e^{ - {\xi _4}}}}}{{\Gamma \left( N \right)}}{\Delta _3}}  + \sum\limits_{i = 5}^6 {{\varsigma _i}{e^{ - {\xi _4}}}{\Delta _4}} } \right)} } ,}
    \label{II4}
	\end{aligned}
\end{equation}
where
${\xi _4} = \left( {{B_i} + {W_i}} \right){\alpha _{th}} - \frac{{{C_i}}}{{\rho_F}}$,
${\Delta _3} = \frac{{{e^{ - {\xi _6}{\varepsilon _1}}}{\omega _6}\left( {{u_1},{\xi _6},{v_2},{l_2}} \right)}}{{\Gamma \left( N \right)}}$,
${\Delta _4} = \frac{{{e^{ - {\xi _6}{\varepsilon _1}}}}}{{{\xi _6}{{\left( {{\rho_F}{\xi _6}{\varepsilon _1} + {\xi _5}} \right)}^N}}}$,
${v_2} = {u_1}{\rho_F}{\varepsilon _1}$,
${l_2} = {u_1}{\varepsilon _1} + {\xi _6}{\rho_F}{\varepsilon _1} + {\xi _5}$,
{
${\xi _5} = \left( {{B_i} + {W_i}} \right){\theta _{th}} + r_E^\alpha $,
and
${\xi _6} = {W_i}{\alpha _{th}}{\rho_B} + \frac{{{C_i}}}{{P_{F}{\alpha _B}}} + r_B^\alpha $.
}

Similar to (\ref{pout4k2}) and (\ref{pout4k1}), we obtain
\begin{equation}
		\begin{aligned}
	\Pr \left\{ {\left| {{{\cal S}_{\mathrm{II}}}} \right| = K - 1, {{\left| {{h_B}} \right|}^2} > {\alpha _B}} \right\} &= \Pr \left\{ {{{\left| {{h_{K - 1}}} \right|}^2} < \frac{{{\tau _B}}}{{{\rho_F}}} < {{\left| {{h_{K}}} \right|}^2},{{\left| {{h_B}} \right|}^2} > {\alpha _B}} \right\}
	\label{pout4K12}
		\end{aligned}
\end{equation}
and
\begin{equation}
	\begin{aligned}
		&\Pr \left\{\max \left\{ {R_K^{\mathrm{I}},R_{K-1}^{{\mathrm{II}}}} \right\}  - {R_E} < {R_{th}} \right\} \\
&= \Pr \left\{ {{{\left| {{h_{K - 1}}} \right|}^2} < {\theta _{th}}{{\left| {{H_E}} \right|}^2} + {\alpha _{th}}}, {{{\left| {{h_K}} \right|}^2} < {\omega _0}\left( {{{\left| {{h_B}} \right|}^2},{{\left| {{H_E}} \right|}^2}} \right)} \right\}.
		\label{pout4K11}
	\end{aligned}
\end{equation}
Then, $P_{out,3}^{K - 1}$ is obtained as
\begin{equation}
	\begin{aligned}
		P_{out,3}^{K - 1} &= \Pr \left\{ {R_{K - 1}^s < {R_{th}},\left| {{{\cal S}_{\mathrm{II}}}} \right| = K - 1,{{\left| {{h_B}} \right|}^2} > {\alpha _B}} \right\}\\
		& = \Pr \left\{ {{{\left| {{h_{K - 1}}} \right|}^2} < \frac{{{\tau _B}}}{{{\rho_F}}} < {{\left| {{h_K}} \right|}^2} < {\omega _0}\left( {{{\left| {{h_B}} \right|}^2},{{\left| {{H_E}} \right|}^2}} \right),{\alpha _B} < {{\left| {{h_B}} \right|}^2} < \left( {{\rho_F}{{\left| {{H_E}} \right|}^2} + 1} \right){\varepsilon _1}} \right\}\\
		& + \Pr \left\{ {{{\left| {{h_{K - 1}}} \right|}^2} < {\theta _{th}}{{\left| {{H_E}} \right|}^2} + {\alpha _{th}},\frac{{{\tau _B}}}{{{\rho_F}}} < {{\left| {{h_K}} \right|}^2} < {\omega _0}\left( {{{\left| {{h_B}} \right|}^2},{{\left| {{H_E}} \right|}^2}} \right),{{\left| {{h_B}} \right|}^2} > \left( {{\rho_F}{{\left| {{H_E}} \right|}^2} + 1} \right){\varepsilon _1}} \right\}\\
		& = \int_0^\infty  {{f_{{{\left| {{H_E}} \right|}^2}}}\left( t \right)dt\int_{{\alpha _B}}^{\left( {{\rho_F}t + 1} \right){\varepsilon _1}} {{f_{{{\left| {{h_B}} \right|}^2}}}\left( s \right){F_{{{\left| {{h_{K - 1}}} \right|}^2},{{\left| {{h_K}} \right|}^2}}}\left( {0,\frac{{{\tau _B}}}{{{\rho_F}}},\frac{{{\tau _B}}}{{{\rho_F}}},{\omega _0}\left( s, t \right)} \right)ds} }  \\
		& + \int_0^\infty  {{f_{{{\left| {{H_E}} \right|}^2}}}\left( t \right)dt\int_{\left( {{\rho_F}t + 1} \right){\varepsilon _1}}^\infty  {{f_{{{\left| {{h_B}} \right|}^2}}}\left( s \right){F_{{{\left| {{h_{K - 1}}} \right|}^2},{{\left| {{h_K}} \right|}^2}}}\left( {0,{\theta _{th}}t + {\alpha _{th}},\frac{{{\tau _B}}}{{{\rho_F}}},{\omega _0}\left( {s,t} \right)} \right)ds} }  \\
        &{=r_B^\alpha r_E^{N\alpha }\sum\limits_{n = 0}^{K - 2} {\frac{{{\mu _0}}}{{{C_0}}}\left( {\sum\limits_{j = 1}^2 {\frac{{{{\left( { - 1} \right)}^{j + 1}}}}{{\Gamma \left( N \right)}}} \left( {{e^{ - {\zeta _1}}}{\Delta _5} + {e^{ - {\zeta _4}}}{\Delta _7}} \right) + \sum\limits_{j = 3}^4 {{{\left( { - 1} \right)}^{j + 1}}\left( {{e^{ - {\zeta _1}}}{\Delta _6} + {e^{ - {\zeta _4}}}{\Delta _8}} \right)} } \right)},}
		\label{sop4K1}
	\end{aligned}
\end{equation}
where
{
${\zeta _1} = {q_j}{\alpha _{th}} - \frac{{{b_j} + {c_j}}}{{\rho_F}}$,
{
${\zeta _2} = {q_j}{\theta _{th}} + r_E^\alpha$,
${\zeta _3} = {q_j}{\alpha _{th}}{\rho_B} + \frac{{{b_j} + {c_j}}}{{P_{F}{\alpha _B}}} + r_B^\alpha$,
}
${\Delta _5} = \frac{{\zeta _3^{N - 1}}}{{u_2^N}}{e^{\frac{{{\zeta _2}{\zeta _3}}}{{{u_2}}}}}\Gamma \left( {1 - N,{\zeta _3}{\alpha _B} + \frac{{{\zeta _2}{\zeta _3}}}{{{u_2}}}} \right) \\
- \frac{{{e^{ - {\xi _3}{\varepsilon _1}}}{\omega _5}\left( {{u_2},{\zeta _3},{v_3},{l_3}} \right)}}{{\Gamma \left( N \right)}}$,
${u_2} = {q_j}{\rho_B}{\theta _{th}}$,
${v_3} = {u_2}{\rho_F}{\varepsilon _1}$,
${l_3} = {u_2}{\varepsilon _1} + {\zeta _3}{\rho_F}{\varepsilon _1} + {\zeta _2}$,
${\Delta _6} = \frac{{{e^{ - {\zeta _3}{\alpha _B}}}}}{{\zeta _2^N{\zeta _3}}} - \frac{{{e^{ - {\zeta _3}{\varepsilon _1}}}}}{{{\zeta _3}{{\left( {{P_F}{\zeta _3}{\varepsilon _1} + {\zeta _2}} \right)}^N}}}$,
${\zeta _4} = \left( {{b_j} + {q_j}} \right){\alpha _{th}} - \frac{{{c_i}}}{{\rho_F}}$,
{
${\zeta _5} = \left( {{b_j} + {q_j}} \right){\theta _{th}} + r_E^\alpha$,
${\zeta _6} = {q_i}{\alpha _{th}}{\rho_B} + \frac{{{c_j}}}{{P_{F}{\alpha _B}}} + r_B^\alpha$,
}
${\Delta_7} = \frac{{{e^{ - {\zeta _6}{\varepsilon _1}}}}{\omega _5}\left( {{u_2},{\zeta _6},{v_4},{l_4}} \right)}{{\Gamma \left( N \right)}}$,
${v_4} = {u_2}{\rho_F}{\varepsilon _1}$,
${l_4} = {u_2}{\varepsilon _1} +{\zeta _6} {\rho_F}{\varepsilon _1} + {\zeta _5}$,
and
${\Delta _8} =\frac{{{e^{ - {\zeta _6}{\varepsilon _1}}}}}{{{\zeta _6}{{\left( {{\rho_F}{\zeta _6}{\varepsilon _1} + {\zeta _5}} \right)}^N}}}$.
}		
\section{Proof of Corollary 4}
\label{appendicesF}
\setcounter{equation}{0}
\renewcommand\theequation{F.\arabic{equation}}

When ${\rho_F} = {\rho_B} \to \infty $, we have ${\alpha _B} \to 0,{\alpha _{th}} \to 0$.
One can obtain
{
$P_{out,1}^{1,\infty } \approx 0$
}
due  to $\Pr \left\{ {{{\left| {{h_B}} \right|}^2} < {\alpha _B}} \right\} \approx 0$.
Based on (\ref{pout31}) and  (\ref{pout321}), ${P_{out,1}^2}$ is approximated as
\begin{equation}\small
	\begin{aligned}
		P_{out,1}^{2,\infty } &\approx \Pr \left\{ {\frac{{{{\left| {{h_B}} \right|}^2}}}{{{\varepsilon _B}}} < {{\left| {{h_1}} \right|}^2} < {{\left| {{h_K}} \right|}^2} < {\omega _0}\left( {{{\left| {{h_B}} \right|}^2},{{\left| {{H_E}} \right|}^2}} \right)} \right\}\\
		&= \int_0^\infty  {\int_0^\infty  {\left( {{F_{{{\left| {{h_1}} \right|}^2},{{\left| {{h_K}} \right|}^2}}}\left( {\frac{x}{{{\varepsilon _B}}},{\omega _0}\left( {x,y} \right)} \right)} \right){f_{{{\left| {{h_B}} \right|}^2}}}\left( x \right)dx{f_{{{\left| {{H_E}} \right|}^2}}}\left( y \right)dy} } \\
		&{\mathop  \approx \limits^{\left( g \right)} \sum\limits_{n = 0}^{K - 2} {\frac{{{\varepsilon _B}{\mu _1}}}{{K{{\left( {\frac{{{r_F}}}{{{r_B}}}} \right)}^\alpha } + {\varepsilon _B}}}},}
		\label{poutAsy31}
	\end{aligned}
\end{equation}
where $(g)$ holds with the same method as $(f)$.

Based on (\ref{pout12}) and  (\ref{pout22}), ${P_{out,2}}$ is approximated as
\begin{equation}
	\begin{aligned}
		P_{out,2}^\infty  &\approx \Pr \left\{ {{{\left| {{h_K}} \right|}^2} < {\theta _{th}}{{\left| {{H_E}} \right|}^2},{{\left| {{h_B}} \right|}^2} > {\varepsilon _B}{{\left| {{H_E}} \right|}^2}} \right\}\\
		&+ \Pr \left\{ {{{\left| {{h_K}} \right|}^2} < \frac{{{{\left| {{h_B}} \right|}^2}}}{{{\varepsilon _B}}},{{\left| {{h_B}} \right|}^2} < {\varepsilon _B}{{\left| {{H_E}} \right|}^2}} \right\}\\
		&= \frac{{r_B^\alpha r_E^{N\alpha }}}{{\Gamma \left( N \right)}}\sum\limits_{i = 0}^K {{\varphi _i}\int_0^\infty  {{y^{N - 1}}{e^{ - \left( {ir_F^\alpha {\theta _{th}} + r_E^\alpha } \right)y}}\int_{{\varepsilon _B}{\theta _{th}}y}^\infty  {{e^{ - r_B^\alpha x}}dxdy} } } \\
		&+ \frac{{r_B^\alpha r_E^{N\alpha }}}{{\Gamma \left( N \right)}}\sum\limits_{i = 0}^K {{\varphi _i}\int_0^\infty  {{y^{N - 1}}{e^{ - r_E^\alpha y}}\int_0^{{\varepsilon _B}{\theta _{th}}y} {{e^{ - \left( {\frac{{ir_F^\alpha }}{{{\varepsilon _B}}} + r_E^\alpha } \right)x}}dxdy} } } \\
	    & {=\sum\limits_{i = 0}^K {\frac{{{\varphi _i}{\varepsilon _B}}}{{i{{\left( {\frac{{{r_F}}}{{{r_B}}}} \right)}^\alpha } + {\varepsilon _B}}}}  + \sum\limits_{i = 0}^K {\frac{{i{\varphi _i}{{\left( {i{\chi _1} + {\chi _2}} \right)}^{ - N}}}}{{i + {\varepsilon _B}{{\left( {\frac{{{r_B}}}{{{r_F}}}} \right)}^\alpha }}}},}
		\label{poutAsy21}
	\end{aligned}
\end{equation}
{where
${\chi _1} = {\theta _{th}}{\left( {\frac{{{r_F}}}{{{r_E}}}} \right)^\alpha }$
and
${\chi _2} = {\varepsilon _B}{\theta _{th}}{\left( {\frac{{{r_B}}}{{{r_E}}}} \right)^\alpha } + 1$.}

Based on (\ref{pout43}), we obtain
\begin{equation}
	\begin{aligned}
		I_3^\infty  &\approx \Pr \left\{ {{{\left| {{h_k}} \right|}^2} < \frac{{{{\left| {{h_B}} \right|}^2}}}{{{\varepsilon _B}}} < {{\left| {{h_{k + 1}}} \right|}^2} < {{\left| {{h_K}} \right|}^2} < {\omega _0}\left( {{{\left| {{h_B}} \right|}^2},{{\left| {{H_E}} \right|}^2}} \right),{{\left| {{h_B}} \right|}^2} < {\varepsilon _B}{\theta _{th}}{{\left| {{H_E}} \right|}^2}} \right\}\\
		&= \sum\limits_{n = 0}^{K - k - 2} {\sum\limits_{m = 0}^{k - 1} {\sum\limits_{i = 1}^6 {{\varsigma _i}\int_0^\infty  {\int_0^{{\varepsilon _B}{\theta _{th}}t} {{e^{ - \left( {{B_i} + {C_i}} \right)\frac{s}{{{\varepsilon _B}}} - {W_i}{\omega _0}\left( {s,t} \right)}}{f_{{{\left| {{h_B}} \right|}^2}}}\left( s \right){f_{{{\left| {{H_E}} \right|}^2}}}\left( t \right)dsdt} } } } } \\
		&{\mathop  = \limits^{\left( h \right)} \sum\limits_{n = 0}^{K - k - 2} {\sum\limits_{m = 0}^{k - 1} {\sum\limits_{i = 5}^6 {\frac{{{\varsigma _i}{\varepsilon _B}\left( {1 - {\chi _3}} \right)}}{{\left( {K + {\varpi _i}} \right){{\left( {\frac{{{r_F}}}{{{r_B}}}} \right)}^\alpha } + {\varepsilon _B}}}} } },}
		\label{pout3I1}
	\end{aligned}
\end{equation}
and
\begin{equation}
		\begin{aligned}
			I_4^\infty  &\approx \Pr \left\{ {{{\left| {{h_k}} \right|}^2} < {\theta _{th}}{{\left| {{H_E}} \right|}^2},\frac{{{{\left| {{h_B}} \right|}^2}}}{{{\varepsilon _B}}} < {{\left| {{h_{k + 1}}} \right|}^2} < {{\left| {{h_K}} \right|}^2} < {\omega _0}\left( {{{\left| {{h_B}} \right|}^2},{{\left| {{H_E}} \right|}^2}} \right),{{\left| {{h_B}} \right|}^2} > {\varepsilon _B}{\theta _{th}}{{\left| {{H_E}} \right|}^2}} \right\}\\
			&= \sum\limits_{n = 0}^{K - k - 2} {\sum\limits_{m = 0}^{k - 1} {\sum\limits_{i = 1}^6 {{\varsigma _i}} \int_0^\infty  {\int_{{\varepsilon _B}{\theta _{th}}t}^\infty  {{e^{ - {B_i}{\theta _{th}}t - {C_i}\frac{s}{{{\varepsilon _B}}} - {W_i}\left( {{\omega _0}\left( {s,t} \right)} \right)}}{f_{{{\left| {{h_B}} \right|}^2}}}\left( s \right){f_{{{\left| {{H_E}} \right|}^2}}}\left( t \right)dsdt} } } }  \\
			&{\mathop  = \limits^{\left( i \right)} \sum\limits_{n = 0}^{K - k - 2} {\sum\limits_{m = 0}^{k - 1} {\sum\limits_{i = 5}^6 {\frac{{{\varsigma _i}{\varepsilon _B}{\chi _3}}}{{\left( {\left( {K - k} \right){{\left( {\frac{{{r_F}}}{{{r_B}}}} \right)}^\alpha } + {\varepsilon _B}} \right)}}} ,} }}
			\label{pout3I4}
		\end{aligned}	
	\end{equation}
where
{${\chi _3} = {\left( {\left( {K + {\varpi _i}} \right){\chi _1} + {\chi _2}} \right)^{ - N}}$,}
$(h)$ and $(i)$ hold with the same method as $(f)$,
and
{
$\varpi  = \left[ {0,0,n + 2,1,m + 1 - k, - k} \right]$.
}

Thus, $P_{out,3}^k$, when ${\rho_F} = {\rho_B} \to \infty $, is approximated as
\begin{equation}
	P_{out,3}^{k,\infty } = \sum\limits_{k = 1}^{K - 2} {\left( {I_3^\infty  + I_4^\infty } \right)}.
	\label{H232}
\end{equation}
With the same method and based on (\ref{sop4K1}), $P_{out,3}^{K - 1} $ is approximated as
\begin{equation}
	\begin{aligned}
		P_{out,3}^{K - 1,\infty } &\approx \Pr \left\{ {R_{K - 1}^s < {R_{th}},\left| {{S_{\mathrm{II}}}} \right| = K - 1} \right\}\\
		&= \Pr \left\{ {{{\left| {{h_{K - 1}}} \right|}^2} < \frac{{{{\left| {{h_B}} \right|}^2}}}{{{\varepsilon _B}}} < {{\left| {{h_K}} \right|}^2} < {\omega _0}\left( {{{\left| {{h_B}} \right|}^2},{{\left| {{H_E}} \right|}^2}} \right),{{\left| {{h_B}} \right|}^2} < {\varepsilon _B}{\theta _{th}}{{\left| {{H_E}} \right|}^2}} \right\}\\
		&+ \Pr \left\{ {{{\left| {{h_{K - 1}}} \right|}^2} < {\theta _{th}}{{\left| {{H_E}} \right|}^2},\frac{{{{\left| {{h_B}} \right|}^2}}}{{{\varepsilon _B}}} < {{\left| {{h_K}} \right|}^2} < {\omega _0}\left( {{{\left| {{h_B}} \right|}^2},{{\left| {{H_E}} \right|}^2}} \right),{{\left| {{h_B}} \right|}^2} > {\varepsilon _B}{\theta _{th}}{{\left| {{H_E}} \right|}^2}} \right\}\\
		&= \int_0^\infty  {{f_{{{\left| {{H_E}} \right|}^2}}}\left( t \right)dt\int_0^{{\varepsilon _B}{\theta _{th}}t} {{f_{{{\left| {{h_B}} \right|}^2}}}\left( s \right){F_{{{\left| {{h_{K - 1}}} \right|}^2},{{\left| {{h_K}} \right|}^2}}}\left( {0,\frac{s}{{{\varepsilon _B}}},\frac{s}{{{\varepsilon _B}}},{\omega _0}\left( {s,t} \right)} \right)ds} } \\
		&+ \int_0^\infty  {{f_{{{\left| {{H_E}} \right|}^2}}}\left( t \right)dt\int_{{\varepsilon _B}{\theta _{th}}t}^\infty  {{f_{{{\left| {{h_B}} \right|}^2}}}\left( s \right){F_{{{\left| {{h_{K - 1}}} \right|}^2},{{\left| {{h_K}} \right|}^2}}}\left( {0,{\theta _{th}}t,\frac{s}{{{\varepsilon _B}}},{\omega _0}\left( {s,t} \right)} \right)ds} } \\
		&{\mathop = \limits^{\left( j \right)} \sum\limits_{n = 0}^{K - 2} {{\mu _4}\sum\limits_{j = 3}^4 {{{\left( { - 1} \right)}^{j + 1}}{\varepsilon _B}\left( {\frac{{1 - {\chi _4}}}{{{\varpi _j}r_B^{ - \alpha } + r_F^{ - \alpha }{\varepsilon _B}}}} \right.} {\mkern 1mu}  + \left. {\frac{{{\chi _4}}}{{r_B^{ - \alpha } + r_F^{ - \alpha }{\varepsilon _B}}}} \right)},}
		\label{H232}
	\end{aligned}
\end{equation}
where
${\mu _4} = \frac{{K!{{\left( { - 1} \right)}^n}\left( {_{\;\;n}^{K - 2}} \right)}}{{\left( {K - 2} \right)!\left( {n + 1} \right)}}$,
{${\chi _4} = {\left( {{\varpi _j}{\chi _1} + {\chi _2}} \right)^{ - N}}$,}
and
$(j)$ holds with the same method as $(f)$.

\end{appendices}	
	

\end{document}